\RequirePackage{fix-cm}
\documentclass[smallextended]{svjour3}       
\smartqed  
\usepackage{amssymb,latexsym,amsfonts,amsmath,mathrsfs}
\usepackage{graphicx}
\usepackage{geometry}
\usepackage{dsfont}
\usepackage{setspace}
\usepackage{tabularx}
\usepackage{placeins}
\usepackage{hyperref}
\usepackage{subfigure}
\newcolumntype{b}{X}
\newcolumntype{s}{>{\hsize=.5\hsize}X}
\newcolumntype{t}{>{\hsize=.2\hsize}X}
\hypersetup{colorlinks=true, linkcolor=blue, citecolor=red}

\smartqed

\usepackage{appendix}
\usepackage{ctable}
\usepackage{threeparttable}

\newcommand{\D}{\displaystyle}

\newcommand{\de}{\Delta}

\newcommand{\br}[1]{\left\{#1\right\}}

\newcommand{\xqedhere}[2]{%
  \rlap{\hbox to#1{\hfil\llap{\ensuremath{#2}}}}}

\allowdisplaybreaks

\begin{document}


\title{Infection severity across scales in multi-strain immuno-epidemiological Dengue model structured by host antibody level}

\author{Hayriye Gulbudak \and Cameron J. Browne}

\institute{H. Gulbudak \at Mathematics Department, University of Louisiana at Lafayette, Lafayette, LA \\
            \email{hayriye.gulbudak@louisiana.edu}    \and  C.J. Browne \at Mathematics Department, University of Louisiana at Lafayette, Lafayette, LA     
  }

\maketitle

\begin{abstract}
Infection by distinct Dengue virus serotypes and host immunity are intricately linked.  In particular, certain levels of cross-reactive antibodies in the host may actually enhance infection severity leading to Dengue hemorrhagic fever (DHF).  The coupled immunological and epidemiological dynamics of Dengue calls for a multi-scale modeling approach.  In this work, we formulate a within-host model which mechanistically recapitulates characteristics of antibody dependent enhancement (ADE) in Dengue infection.  The within-host scale is then linked to epidemiological spread by a vector-host partial differential equation model structured by host antibody level.  The coupling allows for dynamic population-wide antibody levels to be tracked through primary and secondary infections by distinct Dengue strains, along with waning of cross-protective immunity after primary infection.  Analysis of both the within-host and between-host systems are conducted.  Stability results in the epidemic model are formulated via basic and invasion reproduction numbers as a function of immunological variables.  Additionally, we develop numerical methods in order to simulate the multi-scale model and assess the influence of parameters on disease spread and DHF prevalence in the population. 
\end{abstract}


\section{Introduction}
The global burden of Dengue infection has rapidly increased in recent years, with about 400 million dengue infections occurring every year.  Research has delved into the complexities of this mosquito-transmitted disease.  The intricate relationship between host immune response, pathogenesis, viral diversity and epidemiology has received particular attention.  While the immune response ultimately clears Dengue virus from an infected host and provides strain-specific immunity upon recovery, certain levels of cross-reactive antibodies (reacting to multiple Dengue strains) may actually enhance severity of a subsequent (or even a primary) infection manifesting in Dengue hemorrhagic fever (DHF).  Determining the impact of host immune response, distinct viral strains, and population-wide antibody levels on Dengue incidence calls for a multi-scale approach.  The problem is critical for control strategies against Dengue, highlighted by recent debate over vaccination, which by boosting antibody responses, may actually increase DHF prevalence in certain groups.  Herein this paper, we develop a mathematical model linking within-host and between-host scales through host antibody level in order to describe the connection between immunity and Dengue infection dynamics across both scales.     

Dengue fever is caused by four antigenically related but distinct serotypes (DENV-1 to DENV-4). Infection by one serotype confers life-long immunity to that serotype and a period of temporary cross-immunity to other serotypes.  Sequential infection increases the risk of developing severe dengue, due to a process described as antibody-dependent enhancement (ADE), where the pre-existing antibodies to previous dengue infection enhances the new infection \cite{dejnirattisai2010cross}.  The mechanisms behind ADE and consequences on Dengue epidemiology are not completely understood.  However, recent research has found evidence that a certain intermediate window of pre-existent antibody titer in the host population is associated with risk of DHF \cite{katzelnick2017antibody,salje2018reconstruction}.

Previous modeling efforts have studied Dengue infection on either the within-host or between-host scales \cite{ben2015minimal,wearing2006ecological,ferguson2016benefits}.  Several mathematical models investigate multi-strain epidemiological dynamics with potential secondary infection and ADE due to partial or temporary cross-reactivity \cite{aguiar2017mathematical,ferguson1999effect,reich2013interactions,nikin2018unraveling}.  Resulting bi-stable, oscillatory or chaotic dynamics may explain large fluctuations in disease incidence observed in Dengue epidemics \cite{ferguson1999effect,cummings2005dynamic,aguiar2008epidemiology}.  
In all of these epidemic models, ADE is incorporated through parameters associated with secondary infection, however the recent evidence points to pre-existent antibody levels as the determinant of infection severity.  Therefore for a more precise formulation of ADE on the epidemiological scale, a model should track dynamic immune status in host population, which is one of the goals of this paper.  On the within-host scale, several models have considered the phenomenon of ADE \cite{ben2015minimal,gujarati2014virus,nikin2015role}, although explicit dependence of infection severity on \emph{pre-existent} cross-reactive antibody concentration has not been produced.

In addition to the correlation of pre-existent antibody level to risk of developing severe infection, the virus and antibody dynamics within infected hosts determine their inherent infectivity and recovery rates.   With this in mind, multi-scale models linking within-host and between-host dynamics emerge as an appropriate tool for a unified model of Dengue.  In particular, recently studied ``nested immuno-epi'' models offer a useful framework, where a partial differential equation (PDE) epidemiological model includes a structuring variable that also appears on the virus-immune response scale \cite{gilchrist2002modeling,gandolfi2015epidemic,tuncer2016structural,gulbudak2017vector}.  In most nested models, transmission and recovery rates of infected hosts are structured by infection-age $\tau$ depending upon pathogen and immune concentrations within-host $\tau$ units after infection, independent of the epidemic scale and with identical infection course among all hosts.  Recently, more complex scenarios have also  been considered, such as a distribution of immunity among susceptible hosts \cite{pugliese2011role} and a ``pathogen size- structured'' epidemic model with fully coupled feedback through variable initial pathogen load \cite{gandolfi2015epidemic}.  In addition, without explicitly modeling the within-host scale, several works have explored dynamic levels of host immunity in delay differential equation (DDE), PDE, and stochastic epidemic models with re-infection, immune boosting and waning \cite{martcheva2006epidemic,barbarossa2015immuno,veliov2016modelling,diekmann2018waning}.  

Dengue provides a particular example where host immunity has complex and significant effects on infection dynamics across both within-host and between-host scales.  Therefore, in this paper, we construct immunological and epidemiological models that capture signatures of ADE on both scales, connected via a variable tracking host cross-reactive antibody levels through multiple infections by distinct strains, along with recovery and waning.   First, we formulate a within-host model which mechanistically mimics characteristics of ADE in Dengue infection; namely (i) a shorter time to peak viremia, (ii) a higher maximum viral clearance rate, (iii) a higher peak viremia \cite{ben2015minimal}, and (iv) infection severity (measured by peak viremia) modulated by initial antibody concentration with a unimodal relationship \cite{katzelnick2017antibody,salje2018reconstruction}.  Moreover, we prove that our formulation is, in a sense, the minimal model to produce severe infection solely by varying pre-existent antibodies in an intermediate window of concentration.  Next, the within-host scale is linked to epidemiological spread by a PDE model structured by host antibody level.   Stability results in the epidemic model are formulated via basic and invasion reproduction numbers as a function of immunological variables.  Additionally, we develop numerical methods in order to simulate the multi-scale model and assess the influence of ADE on disease spread and burden in the population. Overall our model offers a very promising approach for understanding the connection between immunity and Dengue infection dynamics on both within-host and population scales.

\section{Multi-Scale Dengue Modeling: An Antibody Size-Structured Approach for Sequential Infections}\label{multiscale}
\subsection{Within-host model and analysis}
Here we formulate a within-host model which can describe primary and secondary Dengue infections, along with the host immune response.  We attempt to simplify the within-host dynamics of virus and immune response, while still capturing the mechanisms responsible for Antibody Dependent Enhancement (ADE) of the infection.   In order to model the host immune response, we consider long-lived memory antibodies (IgG) as proxies for the collective immune populations (which includes short-lived innate and IgM antibody responses, along with T-cell responses).  The memory antibodies increase upon infection, and can roughly be grouped into two categories, specific and cross-reactive (non-specific).  Distinct antibody populations target several epitopes (viral proteins) during infection, some of which are common amongst the different Dengue serotypes while other epitopes are \emph{specific} to the infecting strain.  The former are often termed (sero-)\emph{cross-reactive} antibodies and have less affinity to the infecting virus than the more specific antibodies.  This distinction is important as during secondary infection, or even possibly primary infection, pre-existent cross-reactive antibodies within the host may induce ADE.  


\begin{threeparttable}[t]
\caption{Description of the within-host model variables/parameters in \eqref{Secondary} and value chosen for simulations (which mimic qualitative within-host characteristics of Dengue)}
\centering 
\begin{tabularx}{\textwidth}{tbt}
\toprule
Variable/Parameter & \qquad \qquad \qquad Description & Value  \\ [0.5ex]
\toprule
\\
$x(\tau)$ & Dengue virus concentration at $\tau$ days post host infection & -- \\[0.5ex]
$y(\tau)$ & concentration of cross-reactive (IgG) antibodies at host infection age $\tau$ & -- \\[0.5ex] 
$z(\tau)$ &  concentration of specific antibodies at host infection age $\tau$ & --  \\[0.5ex]
$r$ & Within-host virus growth rate & $1$   \\[0.5ex] 
$\alpha_1$ & Viral growth enhancing rate induced by cross-reactive antibody-virus binding (ADE) & $2$   \\[0.5ex] 
$\alpha_2$ & Cross-reactive antibody-virus killing rate upon cooperative binding & $3$  \\[0.5ex] 
$\delta$ & Specific antibody-virus killing rate upon binding & $3.5$ \\[0.5ex] 
$\phi_1$ & Cross-reactive antibody activation rate &  $0.4$  \\[0.5ex]
$\phi_2$ & Specific antibody activation rate & $0.5$ \\[0.5ex]
$k_1, k_2$ & Antibody interference competition coefficient & $0.1$, $0.1$  \\[0.5ex]
$A_1,A_2, C_1$ & Saturation coefficients of Hill functions for cross-reactive antibody & $1,10,10$  \\[0.5ex]
$B,C_2$ & Saturation coefficients of Hill functions for specific antibody & $1,10$  \\[0.5ex]
\bottomrule
\end{tabularx}
\label{immparam} 
\end{threeparttable}\\

Consider the infecting virus strain, $x(\tau)$, specific IgG response to this strain, $z(\tau)$, and the cross-reactive (non-specific) IgG response, $y(\tau)$, where the time variable $\tau$ refers time-since-infection within a host.  The virus is assumed to undergo exponential growth at the rate $r$ for simplicity, and the specific IgG response $z(\tau)$ kills the virus and proliferates according to Michaelis-Menten kinetics.  There are multiple mechanisms for ADE which we include in our model.  First, studies have shown that neutralization of a virion requires more bound cross-reactive antibodies compared to the specific response \cite{dejnirattisai2010cross}.  Thus, we model the neutralization by cross-reactive response $y(\tau)$ with a sigmoidal Hill equation of ``$n=2$'' positive cooperativity \cite{stefan2013cooperative}, as opposed to ``$n=1$'' Michaelis-Menten kinetics.    In this way,  a threshold number of cross-reactive antibodies bound to a virion is required for neutralization and so low concentrations of cross-reactive antibodies have poor neutralization properties, consistent with a ``multiple-hit'' stoichiometry  requirement hypothesis \cite{wahala2011human,ripoll2019molecular,dowd2011antibody}.  Furthermore, \emph{any} antibody-virion binding can actually enhance probability of cell infection \cite{dejnirattisai2010cross}, and thus we add a Michaelis-Menten enhancing term dependent on cross-reactive antibodies $y(\tau)$ to the viral replication rate.  Note that the efficient neutralization by specific antibodies with higher affinity to virions precludes the need to add a similar enhancing term for $z(\tau)$.   Another possible mechanism of ADE is ``original antigenic sin'' where antibody populations compete and interfere with each other \cite{nikin2017modelling}. This is included by interference competition coefficient $k_1$ and $k_2$ inhibiting the proliferation rates of $y$ and $z$.  With these features in mind, the following Dengue within-host model is novel for its enzyme kinetics mechanisms of ADE.

\begin{equation}\label{Secondary} \textbf{Within-host Model\quad} 
\begin{cases}
\D\frac{dx}{d \tau} &\hspace{-4mm}=x\left(r+\D\frac{\alpha_1 y}{A_1+y}-\D\frac{\alpha_2 y^2}{A_2+y^2}-\frac{\delta z}{B+z}\right):=f(x,y,z) \vspace{1.0mm}\\
\D\frac{dy}{d \tau} &\hspace{-4mm}=\D\frac{\phi_1 x y}{C_1+y+k_1 z}:=g(x,y,z) \vspace{1.0mm}\\
\D\frac{dz}{d \tau} &\hspace{-4mm}=\D\frac{\phi_2 x z}{C_2+k_2y+z}:=h(x,y,z)  \vspace{1.0mm}\\
\end{cases}
\end{equation}
\vskip-0.09in 
\noindent
The initial conditions $x(0)=x_0, y(0)=y_0,$ and $z(0)=z_0$ and all parameters are assumed to be non-negative. \\

We remark that the assumed exponential growth of the virus in the absence of immune response is a simplification of viral replication dynamics.  In reality there is a source of target cells which is depleted (and has recruitment), which bounds the virus population dynamics. Target cells have been included in previous models of Dengue infection \cite{ben2015minimal}.  For simplicity, we assume in our model that immune suppression of the virus either overwhelms or is more important than any effect of target cell limitation on the dynamics.  Also, here and as been found in other studies, immune response is necessary to clear Dengue virus \cite{clapham2016modelling}.   The following theorem shows that in our model the immune response is sufficient to clear the virus.  
\begin{proposition}\label{prop1} 
Suppose that $\delta>r+\alpha_1$ and $\alpha_2>r+\alpha_1$ in system \eqref{Secondary}.  If $y_0>0$ or $z_0>0$, then $\lim\limits_{\tau\rightarrow\infty}x(\tau)=0, \lim\limits_{\tau\rightarrow\infty}y(\tau)=\bar y, \lim\limits_{\tau\rightarrow\infty}z(\tau)=\bar z$ where $\bar y,\bar z$ depends on initial conditions.
\end{proposition}
\begin{proof}
Since the boundary of the positive orthant, $\partial \mathbb R^3_+$, is invariant for \eqref{Secondary}, we find that $\mathbb R^3_+$ is also invariant. Thus solutions remain non-negative for all $\tau$.  Also it is not hard to show that solutions exist for all $\tau$ since the differential inequality $\left(x+y+z\right)'\leq c\left(x+y+z\right)$ can be established for appropriate constant $c$, which yields an exponential bound for the solution.   

Now suppose by way of contradiction that $\limsup\limits_{\tau\rightarrow\infty}x(\tau)>0$.  Let $w=y+z$ and observe that $w'\geq \frac{\phi x w}{C+kw}$, where $\phi=\max(\phi_1,\phi_2)$ and $C=\min(C_1,C_2)$.    Integrating, we find $C\ln\left(\frac{w(\tau)}{w_0}\right) + w(\tau)-w(0)\geq \phi\int_0^\tau x(s)ds$.  Thus $\lim\limits_{\tau\rightarrow\infty}w(\tau)=\infty$.  Therefore, since $\delta>r+\alpha_1$ and $\alpha_2>r+\alpha_1$, there exists $\tau^*: \ \forall \tau>\tau^*$, where $r +\D\frac{\alpha_1 y(\tau)}{A_1+y(\tau)}-\D\frac{\alpha_2 y(\tau)^2}{A_2+y(\tau)^2}- \frac{\delta z(\tau)}{B+y(\tau)}<0$.  Integrating the first equation in \eqref{Secondary}, we find $x(\tau)=x_0\exp\left(\int_0^\tau  \left(r +\D\frac{\alpha_1 y(\tau)}{A_1+y(\tau)}-\D\frac{\alpha_2 y(\tau)^2}{A_2+y(\tau)^2}- \frac{\delta z(\tau)}{B+z(\tau)}\right)ds\right)$. Since the integrand is negative for  $\tau>\tau^*$, we conclude $\lim\limits_{\tau\rightarrow\infty}x(\tau)=0$. Next, since $0\leq y'(\tau)$, $y(\tau)$ is increasing, it has either positive limit or diverges to infinity asymptotically.  Suppose by way of contradiction that $\lim\limits_{\tau\rightarrow\infty}y(\tau)=\infty$.  Then $\frac{dx}{d\tau}\leq -a x(\tau)$ for $\tau$ sufficiently large, say $\tau\geq\tau_0$, where $a=\alpha_2-(r+\alpha_1)+\epsilon>0$ for some $\epsilon$ sufficiently small.   Thus $x(\tau)\leq x(\tau_0)e^{-a(\tau-\tau_0)}$, and thus
$$ y'\leq \frac{\phi_1 x y}{C_1+y}\leq \frac{\phi_1 x }{C_1} \Rightarrow y(\tau)-y(\tau_0)\leq \frac{\phi_1}{C_1}\int_{\tau_0}^{\tau}e^{-a(\tau-\tau_0)}d\tau. $$
Therefore $y(\tau)$ is bounded and there exists $\bar y>0$ such that $\lim\limits_{\tau\nearrow\infty}y(\tau)=\bar y$. Similarly $z(\tau)$ converges monotonically to a limit $\bar z$.  
\end{proof}

The trajectory of the virus population mimics the \textit{general pattern} displayed in data of rise and subsequent decline of virus caused by immune response, characteristic of an acute infection.  We obtain a triangular curve in log scale of viral load as found in other studies \cite{ben2015minimal}, and Proposition \ref{prop1} shows that the virus population converges to zero while (memory) immune responses saturate to an equilibrium level dependent on initial concentrations.  Moreover the within-host model mechanistically mimics characteristics of ADE in Dengue infection; namely \textit{(i) a shorter time to peak viremia, (ii) a higher maximum viral clearance rate, (iii) a higher peak viremia} \cite{ben2015minimal}, and \textit{(iv) infection severity (measured by peak viremia) modulated by initial antibody concentration with a unimodal relationship} \cite{katzelnick2017antibody,salje2018reconstruction}.  Note that fitting the within-host model to data is not a goal of the present work.  However, we do tune parameters in system \eqref{Secondary} to first match infectious period of primary infection (Figure \ref{fig1a}), and after a characteristic period where cross-reactive antibodies, $y$, wane to a certain range, subsequent secondary infection displays features (i)-(iii) associated with DHF induced by ADE, as shown in Figure \ref{fig1b}.

\begin{figure}[t]
\subfigure[][]{\label{fig1a}\includegraphics[width=7.5cm,height=4cm]{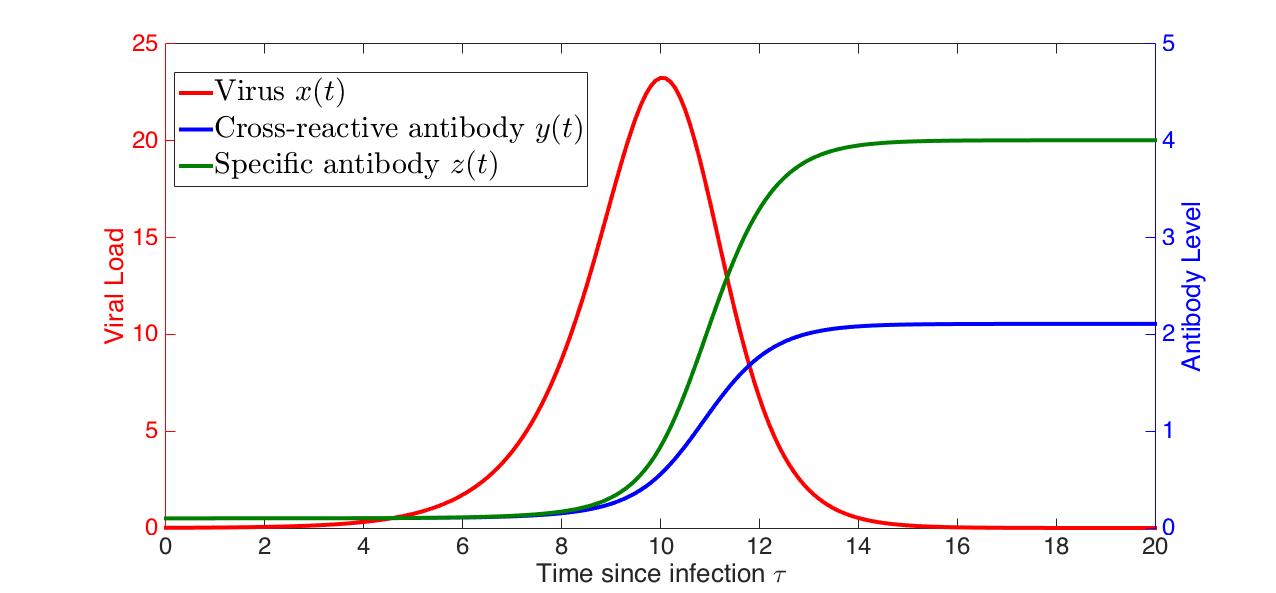}}
\subfigure[][]{\label{fig1b}\includegraphics[width=7.5cm,height=4cm]{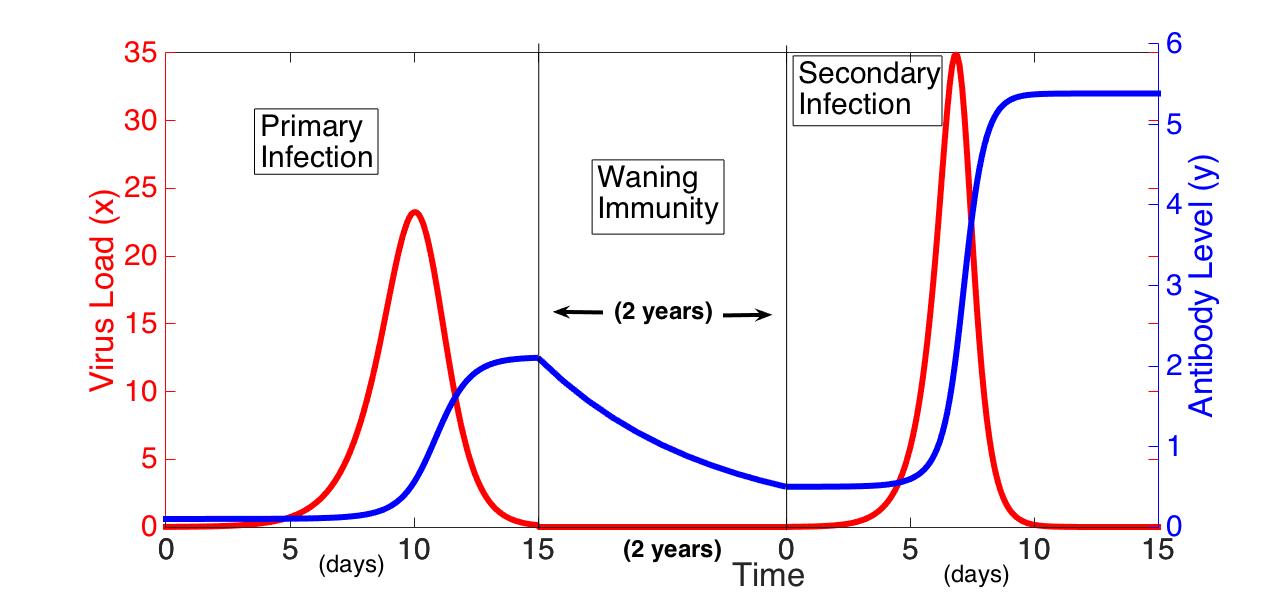}} \\
\subfigure[][]{\label{fig1c}\includegraphics[width=7.5cm,height=4cm]{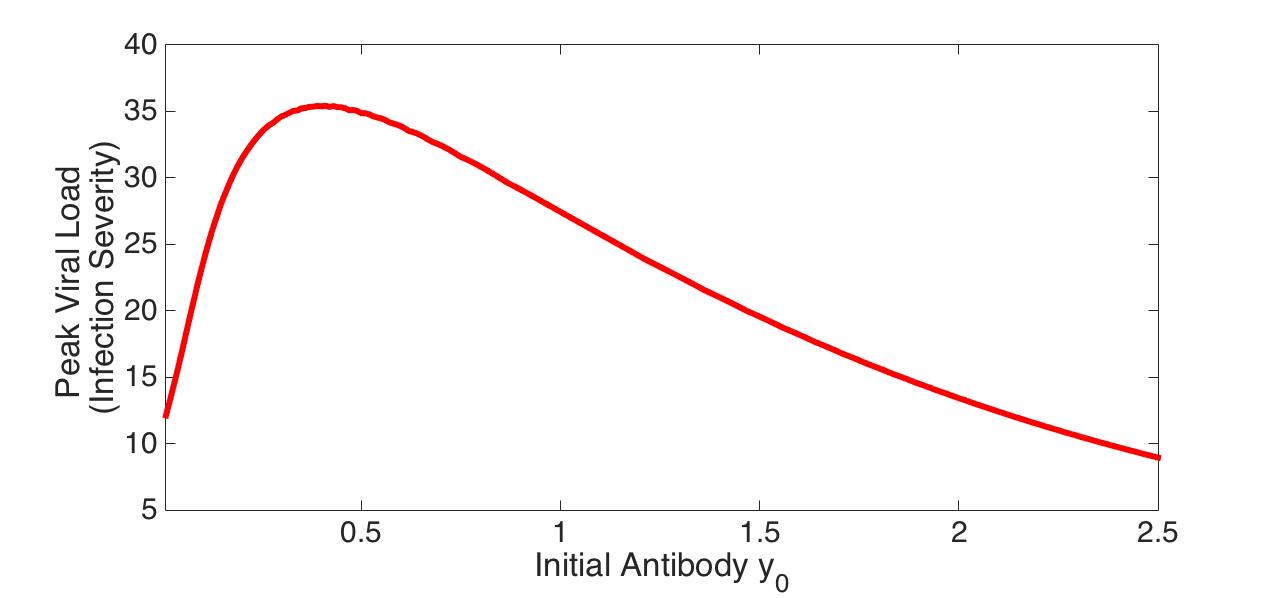}}
\subfigure[][]{\label{fig1d}\includegraphics[width=7.5cm,height=4cm]{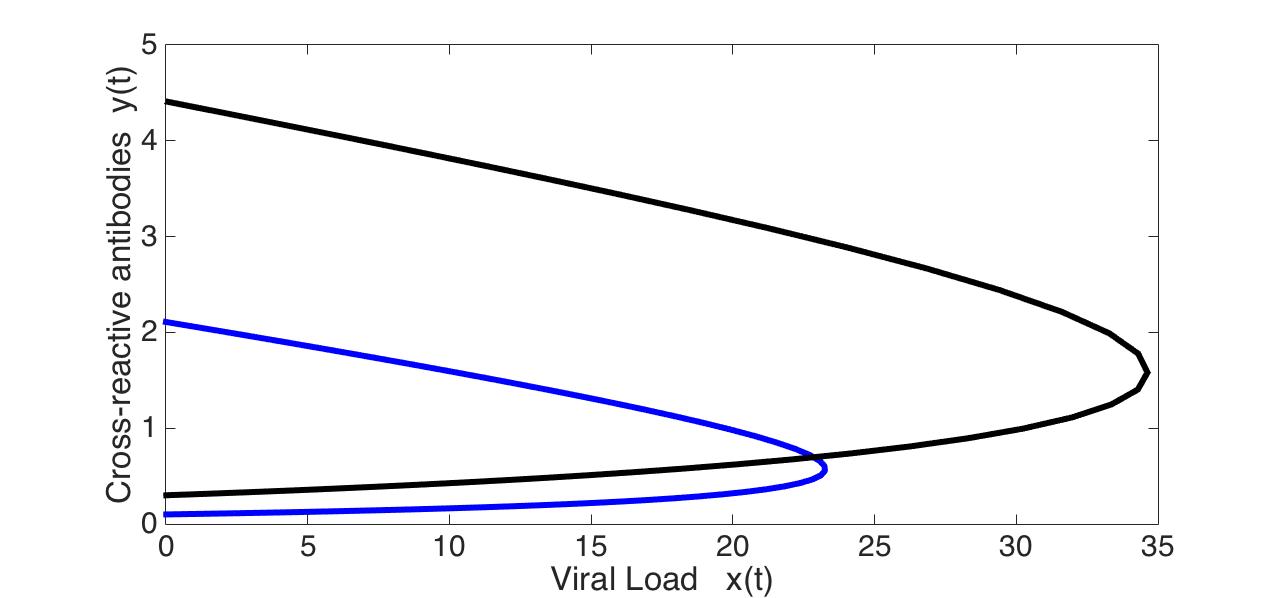}} 
\caption{(a) Example trajectory of within-host model \eqref{Secondary}. (b) Cross-reactive antibody levels ($y$ or $y_0$) are boosted during primary infection and then wane to an intermediate level which can produce severe infection upon secondary infection due to ADE effect. (c) Host infection severity (peak viral load) unimodal function of pre-existent antibody level ($y_0$). (d) Two orbits corresponding to solutions of within-host model \eqref{Secondary}.  The parameters (and initial conditions) of within-host model \eqref{Secondary} are set as follows: (a) $y_0=0.11$; (b) $y_0=0.11$ for primary infection and $y_0=0.5$ for secondary infection (after 2 years of antibody exponential decay given by waning rate \eqref{expwane} with $\xi=0.002$, $y_c=0.02$; (c) $y_0$ varied in range $[0.01,2.5]$; (d) $y_0=0.11$ and $y_0=0.3$.  All other parameters are fixed at $r=1,\alpha_1=2,A_1=1,\alpha_2=3,A_2=10,\delta=3.5,B=1,\phi_1=0.4,\phi_2=0.5,C_1=C_2=10,k_1=k_2=0.1$, and $x_0=0.01,z_0=0.1$. } 
  \label{fig1}
  \end{figure}
  
The system \eqref{Secondary} captures the signature of Dengue infection severity being highest for intermediate pre-existent antibody level ((iv) in previous paragraph), and is, in some sense, the minimal model to produce this unimodal relationship, as the following proposition suggests:
\begin{proposition}\label{maxProp}
Consider system \eqref{Secondary} with $\delta>r+\alpha_1$ and $\alpha_2>r+\alpha_1$.  Let $x(\tau;y_0)$ denote the viral component of solution as a function of initial cross-reactive antibody level $y_0$, and $x_M(y_0):=\max_{\tau\geq 0}x(\tau;y_0)$ denote the peak viral load as a function of $y_0$  (with other parameters fixed).  Then there exists some set of parameters for the system \eqref{Secondary} where $x_M(y_0)$ is an unimodal curve with a single maximum, in particular when $\phi_1=k_2=0,A_2=3,C_2=B$, and $A_1>\frac{8\alpha_1}{3\alpha_2}$.  Furthermore, if $z=0$ ($y=0$), then $x_M(y_0)$ ($x_M(z_0)$) will be a strictly decreasing function.
\end{proposition}



\begin{proof}

In order to prove the first statement, consider the special case where $\phi_1=0$, for simplicity, since we just need to show it for some parameter set.  
In this case the cross-reactive antibody concentration $y$ does not change during infection, i.e. $y(t)\equiv y_0$.   Then the infection dynamics are  
\begin{equation}\label{Sec_1}
\begin{cases}
\D\frac{dx}{d \tau} &\hspace{-4mm}=x\left(r(y_0)-\frac{\delta z}{B+z}\right)
\vspace{1.0mm}\\
\D\frac{dz}{d \tau} &\hspace{-4mm}=\D\frac{\phi_2 x z}{C(y_0)+z}\vspace{1.0mm},\\
\end{cases}
\end{equation}
\vskip-0.09in 
\noindent
where $r(y_0)=r+\frac{\alpha_1 y_0}{A_1+y_0}-\frac{\alpha_2 y_0^2}{A_2+y_0^2}$ and $C(y_0)=C_2+k_2y_0$.   We assume that $\alpha_2>r+\alpha_1$ as in Proposition \ref{prop1}.  We first claim the function $r(y_0)$ is unimodal. Since $r'(y_0)=\frac{\alpha_1A_1}{(A_1+y_0)^2}-\frac{2\alpha_2 A_2y_0}{(A_2+y_0^2)^2}$, we find the following polynomial equation for the roots of $r'(y)$ (which are the critical points of $r(y)$):  
\begin{align}
\alpha_1 A_1 y^4-2\alpha_2A_2y^3+ 2A_1A_2(\alpha_1-\alpha_2)y^2-2A_2A_1^2\alpha_2 y + \alpha_1 A_1A_2^2=0 .
\end{align}
Since $\alpha_1 -2\alpha_2<0$ by assumption, there are either two or zero positive roots by Descarte's rule of signs.  We now find parameters where there are two positive roots, $y_1^*<y_2^*$.   Let $p(y)=\frac{\alpha_1A_1}{(A_1+y)^2}$ and $q(y)=\frac{2\alpha_2 A_2y}{(A_2+y^2)^2}$, and observe that $r'(y)=p(y)-q(y)$.  Note that $p(0)=\frac{\alpha_1}{A_1}>0$, $p(y)$ is decreasing, $q(0)=0$ and $q'(0)>0$.  The goal is to obtain conditions where  $\max\limits_{y>0}q(y)>p(0)$ which would guarantee intersection of $p(y)$ and $q(y)$, and hence a positive root.  It can be shown that $q'(y)=0\Rightarrow A_2+y-4y^2=0$.  Let $A_2=3$, then the maximum of $q(y)$ occurs at $\hat y=1$, and $q(1)=\frac{3}{8}\alpha_2>\frac{\alpha_1}{A_1}=p(0)$ if $A_1>\frac{8\alpha_1}{3\alpha_2}$.  Thus if $A_1=A_2=3$, then there are two positive roots.  Furthermore since $r'(0)>0$, $y_1^*$ is a local maximum of $r(y)$ and $y_2^*$ is a local minimum.  Also $r(\infty)<0$ implies that there exists unique positive root $\overline y$ of $r(y)$, where $r(y)>0$ for $0\leq y< \overline{y}$ and $r(y)<0$ for $y> \overline{y}$ (where $\overline{y}<y_2^*$).  Clearly if $y_0>\overline{y}$, then $\dot{x}(t)<0$ for all $t\geq 0$, so $x_M(y_0)=x_0$ in this case.  Furthermore  we show that peak viral load $x_M(y_0)$ is increasing with respect to the growth rate $r(y_0)$ for the special case $k_2=0,C_2=B$.  Indeed in this case, peak viral load occurs at $z_c(r):=\frac{rB}{\delta-r}$, and dividing equations in \eqref{Sec_1}
\begin{align*}
\frac{dz}{dx}=\frac{\phi_2 z}{rB+(\delta-B)z} &\Rightarrow \int_{z_0}^{z_c(r)} dz\left(\frac{rB}{\phi_2z}+\frac{\delta-B}{\phi_2}\right)= \int_{x_0}^{x_M(r)} dx \\
&\Rightarrow x_M'(r)=\frac{B\delta}{(\delta-r)^2}\left(\frac{\delta-r}{\phi_2}+\frac{\delta-B}{\phi_2}\right)>0
\end{align*} 
Therefore it follows that the peak viral load $x_M(y_0)$ is unimodal with a single maximum.    

 Now to show the second statement.  For given parameters satisfying $\delta>r+\alpha_1$ and $\alpha_2>r+\alpha_1$ with either $z=0$ or $y=0$, Proposition \ref{prop1} implies the pathogen load will tend to zero while the single present antibody population is increasing.  Viewed in the phase plane of pathogen ($x$) and single antibody population (without loss of generality, $y$), the orbits form arcs connecting initial antibody level, $y_0$, with final antibody level $y^*(y_0)=\lim\limits_{\tau\rightarrow\infty}y(\tau;y_0)>y_0$, as $x$ increases to peak and decays to zero.  When increasing initial antibody level, $y_0$, in order for the peak viral load to increase, solutions would have to cross violating flow property of solutions.  Thus peak viral load can then only be a decreasing function of initial antibody load in this case.  
\end{proof}

Proposition \ref{maxProp} implies that the presence of both cross-reactive and specific antibodies within infection are necessary to produce the ADE phenomenon of severe infection for intermediate level of pre-existing antibodies.  The full model \eqref{Secondary} with both antibody types produces the unimodal curve for infection severity versus initial antibody load (Fig. \ref{fig1c}), similar to data from recent epidemiological studies.  Observe in Fig. \ref{fig1d}, the ``crossing'' of two solution \emph{projections} on the $xy$-plane when both antibody components $y,z$ are present and $y_0$ is varied.  In contrast, Proposition \ref{maxProp} states that peak viral load cannot increase with $y_0$ when $z=0$.   Moreover sufficient conditions to generate the unimodal pattern are dominance of (cross-reactive) antibody enhanced viral infection rate at low concentrations (controlled by $A_1$)  switching to dominance of neutralization at higher concentrations (controlled by $\alpha_2$), consistent with the observed mechanisms responsible for ADE in experiments \cite{dejnirattisai2010cross}.





\subsection{Between-host model and linking scales} Here we detail our antibody structured vector-host epidemiological model which links to the within-host model, tracking evolving antibody levels as illustrated in Fig. \ref{fig1bh}. Let $s(t,y)$, $i_k(t,y,y_0)$, $r_k(t,y)$, $i_{kj}(t,y,y_0)$, $r_{kj}(t,y)$ be the density with respect to (cross-reactive) antibody level $y$ (and initial antibody level $y_0$ at time of infection) at time $t,$ of susceptible, primary strain-$k$ infected, primary strain-$k$ infection recovered, secondary strain-$j$ infected hosts, and secondary infected recovered individuals, respectively.   In the vector model, we consider $S_v(t), I_v(t)$ as the number of susceptible and infected vectors, respectively. Vectors are the only mechanism transmitting the disease to susceptible hosts. The host compartments structured by antibody levels, $y$, can be integrated over $y$ (and $y_0$ in the case of infected classes) in order to obtain the number of individuals in each compartment.  For example, the number of susceptible individuals is given by $S(t)=\int\limits_0^\infty s(t,y) \, dy$ and the number of individuals infected by strain $k$ is $\int\limits_{0}^{\infty}\int\limits_{y_0}^{\infty} i_k(t,y,y_0) \,dy \,dy_0$.   The epidemiological dynamics is given by the following vector-host model:

\begin{align}\label{host_epi}
\frac{\partial s(t,y)}{\partial t} &= \Lambda(y) - s(t,y)\sum_{k}\beta_{v}^k(y) I_v^k(t) - \mu s(t,y) \\
\frac{\partial i_k(t,y,y_0)}{\partial t}+\frac{\partial (g_k(y,y_0) i_k(t,y,y_0))}{\partial  y} 
&=  
     - \left(\gamma_k(y,y_0)+\mu\right)i_k(t,y,y_0)   \\
  \frac{\partial r_k(t,y)}{\partial t}- \frac{\partial (\omega_k(y) r_k(t,y))}{\partial y} &=\int_0^{\infty}\gamma_k(y, y_0)i_k(t,y,y_0)dy_0-\mu r_k(t,y)- \beta_v^j(y) r_k(t,y)I_v^j(t) \notag \\
\frac{\partial i_{kj}(t,y,y_0)}{\partial t}+\frac{\partial (g_{kj}(y, y_0) i_{kj}(t,y, y_0))}{\partial y} 
&= - \left(\gamma_{kj}(y, y_0)+\mu\right) i_{kj}(t,y, y_0)  \notag   \\
  \frac{\partial r_{kj}(t,y)}{\partial t}- \frac{\partial (\omega_{kj}(y) r_{kj}(t,y))}{\partial y} &=\int_0^{\infty}\gamma_{kj}(y, y_0)i_{kj}(t,y,y_0)dy_0-\mu r_{kj}(t,y), \quad k\neq j, \notag
 \end{align}

\begin{align}\label{vector_epi}
\frac{dS_v}{dt} &= \Lambda_v - S_v\sum\limits_{k=1}^2\int\limits_{0}^{\infty}\int\limits_{y_0}^{\infty}\beta_k(y,y_0) i_k(t,y,y_0) \,dy \,dy_0- S_v\sum\limits_{\substack{k,j=1 \\ k\neq j}}^2\int\limits_{0}^{\infty}\int\limits_{y_0}^{\infty}\beta_{kj}(y,y_0) i_{kj}(t,y,y_0) \,dy \,dy_0 - \mu_v S_v  \notag \\
\frac{dI_v^k}{dt} &= S_v\int\limits_{0}^{\infty}\int\limits_{y_0}^{\infty}\beta_k(y,y_0) i_k(t,y,y_0) \,dy \,dy_0 + S_v\int\limits_{y_0}^{\infty}\int\limits_{y_0}^{\infty}\beta_{jk}(y,y_0) i_{jk}(t,y,y_0) \,dy \,dy_0 - \mu_v I_v^k,
\end{align}
with the following boundary conditions
\begin{align}\label{BC}
 g_k(y_0,y_0) i_k(t,y_0,y_0) & = \beta_v^k(y_0)s(t,y_0) I_v^k(t), \quad \omega_k(y_c) r_k(t,y_c)= \lim_{y\rightarrow\infty}\omega_k(y)r_k(t,y)=0, \\ g_{kj}(y_0, y_0)i_{kj}(t,y_0, y_0)&=\beta_v^j(y_0) r_k(t,y_0)I_v^j(t), \quad \omega_{kj}(y_{c,2}) r_{kj}(t,y_{c,2})= \lim_{y\rightarrow\infty}\omega_{kj}(y)r_{kj}(t,y)=0, \ \ k\neq j, \notag
\end{align}
and the following initial conditions
\begin{align}\label{IC}
 s(0,y_0)=s_0(y_0), \ \ y_0\geq y_s\geq 0, \quad i_k(0,y,y_0)=i_k^0(y,y_0), \ \ y\geq y_0\geq y_s \\ 
 r_k(0,y)=r_k^0(y), \ \  y\geq y_c \geq 0, \quad i_{kj}(0,y,y_0)=i_{kj}^0(y,y_0), \ \ y\geq y_0\geq y_c,   \notag \\
  r_{kj}(0,y)=r_{kj}^0(y), \ \  y\geq y_c \geq 0, \quad S_v(0)=S_v^0, \quad I_v(0)=I_v^0.  \notag
\end{align}
  
  The host initial conditions are assumed to be non-negative (Lesbesgue) integrable functions, i.e. in $L^1_+$, on their domains specified above, and vector initial conditions are non-negative, i.e. $S_v^0,I_v^0\in\mathbb R_+$.  The parameters $\Lambda(y)$ and $\Lambda_v$ denote the host and vector recruitment rates, and $\mu$ and $\mu_v$ represents the host and vector natural death rates, respectively.  The vector to host transmission rate may depend on the host antibody level, so in general we have this rate as $\beta_v^k(y)$.  We assume $\Lambda(y)$ and $\beta_v^k(y)$ are bounded, measurable non-negative functions, i.e. in $L_+^\infty(0,\infty)$.   The other parameter functions linking antibody levels $y$ and $y_0$ to epidemiological quantities will be detailed in following paragraphs.    First note that assumptions may be relaxed at times to allow for point measure distributions, e.g. all susceptible individuals have the same initial naive amount of cross-reactive antibodies, $y_s$, so that $\Lambda(y)=\Lambda \delta(y_s), s_0(y)=S_0\delta(y_s)$ where $\delta(y)$ is the Dirac delta measure at $y$, and $\Lambda$ is constant.  In this case, we can consider an ODE for $S(t):=\int_0^{\infty}s(t,y)\delta(y_s)dy=s(t,y_s)$.  Also, we remark that the secondary recovered component, $r_{kj}$, decouples from the rest of the system.
  
  \begin{figure}[t!]
\includegraphics[width=15cm,height=4cm]{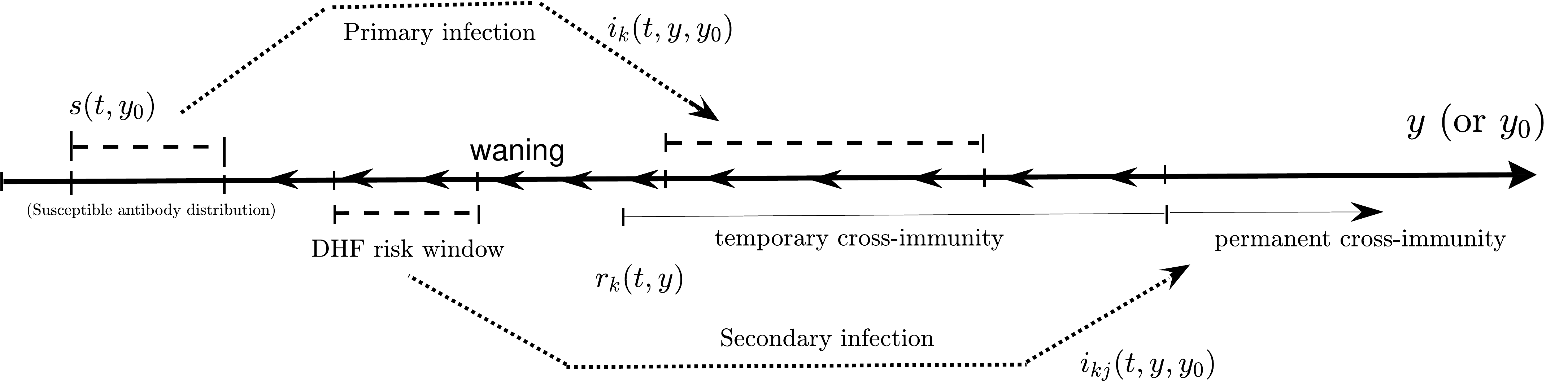}
\caption{Schematic diagram of multi-scale Dengue model \eqref{host_epi}-\eqref{vector_epi} viewed through evolving host antibody level.} 
  \label{fig1bh}
  \end{figure}
  
  \begin{table}[t]
\caption{Definitions of the between-host model variables}
\centering 
\begin{tabularx}{\textwidth}{>{} lX}
\toprule
Variable  & Description \\ [0.5ex]
\toprule
\\
$s(t,y)$ & density of susceptible hosts at time $t$ with (cross-reactive) antibody level $y,$ \\[0.5ex]
$ i_k(t,y,y_0)$ & density of primary strain-$k$ infected hosts at time $t$ with antibody level $y$ and pre-existent antibody level $y_0$, \\[0.5ex]
$r_k(t,y)$ & density of strain-$k$ recovered hosts at time $t$ with antibody level $y,$ \\[0.5ex]
$i_{kj}(t,y,y_0)$ & density of secondary strain-$j$ infected hosts at time $t$ with antibody level $y$ and pre-existent antibody level $y_0$, \\[0.5ex]
$S_v(t)$ & number of susceptible vectors at time $t,$ \\[0.5ex]
$I_v(t)$ & number of infected vectors at time $t,$ \\[0.5ex]
\bottomrule
\end{tabularx}
\label{table:variables} 
\end{table}

 \begin{threeparttable}[t]
\caption{Definition of the between-host model parameters}
\centering 
\begin{tabularx}{\textwidth}{tbs}
\toprule
Parameter  & Meaning & Values (range) \\ [0.5ex]
\toprule
\\
$ \Lambda (y)$ & recruitment rate of susceptible hosts with antibody level $y$ & $\mathbb E(\Lambda)=100$  (variable distribution centered around $y_m=0.11$) \\[0.5ex]
$  \Lambda_v$ & susceptible vector recruitment rate & $0.02$ \\[0.5ex]
$\beta^k_v (y)$ & transmission rate from k-strain infected vector to susceptible hosts with antibody level $y$ & $\begin{cases} 0.00025 & y<y_p \\ 0 & y>y_p \ \end{cases}$  \\[0.5 ex]
$\beta_{k}(y,y_0)$  & transmission rate from primary strain-$k$  infected hosts, with antibody level $y$ (and initial antibody level $y_0$), to susceptible vectors &  \tnote{a}\\[0.5ex]
$\beta_{kj}(y,y_0)$  & transmission rate from secondary strain-$j$  infected hosts, with antibody level $y$ (and initial antibody level $y_0$), to susceptible vectors &  \tnote{a}\\[0.5ex]
$\gamma_k (y,y_0)$ & recovery rate of  primary strain-$k$ infected hosts, with antibody level $y$ (and initial antibody level $y_0$)&  \tnote{a}\\[0.5ex]
$\gamma_{kj} (y,y_0)$ & recovery rate of secondary strain-$j$ infected hosts, with antibody level $y$ (and initial antibody level $y_0$) &  \tnote{a}\\[0.5ex]
$g_k (y,y_0)$ & antibody growth rate ($\frac{dy}{d\tau}$ in \eqref{Secondary}) during  primary infection &  \tnote{a}\\[0.5ex]
$g_{kj} (y,y_0)$ & antibody growth rate ($\frac{dy}{d\tau}$ in \eqref{Secondary}) during  secondary infection &  \tnote{a} \\[0.5ex]
$\omega_k (y)$ & (cross-reactive) antibody waning rate after primary infection &  \tnote{a} \\[0.5ex]
$\omega_{kj} (y)$ & (cross-reactive) antibody waning rate after  secondary infection &  \tnote{a}  \\[0.5ex]
$\mu$ & Host natural death rate & $1/(10\times 365), 1/(55\times 365)$ \\[0.5 ex]
$\mu_v$ & Vector natural death rate   & $1/20$  \\[0.5 ex]
\bottomrule
\end{tabularx}
\label{table:parameters} 
 \begin{tablenotes}
     \item[a] see Table \ref{Secondary} \& \ref{table:linkparameters}, along with linking functions \eqref{recov1}, \eqref{recov1}, and \eqref{transmis}. 
   \end{tablenotes}
\end{threeparttable}
 
The functions $g_k(y,y_0)$ and $g_{kj}(y,y_0)$ represent the memory antibody concentration growth rates corresponding to primary infection with strain $k$ and secondary infection by strain $j$, respectively.  They can be formally defined as follows.    Consider the solution to the within-host system \eqref{Secondary} with parameters corresponding to the particular infection type (strain and primary/secondary).  Note that within the general within-host model \eqref{Secondary}, parameters may differ between strain and whether it is primary or secondary infection.  Given the solution $x(\tau):=x(\tau;x_0,y_0,z_0), y(\tau):=y(\tau;x_0,y_0,z_0), z(\tau):=z(\tau;x_0,y_0,z_0)$, define inverse map $\tau=\tau(y;y_0)$ corresponding to time since infection, noting that $y(\tau;y_0)$ is strictly increasing function of $\tau$ for each $y_0>0$ (holding $x_0,z_0$ fixed).  For primary or secondary infection, we assume fixed initial concentration of pathogen $x_0$ and fixed initial specific naive (memory) specific antibody concentration of $z_0>0$.  The initial cross-reactive antibody concentration, $y_0=\tilde y$, is given by the structuring variable of the susceptible host which becomes infected, $s(t,\tilde y)$, in the case of primary infection.   Then $g_k(y,y_0)=g(x(\tau;x_0,\tilde y,z_0),y(\tau;x_0,\tilde y,z_0),z(\tau;x_0,\tilde y,z_0)))$ where $g$ is the second component of the vector field in the within-host system \eqref{Secondary} with parameters corresponding to primary infection by strain $k$.  Similarly $g_{kj}(y,y_0)=g(x(\tau;x_0,\tilde y,z_0),y(\tau;x_0,\tilde y,z_0),z(\tau;x_0,\tilde y,z_0))$ with parameters corresponding to secondary infection by strain $j$ in system \eqref{Secondary}.

  Multiple studies have shown that following primary infection, individuals have a temporary period of immunity to different serotypes induced by cross-reactive antibodies primed by the primary infecting serotype.  This immunity can be generated in our within-host models via the rise of antibody concentration during primary infection to levels sufficient for inhibition of secondary infection.  However, in reality the immunity can wane through time allowing for secondary infection by serotypes distinct from the primary strain, potentially manifesting in dengue hemorrhagic fever caused by ADE at intermediate levels of cross-reactive antibody.  Note that although the antibody level wanes through time, recovered individuals remain immune to the primary infecting strain, thus the \emph{total} antibody levels (in particular specific antibodies) stay above some critical level for strain-specific immunity.    With these features in mind, we include a drift term for waning antibody level.  Suppose that antibody levels change according to $\dot y=-\omega_k(y)$, where $\omega_k \geq 0$ is the rate of antibody waning after recovery from primary infection by strain $k$.  We assume that $\omega_k(y)\rightarrow0$ as $y\rightarrow y_c^+$, so antibody levels stay above some $y_c\geq 0$, as studies show that antibodies do not wane completely.  Similar assumptions are made for individuals recovered from secondary infection, however since they have permanent immunity to both strains, the $r_{kj}$ equation is decoupled from the system.   A specific example, supported by a study of waning IgG  \cite{antia2018heterogeneity}, is exponential decay of memory antibodies, above the critical level $y_c$:
 \begin{align}
 \omega_k(y)&=\xi (y-y_c),  \qquad
 \text{so that} \  \ y(\tau)=(y_0-y_c)e^{-\xi\tau} + y_c . \label{expwane}
 \end{align}
This exponential rate form is utilized in between primary and secondary infection in Fig. \ref{fig1b} with $\xi$ and $y_c$ calibrated to produce the displayed waning antibody level in the characteristic period of 2 years corresponding to loss of cross-immunity \cite{reich2013interactions}.

\begin{figure}[t]
\subfigure[][]{\label{fig2a}\includegraphics[width=7.5cm,height=4cm]{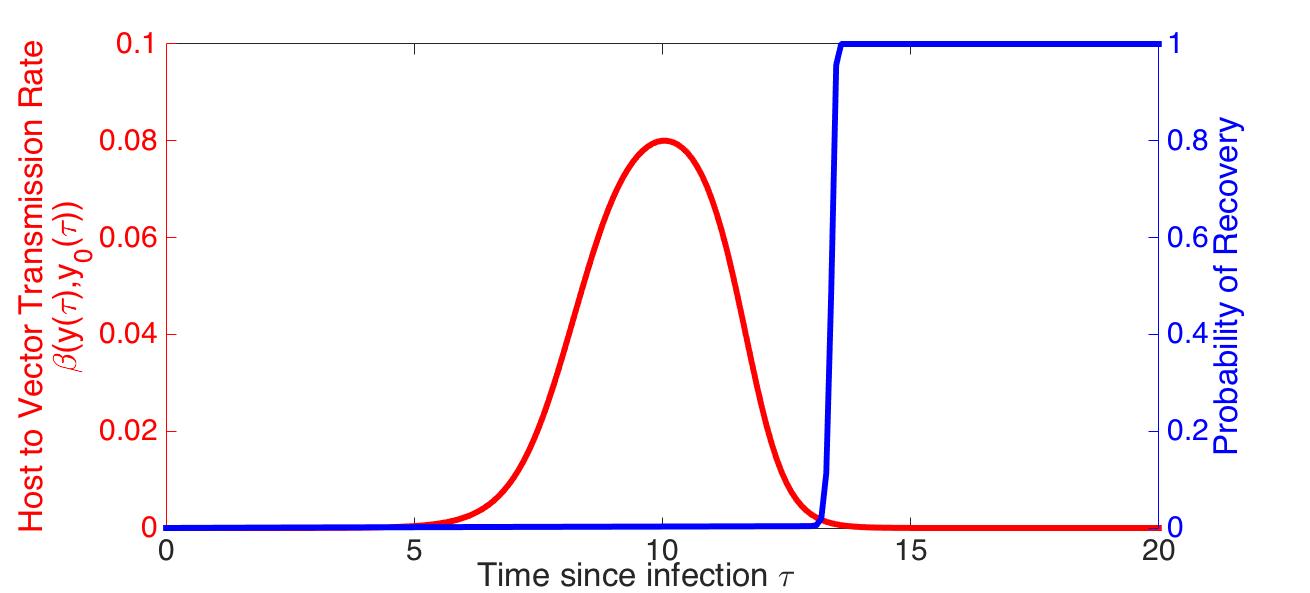}}
\subfigure[][]{\label{fig2b}\includegraphics[width=7.5cm,height=4cm]{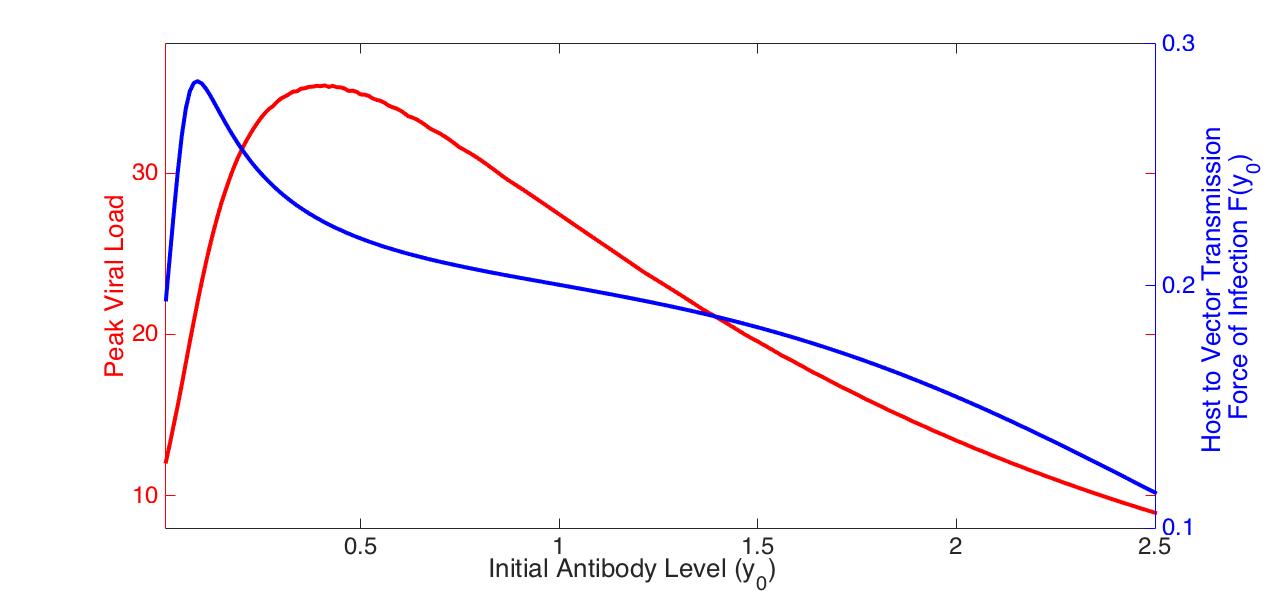}}
\caption{(a)  Transmission rate, $\beta_{\ell}(\tau(y,y_0))$, and probability of exiting infectious period (through recovery or death), given by $1-\pi_{\ell}(\tau(y,y_0))$, as functions of time since infection $\tau=\tau(y,y_0)$. (b) Infection severity (as shown in Fig. \ref{fig1c}) and host-vector force of infection $\mathcal F_k(y_0)$ as functions of pre-existent antibody level $y_0$.  The parameters utilized for the within-host model \eqref{Secondary} are the same as Fig. \ref{fig1} with $y_0=0.11$ for (a).  For the linking functions, $\beta(\tau(y,y_0))$ and $\gamma(\tau(y,y_0))$, we utilize \eqref{transmis} with $\psi=0.001$ and $C_{\ell}=135$ (based on \cite{nguyen2013host}), and \eqref{recov1} with $\rho=\nu=10$, respectively.  } 
  \label{fig2}
  \end{figure}

On the host population scale, the preexistent cross-reactive antibody level, $y_0$, can vary according to the distribution $\Lambda(y_0)$, leading to variable within-host primary infection dynamics in the population.  The recovery and waning process creates more heterogeneity in antibody level among the host population, leading to variable responses to secondary exposure.  In order to formulate the recovery rate, first note that Dengue is an acute infection with an approximate triangular viral load trajectory in log scale, suggested by both data and our within-host model.   In addition, recovery from primary infection induces lifelong immunity to the primary infecting serotype.  With these features in mind, we consider a few possible recovery rate forms.  First, we assume the recovery rate increases as viral load slope becomes negative and the viral load becomes small.   The second rate form specifies that viral load should be decreasing, i.e. viral load slope becomes negative, which is a necessary condition for protective immunity against the same strain.  Last, we suppose that recovery occurs at a fixed level of antibody.  In particular, we give the following three examples of recovery rates;
\begin{align} 
(i) \ & \gamma_k(\tau(y,y_0))= e^{-\left(\rho x(\tau)+ \nu(\log x(\tau))'\right)}, \label{recov1}\\
(ii) \ & \gamma_k(\tau(y,y_0))= \nu e^{-\rho x(\tau)} \mathds{1}_{\left\{f_k(\tau)<0\right\}}, \label{recov2}\\
(iii) \ & \gamma_k(y,y_0)=  \delta(y-y^*(y_0)), \label{recov3}
\end{align}
where $\rho, \nu$ are factors determining the distribution of recovered hosts with respect to infection dynamics and $f_k(\tau)=f_k(x(\tau),y(\tau),z(\tau))$ is the pathogen growth rate.  For (ii), note that $\gamma$ is zero prior to the critical time when the pathogen begins to decrease (when $f_k(\tau)=0$).  For (iii), $y^*$ corresponds to a constant antibody level dependent on initial conditions at which infected hosts recover.  Note that this case covers the situation where the virus must decline below a fixed threshold, $x^*$, which can be related to $y^*$ via the inverse map.  Also for (iii), the infectious period for strain-k infected hosts with initial  level $y_0$ is given by $T(y_0)=\frac{1}{\int_{y_0}^{y^*(y_0)} g_k(y,y_0)dy}$.   

The host to vector transmission rate also depends upon the within-host infection dynamics.  
Data suggest that the probability of an mosquito getting infecting by a bite from an infected individual is a Holling type III function of the pathogen load at a given time-since-infection $\tau$ (\cite{handel2015crossing}, \cite{nguyen2013host},\cite{tuncer2016structural}). Thus, the form of host to vector transmission rate utilized is as follows;
\begin{equation} \label{transmis}
\beta_{\ell}(y,y_0)= \psi \frac{(x(\tau(y,y_0)))^2}{C_{\ell} + (x(\tau(y,y_0)))^2},  \quad  \ell\in \left\{1,2,12,21\right\}
\end{equation}
where  $C_{\ell},$ and $\psi$ are half saturation and transmission constants.  We utilize data of DENV-1 from \cite{nguyen2013host} in order to parameterize the half-saturation constant $C:=C_{\ell}$.  In particular, although we do not concern about the scale of viral load in our simulations (e.g. Fig. \ref{fig1a}), $C$ is chosen so that the ratio of peak viral load and viral load causing $50\%$ infectiousness does match the data in \cite{nguyen2013host}. Both the transmission rate, $\beta_{\ell}(\tau(y,y_0))$, and probability of exiting infectious period (through recovery or death), given by $1-\pi_{\ell}(\tau(y,y_0))$ (defined later by formula \eqref{piF}), are simulated in Fig. \ref{fig2a}.

 \begin{threeparttable}[t]
\caption{Description of linking parameters in \eqref{recov1}, \eqref{recov1}, and \eqref{transmis}.}
\centering 
\begin{tabularx}{\textwidth}{tbs}
\toprule
Parameter  & Description & Value  \\ [0.5ex]
\toprule
\\
$\psi$ & max host-vector transmission rate &  $0.001$   \tnote{a}  \\[0.5ex]
$C$ & host-vector transmission saturation constant & $135$ \tnote{a} \\[0.5ex]
$\rho$ & host recovery shape parameter for viral load & 10\\[0.5 ex]
$\nu$ & host recovery shape parameter for viral load slope &  10 \\[0.5ex]
$\xi$  & antibody waning rate  & 0.002 \tnote{b}\\[0.5ex]
$y_c$ & antibody lower bound & 0.02 \tnote{b} \\[0.5ex]
\bottomrule
\end{tabularx}
\label{table:linkparameters} 
 \begin{tablenotes}
     \item[a] Chosen to match data in \cite{nguyen2013host}.
     \item[b] Corresponds to approximately 2 years before recovered primary are susceptible to severe secondary infection.
   \end{tablenotes}
\end{threeparttable}

\begin{remark}\label{remark1}
Finally, we remark that the host population can be \emph{equivalently} structured by both antibody variables $y$ and $z$.  Then, for example, the following equation for $\tilde i_k:=\tilde i_k(t,y,z,y_0,z_0)$ would appear:
\begin{align}\label{host_epiG}
\frac{\partial \tilde i_k}{\partial t}+\frac{\partial (g_k \tilde i_k)}{\partial  y} +\frac{\partial (h_k \tilde i_k)}{\partial  z} 
&=  
     - \gamma_k\tilde i_k    \\
  (g_k(y_0,z_0,y_0,z_0)+h_k(y_0,z_0,y_0,z_0)) \tilde i_k(t,y_0,z_0,y_0,z_0) & = \beta_v^k(y_0,z_0)s(t,y_0,z_0) I_v^k(t)  \notag 
\end{align}
In such a model, $y$ might be interpreted as antibodies specific to strain 1 and $z$ as antibodies specific to strain 2.  In this way, the model affords flexibility in terms of how one defines specific versus non-specific antibodies.  While  tracking multiple antibody variables may seem to complicate matters, observe that there is a 1-1 relationship between $y$ and $z$, where the additional variable $z$ is mapped onto $y$ via the inverse map.  Thus we can utilize our original system \eqref{host_epi} (with additional ``static'' variable $z_0$) and simply calculate $z(y)$ for each cohort.  Note also that the infection-age $\tau$ is in 1-1 relationship with $y$, therefore we can transform the system to an age-structured model as done in \cite{gandolfi2015epidemic}.  However since we are interested in tracking antibody level in host population, we only pursue this direction when it can be advantageous for numerical simulations in the special case when susceptible antibody level ($\Lambda(y),s_0(y)$) is a Dirac point measure distribution (see Section \ref{SecEx}) .   
\end{remark}

\section{Analysis of between-host model} \label{secAnal}
The aims of this paper are \emph{multi-scale} model formulation,  and equilibrium, linearized stability and numerical analysis, with the goal of capturing Dengue \emph{ADE} across scales. We do not delve into establishing existence, uniqueness, regularity and positivity of
solutions of \eqref{host_epi}-\eqref{vector_epi}.   However, in a sequel to this paper, we will analyze the uniform persistence of solutions, which will require rigorous proof of model well-posedness.  Thus we reserve such questions addressing existence, regularity and global properties of solutions for our follow-up study. 
We do remark here that use of abstract semigroup theory \cite{thiemea1990integrated} or transformation to an age-structured model as noted in Remark \ref{remark1} combined with fixed point techniques applied to an integral form of the system, as in \cite{browne2013global}, can yield existence and uniqueness
results.  In both approaches, assuming antibody-dependent rates, e.g. $\Lambda(y)$, to be $L^{\infty}_+(0,\infty)$ should be sufficient for finding unique solutions in an appropriate product space consisting of $L^{1}_+(0,\infty)$ components of non-negative Lebesgue functions defined on $(0,\infty)$.  The main challenge, different from most structured population models but similar to \cite{barbarossa2015immuno}, is to properly control the evolution of the recovered distributions, $r_k(t,y)$, through waning (decreasing) antibody level $y$ at rates $\omega_k(y)$.  We conjecture that the proportional waning rate (exponential decay of $y$ with lower bound $y_c$), given by \eqref{expwane}, will ensure $r_k(t,y)$ remains in $L^{1}_+(y_c,\infty)$, as we show for the equilibrium solutions in Section \ref{SecEx}.   Restriction to Sobolev spaces can further strengthen linearized stability results obtained later in this current section.    Also, while considering the more general setting of measures \cite{carrillo2012structured} can allow for Dirac delta measures for certain rates and components, it may be easier to transform the model to equivalent ordinary differential or delay equations in the specific examples where we utilize Dirac delta point measures.  

To begin our analysis, define the total vector population, $N_v(t)=S_v(t)+\sum\limits_{k=1}^2 I_v^k(t)$, and total host population as
\begin{align*}
N_h(t)&=\int\limits_0^\infty s(t,y) \, dy + \sum\limits_{\substack{k,j=1 \\ k\neq j}}^2 \left( \int\limits_{0}^{\infty}\int\limits_{y_0}^{\infty} i_k(t,y,y_0) \,dy \,dy_0 +\int\limits_0^\infty r_k(t,y)\,dy+\int\limits_{0}^{\infty}\int\limits_{y_0}^{\infty} i_{kj}(t,y,y_0) \,dy \,dy_0\right).
\end{align*}
\begin{proposition}\label{bounded}
Solutions to the system \eqref{host_epi}-\eqref{vector_epi} remain bounded in forward time, and moreover $$\limsup\limits_{t\rightarrow\infty} N_h(t)=\frac{\int_0^\infty \Lambda(y)dy}{\mu}, \quad  \limsup\limits_{t\rightarrow\infty} N_v(t)=\frac{\Lambda_v}{\mu_v}. $$
\end{proposition}
\begin{proof}
Consider the differential equation satisfied by $N_h(t)$ derived from \eqref{host_epi}:
\begin{align*}
\frac{dN_h(t)}{dt}&=\int\limits_0^\infty \Lambda(y)dy -\mu N_h(t)- \int\limits_0^\infty s(t,y)\sum_{k}\beta_{v}^k(y)I_v^k(t)dy \\
& \qquad +\sum_k \int\limits_{0}^{\infty}\int\limits_{y_0}^{\infty}\left(-\frac{\partial (g_k(y,y_0) i_k(t,y,y_0))}{\partial  y} 
-(\mu+ \gamma_k(y,y_0)i_k(t,y,y_0)) \right) dydy_0 \\
& \qquad +\sum_k \int\limits_{0}^{\infty}\left(\frac{\partial (\omega_k(y) r_k(t,y))}{\partial  y} +\int_0^{\infty}\gamma_k(y, y_0)i_k(t,y,y_0)dy_0
-\mu r_k(t,y)- \beta_v^j(y) r_k(t,y)I_v^j(t) \right) dy \\
& \qquad +\sum_{j\neq k} \int\limits_{0}^{\infty}\int\limits_{y_0}^{\infty}\left(-\frac{\partial (g_{kj}(y, y_0) i_{kj}(t,y, y_0))}{\partial  y} 
- (\gamma_{kj}(y, y_0)+\mu)\right) i_{kj}(t,y, y_0) dydy_0 \\
& \qquad +\sum_{j\neq k} \int\limits_{0}^{\infty}\left( \frac{\partial (\omega_{kj}(y) r_{kj}(t,y))}{\partial y} +\int_0^{\infty}\gamma_{kj}(y, y_0)i_{kj}(t,y,y_0)dy_0-\mu r_{kj}(t,y) \right) dy
\end{align*}
Applying the fundamental theorem of calculus, along with the boundary conditions of \eqref{host_epi} (noting for example that $\omega_k(y) r_k(t,y) =0$ for $y\leq y_c\geq 0$ and all the populations decay to zero as $y\rightarrow\infty$), several cancellations occur and we simply obtain
\begin{align*}
\frac{dN_h(t)}{dt}&=\int\limits_0^\infty \Lambda(y)dy -\mu N_h(t).
\end{align*}
Similar conclusion holds for $N_v(t)$ and the result follows.
\end{proof}

\subsection{Reproduction Number \& Dynamical Properties of the Nested Systems}
Define the following quantities giving the probability of host recovery and host-vector force of infection as functions of antibody levels:
\begin{align}
\pi_k(y,y_0)=e^{-\int_{y_0}^{y}\frac{\gamma(a,y_0)+\mu}{g_k(a,y_0)}\,da}, \qquad \mathcal F_k(y_0) = \int\limits_{y_0}^{\infty}\frac{\pi_k(y,y_0)}{g_k(y,y_0)}\beta_k(y,y_0)\,dy \label{piF}
\end{align}
Then the \textit{antibody level dependent} basic reproduction number for each strain is given by:
     \begin{align}
     \mathcal R_0^k = \D\frac{ N_v}{\mu\mu_v} \int\limits_{0}^{\infty}\Lambda(y_0)\beta_{v}^k(y_0) \mathcal F_k(y_0)\,dy_0 ,  \label{R0f}
     \end{align}
where  $N_v =\D\frac{\Lambda_v}{\mu_v}$.  Note that host-vector force of infection $\mathcal F_k(y_0)$ may not mimic infection severity measured by peak viral load within-host, as shown in Fig. \ref{fig2b} where $\mathcal F_k(y_0)$ achieves maximum levels for smaller magnitudes of pre-existent antibody level $y_0$ than infection severity (due to decreasing infectious period associated with ADE).  
 
 The following proposition shows that the basic reproduction number $\mathcal R_0^k$ is a threshold for disease extinction (locally) and existence of an endemic equilibrium. 
\begin{proposition} \label{disfree_prop}
If $\mathcal R_0^k<1, \ k=1,2$, then the DFE $\mathcal E_0$ is locally asymptotically stable.  If $\mathcal R_0^k>1,$ there exists a single strain boundary equilibrium $\mathcal E_k$ and $\mathcal E_0$ is unstable.  
\end{proposition}
\begin{proof}
For the single-strain equilibrium $\mathcal E_k$, we derive the equilibrium equation for $\bar I_v^k$:
$$\frac{\mu_v}{N_v-\bar I_v^k}= \int\limits_{0}^{\infty}\frac{\Lambda(y_0)}{\beta_{v}^k(y_0) \bar I_v^k+\mu}\beta_{v}^k(y_0) \mathcal F_k(y_0)\,dy_0 $$
Let the right-hand side of the above equation be denoted $F(\bar I_v^k)$ and the left-hand side $G(\bar I_v^k)$.  Then $G$ is increasing on $[0,N_v)$, approaching $+\infty$ as $\bar I_v^k$ approaches $N_v$ from the left, and $F$ is decreasing on $[0,N_v)$.  Note that $F(0)>G(0) \Leftrightarrow \mathcal R_0^k>1$.  Thus there exists a positive equilibrium value $\bar I_v^k$ if and only if $\mathcal R_0^k>1$.  

Now for the stability of $\mathcal E_0$, we will consider the linearized equation for deviations of solutions: $\tilde s(t,y)=s(t,y)-\bar s(y), \tilde i_k(t,y,y_0),   \tilde r_k(t,y), \tilde i_{jk}(t,y,y_0), \tilde S_v(t)=S_v(t)-\bar S_v, \tilde I_v^k(t)$.  After discarding higher order terms in \eqref{host_epi}-\eqref{vector_epi}, we obtain the following linearized equations:

\begin{align}\label{host_epi_lin}
\frac{\partial \tilde s(t,y)}{\partial t} &= -\mu \tilde s(t,y)-\bar s(y)\sum_{k}\beta_{v}^k(y) \tilde I_v^k(t)   \notag \\
\frac{\partial \tilde i_k(t,y,y_0)}{\partial t}+\frac{\partial (g_k(y,y_0) \tilde i_k(t,y,y_0))}{\partial  y} 
&=  
     - \left(\gamma_k(y,y_0)+\mu\right)\tilde i_k(t,y,y_0)  \notag  \\
  g_k(y_0,y_0) \tilde i_k(t,y_0,y_0) & = \beta_v^k(y_0)\bar s(y)\tilde I_v^k(t)  \notag \\
\frac{\partial \tilde r_k(t,y)}{\partial t}- \frac{\partial (\omega_k(y)\tilde r_k(t,y))}{\partial y} &=\int_0^{\infty}\gamma_k(y, y_0)\tilde i_k(t,y,y_0)dy_0-\mu \tilde r_k(t,y) \\
\frac{\partial\tilde i_{kj}(t,y,y_0)}{\partial t}+\frac{\partial (g_{kj}(y, y_0)\tilde i_{kj}(t,y, y_0))}{\partial y} 
&= - \left(\gamma_{kj}(y, y_0)+\mu\right)\tilde i_{kj}(t,y, y_0), \quad \tilde i_{kj}(t,y_0, y_0)=0  \notag 
\end{align}

\begin{align}\label{vector_epi_lin}
\frac{d\tilde S_v}{dt} &= -\mu_v\tilde S_v- N_v\left(\sum\limits_{k=1}^2\int\limits_{0}^{\infty}\int\limits_{y_0}^{\infty}\beta_k(y,y_0) \tilde i_k(t,y,y_0) \,dy \,dy_0+\sum\limits_{\substack{k,j=1 \\ k\neq j}}^2\int\limits_{0}^{\infty}\int\limits_{y_0}^{\infty}\beta_{kj}(y,y_0) \tilde i_{kj}(t,y,y_0) \,dy \,dy_0 \right)  \notag \\
\frac{d\tilde I_v^k}{dt} &= N_v\left(\int\limits_{0}^{\infty}\int\limits_{y_0}^{\infty}\beta_k(y,y_0) \tilde i_k(t,y,y_0) \,dy \,dy_0 + \int\limits_{y_0}^{\infty}\int\limits_{y_0}^{\infty}\beta_{jk}(y,y_0) \tilde i_{jk}(t,y,y_0) \,dy \,dy_0 \right)- \mu_v \tilde I_v^k
\end{align}

We assume exponential form of the deviations of solutions from $\mathcal E_0$ (using separation of variables for the PDE's), and thus insert the following variables into the linearized system:  $\tilde s(t,y)=\hat s(y)e^{\lambda t}, i_k(t,y,y_0)=\hat i_k(y,y_0)e^{\lambda t},   r_k(t,y)=\hat r_k(y)e^{\lambda t}, i_{jk}(t,y,y_0)=\hat i_{jk}(y,y_0)e^{\lambda t}, S_v(t)=\hat S_ve^{\lambda t}, I_v^k=\hat I_v^ke^{\lambda t}$.  After some simplification, we arrive at the following equations for $\lambda\in\mathbb C$ and $ \hat S_v,\hat I_v^1, \hat I_v^2 \in \mathbb R_+$: 
\begin{align}
\lambda \hat S_v&=-\mu_v\tilde S_v- N_v\left(\sum\limits_{k=1}^2\int\limits_{0}^{\infty} \beta_v^k(y_0)\bar s(y_0)\hat I_v^k\mathcal L(y_0)\br{\varphi_1}(\lambda)dy_0\right) \label{eigEq0} \\
\lambda \hat I_v^k&=-\mu_v \hat I_v^k +N_v\int\limits_{0}^{\infty} \beta_v^k(y_0)\bar s(y_0)\hat I_v^k\mathcal L(y_0)\br{\varphi_1}(\lambda)dy_0 \notag 
\end{align}
where $\varphi_k(y,y_0)=\frac{\pi_k(y,y_0)}{g_k(y,y_0)}\beta_k(y,y_0)$, $\mathcal L(y_0)\br{\cdot}(\lambda)$ denotes Laplace transform (with additional variable $y_0$).  The $\hat I_v^k$ equation above yields the characteristic equation:
\begin{align}\label{chareqnk}
1 &=\frac{1}{\lambda+\mu_v} \frac{N_v}{\mu}\int\limits_{0}^{\infty} \beta_v^k(y_0)\Lambda (y_0)\mathcal L(y_0)\br{\varphi_1}(\lambda)dy_0:=\Psi_k(\lambda)
\end{align}
Then $\Psi_k(0)=\mathcal R_0^k$ and $\lim\limits_{\lambda\rightarrow\infty} \Psi_k(\lambda)=0$ for $\lambda\in\mathbb R$. Thus we readily infer that if $\mathcal R_{0}^k>1$, then $\mathcal E_0$ is unstable since there exists eigenvalue $\lambda>0$ corresponding to eigenvector with $\hat I_v^k>0$.  On the other hand suppose that  $\mathcal R_{0}^k<1, \ k=1,2$.  Suppose by way of contradiction that there exists an eigenvalue with non-negative real part; $\lambda=a+bi, \ a\geq 0$.  Cleary that can not happen if $\hat I_v^k=0, \ k=1,2$, since in that case $\lambda=-\mu_v$.  So assume $\hat I_v^k>0$.  Then taking modulus of \eqref{chareqnk}, we find that
\begin{align*}
1=|\Psi_k(\lambda)| & \leq \frac{1}{\lambda+\mu_v} \frac{N_v}{\mu}\int\limits_{0}^{\infty} \beta_v^k(y_0)\Lambda (y_0)|\mathcal L(y_0)\br{\varphi_1}(\lambda)|dy_0\leq \Psi_k(0) = \mathcal R_0^k<1,
\end{align*}
which gives a contradiction.  
\end{proof}
In addition to existence of unique single strain equilibria, $\mathcal E_k$, when $\mathcal R_0^k>1,$ when $\beta_v^k$ is constant (does not depend on host antibody level $y_0$),  $\mathcal E_k$ can be explicitly found to have the following positive components:
\begin{align}
\bar I_v^k&=\frac{\mu_v\mu}{\beta_v^k(\mu_v+\int_{0}^{\infty}\Lambda(y_0)\mathcal F_k(y_0)\,dy_0)} \left(\mathcal R_0^k-1\right), \quad \bar S_v^k=N_v-\bar I_v^k, \quad \bar s(y)=\frac{\Lambda (y)}{\mu+\beta_v^k\bar I_v^k}, \notag \\
\bar i_k(y,y_0)&=\beta_v^k\bar s(y_0)\bar I_v^k \frac{\pi_k(y,y_0)}{g_k(y,y_0)}, \quad \bar r_k(y)=\frac{1}{\omega_k(y)}\int\limits_y^{\infty}e^{\int\limits_a^y \frac{\mu}{\omega_k(s)}ds} \int\limits_0^{\infty} \gamma_k(a,y_0) \bar i_k(a,y_0) dy_0 da. \label{ssEq}
\end{align}

Next, we define the following host to vector ``force of infection'' quantities with respect to primary and secondary infections, respectively:
 \begin{align} 
 \mathcal G_k&=\int\limits_{0}^{\infty}\Lambda(y_0) \mathcal F_k(y_0)\,dy_0, \qquad \text{if} \ \beta_v^k \ \text{constant}, \ \mathcal G_k=\frac{\mu\mu_v}{N_v\beta_v^k} \mathcal R_0^k,\label{invG} \\
 \mathcal H_{kj}&= \int\limits_{0}^{\infty}\frac{1}{\omega_k(y_0)}\int\limits_{y_0}^{\infty}e^{\int\limits_a^{y_0} \frac{\mu}{\omega_k(s)}ds} \int\limits_0^{\infty} \gamma_k(a,z) \Lambda(z) \frac{\pi_k(y_0,z)}{g_k(y_0,z)}dz da \int\limits_{y_0}^{\infty}\beta_{kj}(y,y_0)\frac{\pi_{kj}(y,y_0)}{g_{kj}(y,y_0)} \,dy\,dy_0,  \label{invH}
 \end{align}
Note that in the absence of waning, i.e. $\omega_k\equiv 0$, then $\bar r_k(y)$ is proportional to the probability density corresponding to the antibody concentration after primary infection: 
\begin{align}
\bar r_k(y)&=\int\limits_{y_0}^{\infty}\frac{\gamma_k(y,y_0)\bar i_k(y,y_0)}{\mu} dy_0, \notag\\
 \mathcal H_{kj}&= \int\limits_{0}^{\infty}\int\limits_{y_0}^{\infty} \frac{\gamma_k(y_0,z)}{\mu} \Lambda(z) \frac{\pi_k(y_0,z)}{g_k(y_0,z)}dz  \int\limits_{y_0}^{\infty}\beta_{kj}(y,y_0)\frac{\pi_{kj}(y,y_0)}{g_{kj}(y,y_0)} \,dy\,dy_0.  \label{invHnoW}
 \end{align}

The invasion reproduction number for strain $j$ invading strain $k$ is the following: 
\begin{align} \label{invR0}
\mathcal R_{inv}^j&= \frac{\mathcal R_0^j}{\mathcal R_0^k} +\frac{1}{\mu_v}\int\limits_{0}^{\infty}\beta_v^j(y_0)\bar r_k(y_0)\int\limits_{y_0}^{\infty} \beta_{kj}(y,y_0)\frac{\pi_{kj}(y,y_0)}{g_{kj}(y,y_0)} dy \,dy_0 
\end{align}
Plugging in equilibrium value of $\bar r_k(y)$ to \eqref{invR0}, for the case of constant vector to host transmission rate $\beta_v^k$, we obtain $$\mathcal R_{inv}^j= \frac{\mathcal R_0^j}{\mathcal R_0^k} +\frac{\mu_v\beta_v^j\left(\mathcal R_0^k-1\right)}{\mathcal R_0^k\mu_v+\mathcal G_k} \mathcal H_{kj}. $$

\begin{theorem}\label{invThm} Consider the case that $\beta_v^k$ is constant (does not depend on host antibody level $y_0$).  Let $j,k\in \left\{1,2\right\}$, $j\neq k$ and suppose that $\mathcal R_0^k>1$.  If $\mathcal R_{inv}^j<1$, then $\mathcal E_1^k$ is locally asymptotically stable.    If $\mathcal R_{inv}^j>1$, then $\mathcal E_1^k$ is unstable.
\end{theorem}
\begin{proof}
Without loss of generality, let $k=1, j=2$.  The linearized PDE system solved with separation of variables, similar to \eqref{eigEq0} reduces to the following equations for $\lambda\in\mathbb C$ and $ \hat S_v,\hat I_v^1, \hat I_v^2 \in \mathbb R_+$: 
\begin{align}
\lambda \hat S_v&=-\left(\mu_v + \bar I_v^1\int\limits_0^{\infty}\bar s(y)\beta_v^1 \mathcal F_1(y) dy\right) \hat S_v  \label{eigEq} \\
& \qquad -\bar S_v\int\limits_0^{\infty}\left[\beta_v^1(\bar s(y_0)\hat I_v^1+ \hat s(y_0) \bar I_v^1)\mathcal L(y_0)\br{\varphi_1}(\lambda) +\beta_v^2\bar s(y_0)\hat I_v^2\mathcal L(y_0)\br{\varphi_2}(\lambda) \right]dy_0 \notag \\ 
& \qquad - \hat I_v^2 \int\limits_{0}^{\infty} \left(\beta_v^2\bar r_1(y_0)\mathcal L(y_0)\br{\varphi_{12}}(\lambda)  +\beta_v^1\hat r_2(y_0)\mathcal L(y_0)\br{\varphi_{21}}(\lambda) \right)\,dy_0   \notag \\
\lambda \hat I_v^1&=-\mu_v \hat I_v^1 +\bar S_v\int\limits_0^{\infty}\beta_v^1(\bar s(y_0)\hat I_v^1+ \hat s(y_0) \bar I_v^1)\mathcal L(y_0)\br{\varphi_1}(\lambda)\,dy_0 \notag \\ 
& \qquad +\hat I_v^2 \int\limits_{0}^{\infty} \beta_v^1\hat r_2(y_0)\mathcal L(y_0)\br{\varphi_{21}}(\lambda)\,dy_0 + \bar I_v^1\int\limits_0^{\infty}\bar s(y)\beta_v^1 \mathcal F_1(y) dy \notag \\ 
\lambda \hat I_v^2&=-\mu_v \hat I_v^2 +\bar S_v\int\limits_0^{\infty}\beta_v^2\bar s(y_0)\hat I_v^2\mathcal L(y_0)\br{\varphi_2}(\lambda) \,dy_0+\hat I_v^2 \int\limits_{0}^{\infty} \beta_v^2\bar r_1(y_0)\mathcal L(y_0)\br{\varphi_{12}}(\lambda)  \,dy_0 \label{eigEq2}
\end{align}
where $\varphi_k(y,y_0)=\frac{\pi_k(y,y_0)}{g_k(y,y_0)}\beta_k(y,y_0)$, $\mathcal L(y_0)\br{\cdot}(\lambda)$ denotes Laplace transform (with additional variable $y_0$), and $$\hat s(y)=\frac{\Lambda (y)}{\lambda+\mu+\beta_v^1\bar I_v^1+\sum_k\beta_v^k\hat I_v^k}, \quad \hat r_2(y)=\frac{1}{\omega_k(y)}\int\limits_y^{\infty}e^{\int\limits_a^y \frac{\lambda+\mu+\beta_v^1\bar I_v^1}{\omega_k(s)}ds}\int\limits_{0}^{\infty}\beta_v^2\bar s(y_0)\hat I_v^2\mathcal L(y_0)\br{\varphi_2}(\lambda)dy_0da.$$ 
Assume that $\hat I_v^2>0$.  Upon plugging in equilibrium values, we use that 
$$\bar s(y_0)\bar S_v=\frac{\mu_v\Lambda(y_0)}{\beta_v^1\int_0^\infty \Lambda(y)\mathcal F_1(y) dy}=\frac{N_v\Lambda(y_0)}{\mathcal R_0^1\mu} \ .$$ The $I_v^2$ equation \eqref{eigEq2} becomes 
\begin{align}\label{chareqn0}
1 &=\frac{1}{\lambda+\mu_v} \left[\frac{N_v}{\mathcal R_0^1\mu} \int\limits_0^{\infty}\beta_v^2\Lambda(y_0)\mathcal L(y_0)\br{\varphi_2}(\lambda) \,dy_0+ \int\limits_{0}^{\infty} \beta_v^2\bar r_1(y_0)\mathcal L(y_0)\br{\varphi_{12}}(\lambda)  \,dy_0 \right]
\end{align}
This yields the characteristic equation $1=G(\lambda)$ for an eigenvalue $\lambda$ where $G(\lambda)$ is the right hand side of \eqref{chareqn0}. 
Note $G(0)=\mathcal R_{inv}^2$ and $\lim\limits_{\lambda\rightarrow\infty} G(\lambda)=0$ for $\lambda\in\mathbb R$.  Thus we readily infer that if $\mathcal R_{inv}^2>1$, then $\mathcal E_1^1$ is unstable since there exists eigenvalue $\lambda>0$ corresponding to eigenvector with $\left(\hat S_v,\hat I_v^1, \hat I_v^2\right)=(0,0,1)$.  

Now consider the case where $\mathcal R_{inv}^2<1$.  We claim that $\mathcal E_1^1$ is locally asymptotically stable. First consider the case $\hat I_v^2>0$.  Suppose by way of contradiction that there exists an eigenvalue with non-negative real part; $\lambda=a+bi, \ a\geq 0$.  Then taking modulus of \eqref{chareqn0}, we find that
\begin{align*}
1=|G(\lambda)| & \leq \frac{1}{\mu_v} \left[\frac{N_v}{\mathcal R_0^1\mu} \int\limits_0^{\infty}\beta_v^2\Lambda(y_0)|\mathcal L(y_0)\br{\varphi_2}(\lambda)| \,dy_0+ \int\limits_{0}^{\infty} \beta_v^2\bar r_1(y_0)|\mathcal L(y_0)\br{\varphi_{12}}(\lambda)|  \,dy_0 \right] \\
&\leq G(0) = \mathcal R_{inv}^2<1,
\end{align*}
which gives a contradiction.  

Next consider the case where $\hat I_v^2=0$.  This also implies that $\hat r_2=0$.  Notice that all of terms referring to strain 2 are now zero in \eqref{eigEq}, thus we drop the superscript referring to strain 1.  Adding two vector equations in \eqref{eigEq}, we obtain $\lambda (\hat S_v +\hat I_v)=-\mu_v (\hat S_v +\hat I_v)$.  Therefore if $\hat S_v \neq -\hat I_v$, then $\lambda=-\mu_v$.  So consider the case $\hat S_v = -\hat I_v$.  Suppose by way of contradiction that there exists an eigenvalue with non-negative real part; $\lambda=a+bi, \ a\geq 0$.  Then after substitution and cancellation involving the equations of \eqref{eigEq}, we obtain
\begin{align*}
\left(\lambda+\mu_v+\beta_v\bar I_v\int\limits_0^\infty\mathcal F(y)\bar s(y)dy\right)(\lambda +\beta_v\bar I_v+\mu)=(\lambda+\mu)& \beta_v\bar S_v \int\limits_0^\infty\bar s(y)\mathcal L(y)\br{\varphi}(\lambda)dy \\
\Rightarrow 1 < \frac{|\lambda +\beta_v\bar I_v+\mu|}{|\lambda +\mu|}=\frac{|\beta_v\bar S_v \int\limits_0^\infty\bar s(y)\mathcal L(y)\br{\varphi}(\lambda)dy)|}{|\lambda+\mu_v+\mathcal F\beta_v\bar S\bar I_v|} &\leq  \frac{\beta_v\int\limits_0^\infty\mathcal F(y)\bar s(y)\bar S_v dy }{\mu_v+\int\limits_0^\infty\mathcal F(y)\bar s(y)dy\beta_v\bar I_v}\\
&=\frac{\mu_v }{\mu_v+\int\limits_0^\infty\mathcal F(y)\bar s(y)dy\beta_v\bar I_v} <1
\end{align*}
This yields a contradiction.  
\end{proof}

 \subsection{Coexistence equilibrium} 
 The complexity of the model challenges explicit formulation and conditions for a coexistence equilibrium.  However, general equations of two variables for coexistence equilibria can be derived, which reduced to a single variable equation in the case of symmetric strains.  Furthermore, when there is no waning ($\omega\equiv 0$), a quadratic equation determines coexistence equilibria.  Thus we first consider the case of no waning.
 
 \emph{Case: No waning} \\ Recall the host to vector \emph{force of infection} quantities $\mathcal G_{k}$ and $\mathcal H_{kj}$ given by \eqref{invG} and \eqref{invHnoW}.
Let $\bar x_k=\beta_v \bar I_v^k$.  Then the following equations for $\bar x_k, \ k=1,2$ can be derived for a coexistence equilibrium (denoted by $\mathcal E_c$):
 \begin{align} 
 \frac{\mu_v}{\beta_v^k}&=\left(\frac{\bar x_j\mathcal H_{kj}}{\mu+\bar x_j} + \frac{\mathcal G_k}{\mu+\bar x_1+\bar x_2} \right) \left(N_v-\frac{\bar x_1}{\beta_v^1}-\frac{\bar x_2}{\beta_v^2}\right), \quad j\neq k, \label{genCoexEq} \\
 \Rightarrow &  \frac{\bar x_2\mathcal H_{21}}{\mu+\bar x_1}+\frac{\mathcal G_1}{\mu+\bar x_1+\bar x_2}= \frac{\bar x_1\mathcal H_{12}}{\mu+\bar x_2}+\frac{\mathcal G_2}{\mu+\bar x_1+\bar x_2}\label{genCoexEq2}
 \end{align} 

Indeed obtain the following equations for $\mathcal E_c$ from the model:
 \begin{align*}
 \bar S_v^k&=N_v-\sum_k \bar x_k, \quad \bar s(y)=\frac{\Lambda (y)}{\mu+\sum_k \bar x_k}, \\
 \bar i_k(y,y_0)&=\bar s(y_0)\bar x_k \frac{\pi_k(y,y_0)}{g_k(y,y_0)}, \quad \bar r_k(y)=\frac{x_k}{\mu+x_j}\int_{0}^{\infty}\gamma_k(y,y_0)\bar s(y_0) \frac{\pi_k(y,y_0)}{g_k(y,y_0)}dy_0, \\
\bar i_{kj}(y,y_0)&=\beta_v\bar r_k(y_0)\bar I_v^j \frac{\pi_{kj}(y,y_0)}{g_{kj}(y,y_0)}, \\
\frac{\mu_v\bar x_k}{\beta_v^k}&= \bar S_v\bar x_k\int\limits_{0}^{\infty}\mathcal F_k(y_0)\bar s(y_0) \,dy_0 + \bar S_v \frac{\bar x_k\bar x_j}{\mu+\bar x_k}\int\limits_{0}^{\infty}\mathcal F_{jk}(y_0)\tilde s_k(y_0) \,dy_0, \\
\frac{\mu_v}{\beta_v^k}&= \left(N_v-\frac{\bar x_1}{\beta_v^1}-\frac{\bar x_2}{\beta_v^2}\right)\left(\int\limits_{0}^{\infty}\mathcal F_k(y_0)\Lambda(y_0) \,dy_0 + \frac{(\mu+\bar x_1+\bar x_2)\bar x_j}{\mu+\bar x_k}\int\limits_{0}^{\infty}\mathcal F_{jk}(y_0)\tilde s_j(y_0) \,dy_0\right),
\end{align*}
where $$\mathcal F_{kj}(y_0)= \int\limits_{y_0}^{\infty}\frac{\pi_{kj}(y,y_0)}{g_{kj}(y,y_0)}\beta_{kj}(y,y_0)\,dy, \quad \tilde s_k(y_0)=\int\limits_{0}^{\infty}\frac{\pi_k(y_0,z)}{g_k(y_0,z)}\Lambda(z)\gamma_k(y_0,z)\,dz. $$
The equations \eqref{genCoexEq} and \eqref{genCoexEq2} follow from the definition of $\mathcal G_{k}$ and $\mathcal H_{kj}$.

 Under general parameters, the equations for the coexistence equilibrium yield a quadratic equation for $\bar x_2$ in terms of $\bar x_1$, and therefore it does not reduce to a polynomial equation in the single variable $\bar x_1$.  Thus, for tractability, we consider the case of identical strains and constant vector to host transmission rate, i.e. $\beta_v=\beta_{v}^1=\beta_{v}^2, \mathcal G=\mathcal G_1=\mathcal G_2, \mathcal H=\mathcal H_{12}=\mathcal H_{21}$ .    Then the equation \eqref{genCoexEq2} simplifies to
 $(\bar x_2-\bar x_1)(\mu+\bar x_1+\bar x_2)=0$, which implies that $\bar x_1=\bar x_2$.  From \eqref{genCoexEq}, we obtain the following quadratic equation for $\bar x= \bar x_k=\beta_v \bar I_v^k$:
 \begin{align}\label{CoexEq}
2\mu_v\left((\mathcal G+\mathcal H)+\mu_v\right)x^2+\left[\mu_v\mu(3\mu_v+2\mathcal G)-\Lambda_v\beta_v(\mathcal G+\mathcal H)\right]x+(\mu_v\mu)^2\left(1-\mathcal R_0\right) =0
 \end{align} 
We cannot rule out the existence two positive (subthreshold) coexistence equilibria when $\mathcal R_0 \ (:=\max_k(\mathcal R_0^k)) \ <1$, known as \emph{backward bifurcation} \cite{gulbudak2013forward}.  However, we can preclude existence of subthreshold equilibria if the second (linear) coefficient of the quadratic \eqref{CoexEq} is positive.   For instance, the following result follows from Proposition \ref{disfree_prop} and equation \eqref{CoexEq}:
\begin{proposition}\label{coex_prop}
Consider the case of symmetric strains, constant vector-host transmission and no waning, i.e.  $\mathcal R_0=\mathcal R_k$, $\beta_v^k(y)=\beta_v$ and $\omega_k\equiv 0$, $k=1,2$.  Furthermore, assume that $\mathcal H\leq\mathcal G$  and $\Lambda_v\beta_v<2\mu_v\mu$.  If $\mathcal R_0<1$, then there are no endemic equilibria.  If $\mathcal R_0>1$, then (in addition to existence of single-strain equilibria $\mathcal E_k$) there exists a unique coexistence equilibrium, $\mathcal E_c$.
\end{proposition}
Note that the condition $\mathcal H\leq\mathcal G$ in the hypothesis can be interpreted as ``secondary'' force of infection is less than primary.

A coexistence equilibrium under the condition of symmetric strains, denoted by $\mathcal E_c$, has the following components:
\begin{align}
\bar I_v^k &= \bar x/\beta_v, \qquad \text{where } \bar x \ \text{is root of \eqref{CoexEq}}, \notag \\
\bar S_v^k&=N_v-\bar I_v^k, \quad \bar s(y)=\frac{\Lambda (y)}{\mu+\beta_v\sum_k\bar I_v^k}, \quad \bar i_k(y,y_0)=\beta_v\bar s(y_0)\bar I_v^k \frac{\pi_k(y,y_0)}{g_k(y,y_0)}, \notag \\ \bar r_k(y)&=\int\limits_{0}^{\infty}\frac{\gamma_k(y,y_0)\bar i_k(y,y_0)}{\mu+\beta_v^j\bar I_v^j} dy_0,  \label{rk_coex}\\
\bar i_{kj}(y,y_0)&=\beta_v\bar r_k(y_0)\bar I_v^j \frac{\pi_{kj}(y,y_0)}{g_{kj}(y,y_0)}, \notag
\end{align}
where the quantities $\mathcal G_k$ and $\mathcal H_{kj}$, given by \eqref{invG} and \eqref{invHnoW}, are identical for strains $k,j$.  

 \emph{Case: Waning} \\
In the case of continuous waning, there is no analytical solution for the coexistence equilibria.  Indeed, even for symmetric strains, the equations become transcendental, shown below:
 \begin{align}
 \bar r_k(y)&=\frac{1}{\omega_k(y)}\int\limits_y^{\infty}e^{\int\limits_a^y \frac{\mu+\bar x_j}{\omega_k(s)}ds} \int\limits_0^{\infty} \gamma_k(a,y_0) \bar i_k(a,y_0) dy_0 da \notag  \\
 \frac{1}{\beta_v^k}&= \bar S_v\int\limits_{0}^{\infty}\mathcal F_k(y_0)\bar s(y_0) \,dy_0 + \bar S_v \bar x_j\int\limits_{0}^{\infty}\frac{1}{\omega_j(y_0)}\int\limits_{y_0}^{\infty}e^{\int\limits_a^{y_0} \frac{\mu+\bar x_k}{\omega_j(s)}ds}\mathcal F_{jk}(y_0)\tilde s_j(a) \,da \,dy_0   \notag \\
\frac{\mu_v}{\beta_v^k}&= \left(N_v-\frac{\bar x_1}{\beta_v^1}-\frac{\bar x_2}{\beta_v^2}\right)\left(\frac{\mathcal G_k}{\mu+\bar x_1+\bar x_2} +\bar x_j\int\limits_{0}^{\infty}\frac{\mathcal F_{jk}(y_0)}{\omega_j(y_0)}\int\limits_{y_0}^{\infty}e^{\int\limits_a^{y_0} \frac{\mu+\bar x_k}{\omega_j(s)}ds}\tilde s_j(a) \,da \,dy_0\right), \quad k\neq j. \label{eqForm}
\end{align}
In this case, we are also not able to prove that a coexistence equilibrium, $\mathcal E_c$, of symmetric strains must have equal components ($\bar x_1=\bar x_2$), as shown when there is no waning.  The above equations do allow for numerical approximation of roots, which we will perform for examples considered in Section \ref{SecEx} in the presence of waning.

\section{Numerical Scheme}\label{Numerical}
We develop a finite difference scheme combined with a ODE solver in MatLab in order to numerically solve the coupled immuno-epidemiological model.  To simulate the coupled system, first consider the within-host model ODE for relevant ranges of pre-existent antibody levels of susceptible hosts, $y_1\leq y_0 \leq y_{M_0}$, and time since infection $0\leq \tau\leq \tau_{end}$.  Here $y_0$ supplies the variable initial condition in within-host system \eqref{Secondary}, where other initial conditions $x_0,z_0$ are fixed.  Note that the distinct strains may have different parameter values in \eqref{Secondary}, in which case the ODE simulations must be conducted for each strain.  The mesh chosen for the interval $[y_1,y_{M_0}]$ can have equal or variable step size with $M_0$ mesh-points.  The output of the ODE solver is the solution vector, denoted here $\left(\phi(\tau_k;y_m)\right)_{k=1}^{N_0}$, where $\left(\tau_k\right)_{k=1}^{N_0}$ is a partition of $0\leq \tau\leq \tau_{end}$ for each $y_m, \ m=1,\dots, M_0$.  We utilize MatLab solver ODE45, which adaptively chooses the time partition and interpolates at time points $\left(\tau_k\right)_{k=1}^{N_0}$.  Consider  the state variable $y(\tau;y_m)$ giving antibody level $y$ during primary infection.  It is possible to consider the partitions $\left(y(\tau_k;y_m)\right)_{k=1}^{N_0}$ for each $m=1,\dots, M_0$. For small $M_0$, for instance the case of susceptible point distribution ($M_0=1$), this method of partitioning increases speed and is equivalent to transforming the infected host antibody level, $y$, to time since infection $\tau$, similar to the approach in \cite{gandolfi2015epidemic}.  However for $M_0>1$, the number of stored $y$ meshpoints will rapidly increase with $M_0$, and thus it is advantageous to utilize a ``global'' partition $\left( y_{\ell}\right)_{\ell=1}^{M}$ which contains as a sub-partition the initial mesh $y_1,\dots,y_{M_0}$ and covers all necessary stored antibody variables.  Then we can interpolate the pathogen and specific antibody, $x$ and $z$, as functions of $y$ onto this global mesh $\left( y_{\ell}\right)_{\ell=1}^{M}$, in order to compute linking functions $\gamma(y,y_0), \beta(y,y_0), g_k(y,y_0)$ at each (reachable) grid point $(y_{\ell},y_m)$, $1\leq \ell \leq M, 1\leq m \leq M_0$.  The same logic is utilized for secondary infection, along with numerical integration for reproduction numbers and equilibria values.

For the epidemiological model, we approximate solutions to the antibody-level structured PDE vector-host model with the stored within-host calculations.  Let $0\leq t\leq T$ be the time interval of interest and $\left\{t_n\right\}_{n=1}^N$ be a partition of $0\leq t\leq T$ with fixed time step $\de t=T/N$.  In the following, we denote the time iteration, $n$, in the superscript of state variables and the antibody levels as function arguments, e.g. $I_v^{k,n}$, $i^{n+1}_{k}(y_{\ell},y_{m})$.  For clarity, we utilize antibody variables $(y_{\ell},y_m)$ in function arguments, but note that the outer state variables and linking parameters are computed at (reachable) grid points in within-host part of numerical scheme.  The numerical algorithm for approximating solutions at times $t_n, \ n=1,\dots,N$, is the following:
 \begin{align*}
\mathcal I_k^n=\sum_{m=1}^{M_0} \sum_{\ell=1}^M &\beta_{k}(y_{\ell+1},y_{m+1})  i_{k}^n(y_{\ell+1},y_m) \de y_\ell \de y_m, \quad  \mathcal I_{j k}^n=\sum_{i=1}^M \sum_{\ell=1}^M \beta_{j k}(y_{i+1},y_{\ell+1})  i_{jk}^n(y_{i+1},y_{\ell+1})  \de y_i \de y_\ell,\\ S_v^{n+1}&=\frac{ S_v^{n}+\Lambda_v\de t}{1+\de t\big(\mu_v+\sum\limits_{\substack{k,j=1 \\ k\neq j}}^2 \mathcal I_k^n+\mathcal I_{jk}^n\big)},\qquad
I_v^{k,n+1}=\frac{ I_v^{k,n}+\de t S_v^{n+1}(\mathcal I_k^n+\mathcal I_{jk}^n)}{1+\de t \mu_v} \\
s^{n+1}(y_{m})&=\frac{s^{n}(y_{m})+\de t \Lambda(y_m)}{1+\de t\left(\sum\limits_{k=1}^2\beta_{v}^k(y_m) I_v^{k,n+1} + \mu \right)}, \\ i^{n+1}_k(y_m,y_m)&=\frac{\beta_{v}^k(y_m)}{g_{k}(y_{m},y_{m})} I_v^{k,n+1} s^{n+1}(y_{m}), \\
i^{n+1}_{k}(y_{\ell+1},y_{m})&=\frac{m^{n}_{k}(y_{\ell+1},y_{m+1})+\frac{\de t}{\de y_\ell}g_{k}(y_{\ell},y_{m})i^{n+1}_{k}(y_{\ell},y_{m})}{1+\de t\left(\frac{g_{k}(y_{\ell+1},y_{m})}{\de y_\ell}+\gamma_{k}(y_{\ell+1},y_{m})+ \mu \right)}, \\
  r^{n+1}_{k}(y_{M-\ell})&=\frac{ r^{n}_{k}(y_{M-\ell})+\de t\left(\frac{\omega_{k}(y_{M-\ell})}{\de y_{M-\ell}}r^{n+1} _{k}(y_{M-\ell+1})+\sum\limits_{i=1}^M \gamma_{k}(y_{M-\ell},y_i)i^{n+1}_{k}(y_{M-\ell},y_i)\de y_i \right)}{1+\de t\left(\frac{\omega_{k}(y_{M-\ell+1})}{\de y_{M-\ell}}+\beta_v^2(y_{M-\ell})I_v^{2,n+1}+ \mu \right)}, \\
  i^{n+1}_{jk}(y_{\ell},y_{\ell})&=\frac{\beta_{v}^k(y_{\ell})}{g_{jk}(y_{\ell},y_{\ell})} I_v^{k,n+1} r_j^{n+1}(y_{\ell}), \\ i^{n+1}_{jk}(y_{i+1},y_{\ell})&=\frac{i^{n}_{jk}(y_{i+1},y_{\ell+1})+\frac{\de t}{\de y_i}g_{jk}(y_{i},y_{\ell})i^{n+1}_{k}(y_{i},y_{\ell})}{1+\de t\left(\frac{g_{jk}(y_{i+1},y_{\ell})}{\de y_i}+\gamma_{jk}(y_{i+1},y_{\ell})+ \mu \right)}, 
\end{align*}
where $i,\ell,m$  denote index for discretized antibody level during secondary infection, during and after recovery of primary infection, and before primary infection, respectively, with $1\leq i\leq M, 1\leq \ell \leq M, 1\leq m \leq M_0$ being the ranges of allowed antibody level, and $j,k$ denote distinct strains.  Note that we utilize an implicit-explicit approximation mixture in the above finite difference scheme.  In particular, the calculation procedure allows for implicit terms involving components that have already been updated, thereby gaining advantages of an implicit form without having to pay the computational price of matrix inversion (which is common with implicit schemes).  We remark that an explicit-implicit mixture approach has been used for the approximation of size-structured models \cite{ackleh2014structured}. In the appendix, we present several numerical tests of rates of convergence for the algorithm described here.


\section{Epidemiological Implications} \label{SecEx}
In this section, we consider examples and potential implications for vaccination, utilizing formulas established in Section \ref{secAnal} and the numerical scheme developed in Section \ref{Numerical}.  First, we point out an important severity measure for dengue; the prevalence of DHF.
Define the measure of DHF in the population of strain-$k$ infected individuals by the following: \begin{align}\mathcal D_k(t)=\int\limits_{0}^{\infty}\int\limits_{y_0}^{\infty}\mathds{1}_{\left\{x(\tau)>V_c\right\}} \left(i_k(t,y,y_0)+i_{jk}(t,y,y_0)\right) \,dy \,dy_0=\int\limits_{\tilde y_l}^{\tilde y_u} \int\limits_{y_l}^{y_u} \left(i_k(t,y,y_0)+i_{jk}(t,y,y_0)\right) \,dy \,dy_0, \label{dhfC}
\end{align}
where the constant $V_c$ is a threshold critical lower bound such that if viral load during infection, $x(\tau)$, is above $V_c$, the patient will experience DHF.   The right hand side of the above equation reflects that this will translate into a particular range of antibody level, $y_l\leq y \leq y_u$, and a particular window of pre-existent antibody level, $\tilde y_l\leq  y_0 \leq \tilde y_u$, which will precipitate DHF.  A related measure which can be useful is the incidence of individuals who will experience DHF (by strain-$k$) at time $t$ predicated on their initial antibody level at time of infection:
\begin{align}
\mathcal I_k(t)=I_v(t)\int\limits_{\tilde y_l}^{\tilde y_u} \beta_v^k(y_0) \left(s(t,y_0)+r_{j}(t,y_0)\right) \,dy_0, \label{iik}
\end{align}
In the following examples, we perform numerical simulations utilizing the derived equilibria formulae and numerical scheme for solutions of the model.  In order to simplify the model for numerical validation, we assume symmetric strains.  We fix parameter values, except for waning rate ($\rho(y)$), vector-host transmission rate ($\beta_v(y)$) and susceptible recruitment rate ($\Lambda(y)$), in order to compare simulations for cases with or without waning, temporary cross-immunity, and a distribution of susceptible antibody level. 
\subsection{No Waning}
    \begin{figure}[t!]
\subfigure[][]{\label{fig3a}\includegraphics[width=7.35cm,height=4cm]{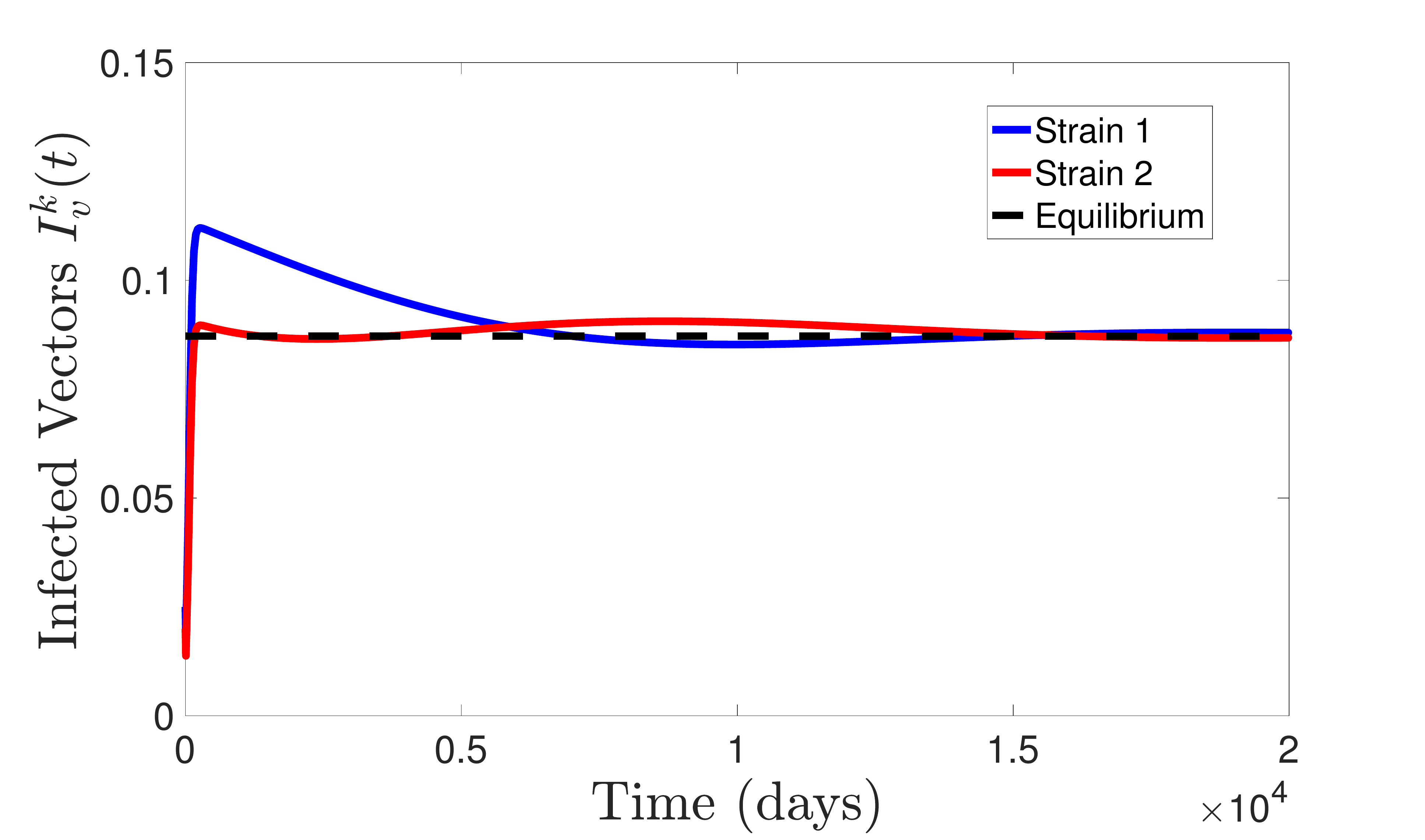}}
\subfigure[][]{\label{fig3b}\includegraphics[width=7.35cm,height=4cm]{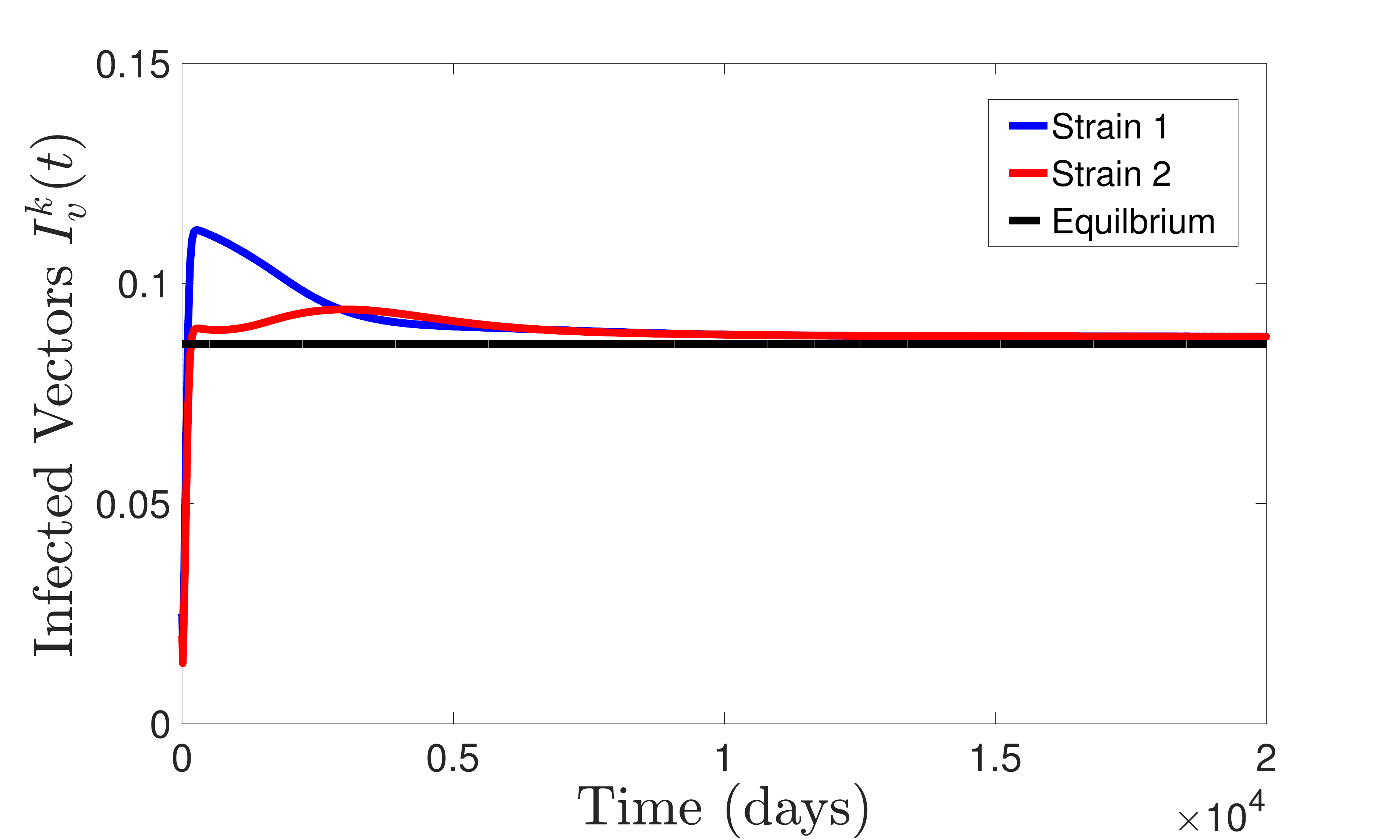}} \\
\subfigure[][]{\label{fig3c}\includegraphics[width=7.35cm,height=4cm]{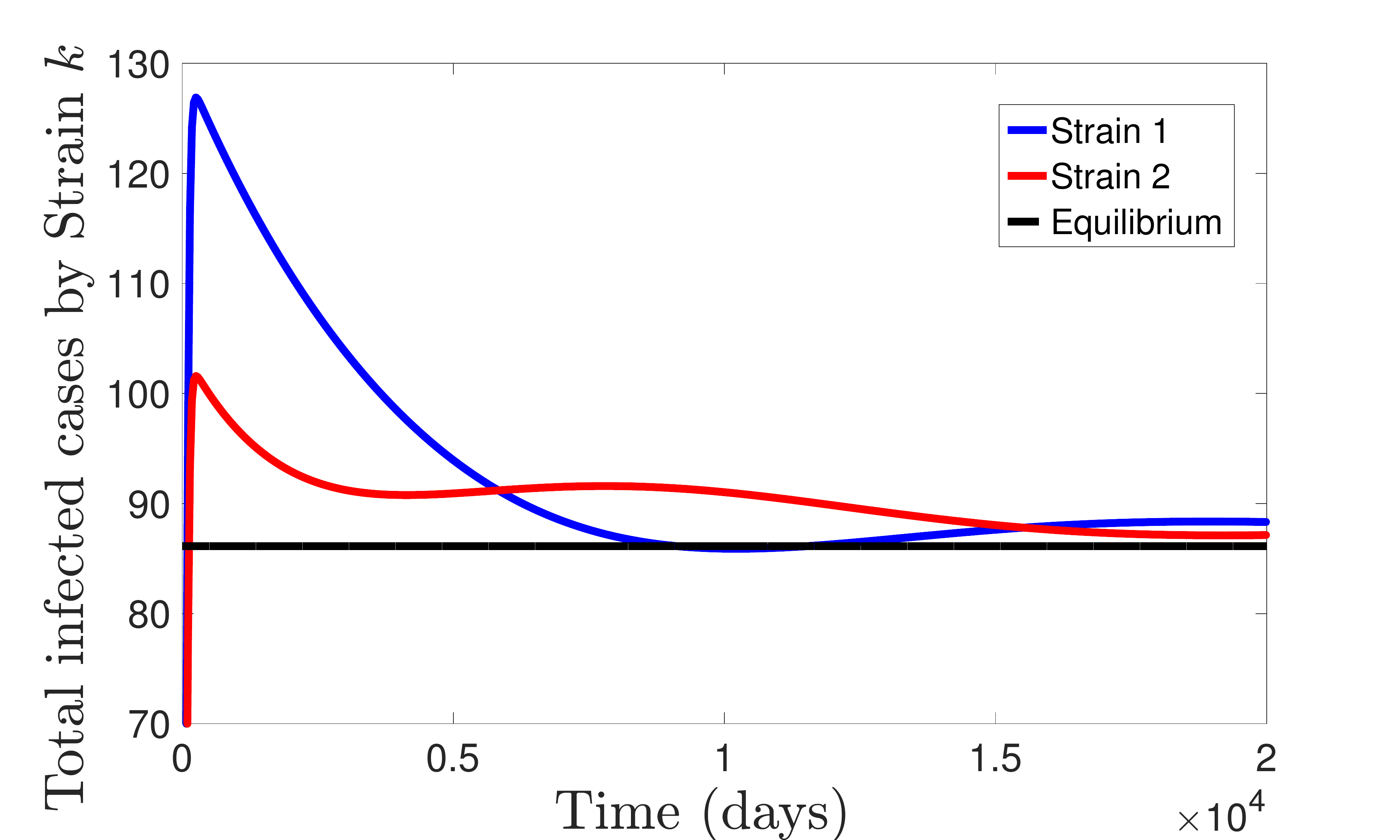}}
\subfigure[][]{\label{fig3d}\includegraphics[width=7.35cm,height=4cm]{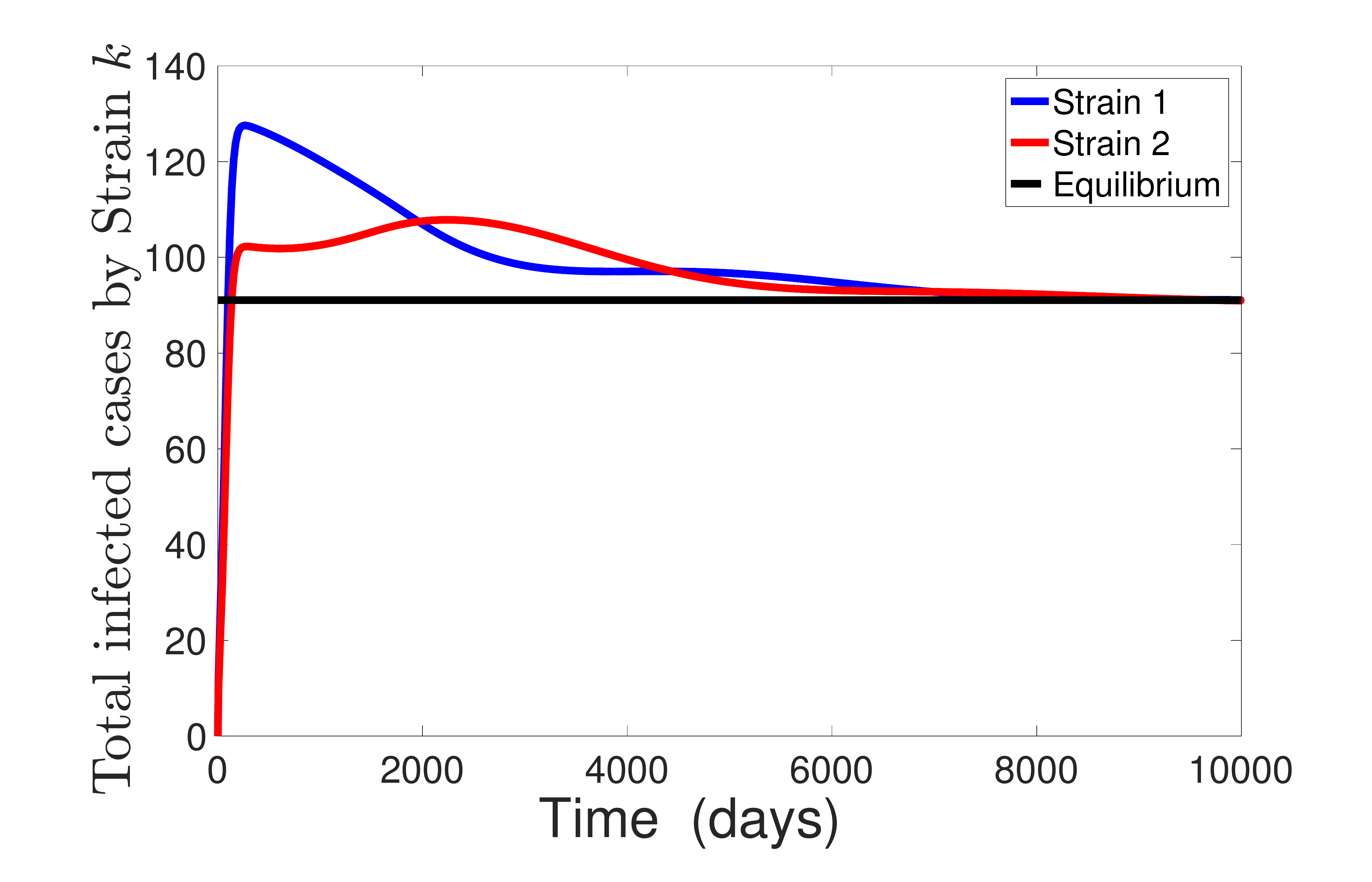}} \\
\subfigure[][]{\label{fig3e}\includegraphics[width=7.35cm,height=4cm]{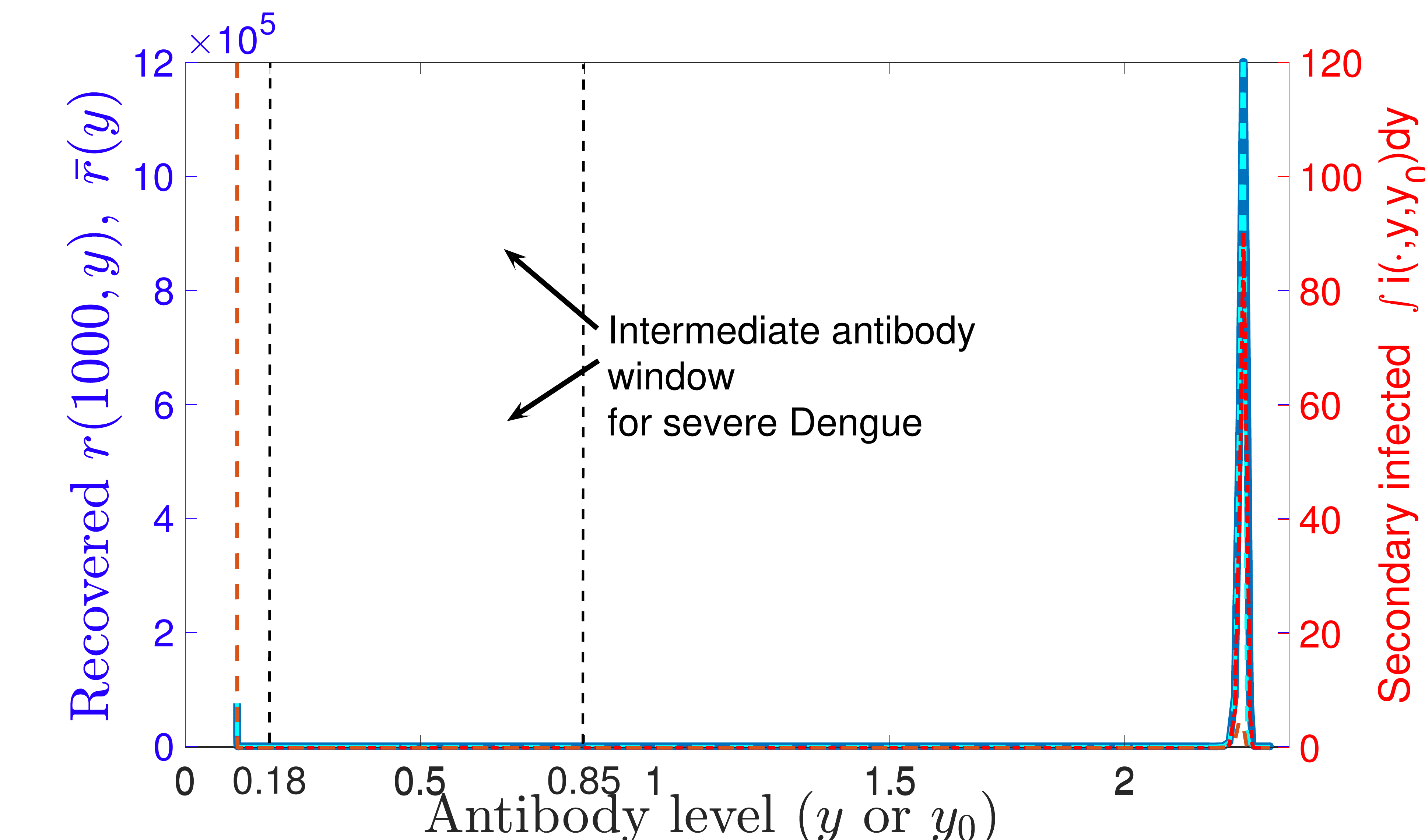}}
\subfigure[][]{\label{fig3f}\includegraphics[width=7.35cm,height=4cm]{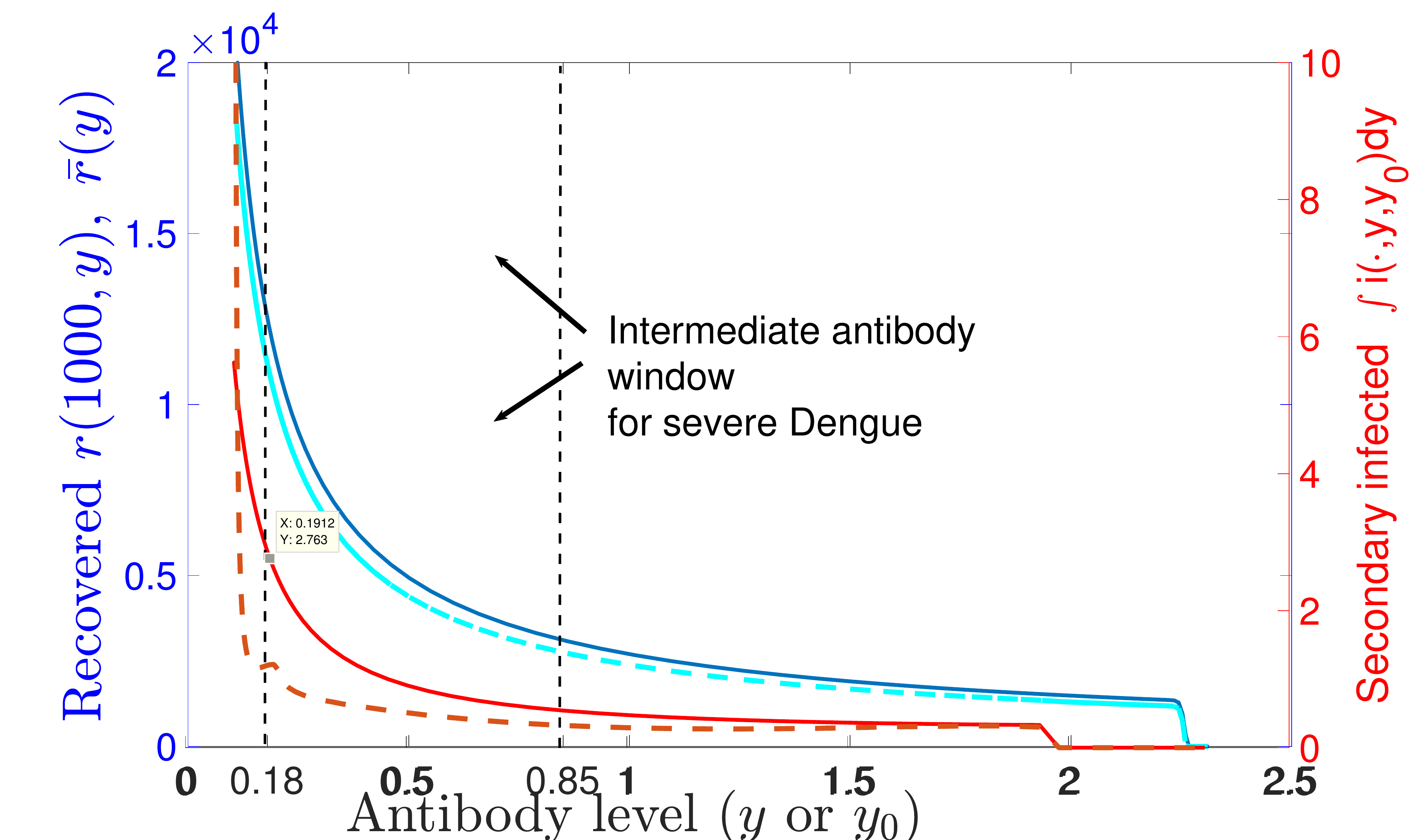}} 
\caption{Numerical Simulations for the case of no waning (a),(c),(e), and for case of waning and temporary cross-immunity (b),(d),(f).  Here, we utilize a point distribution for susceptible antibody level  $\Lambda(y)=\Lambda \delta(y_s), S(0,y)=S_0\delta(y_s)$, with $y_s=0.11$, along with all other within-host and linking parameters as in Figs \ref{fig1} and \ref{fig2}.  The demographic parameters are $\Lambda=100, \Lambda_v=0.02, \mu=1/(10\times 365),\mu_v=1/20$.  The vector-host transmission rate $\beta_v(y)$ is constant with $\beta_v=0.00025$ in (a,c,e) and piecewise constant given by \eqref{tempci} with $\beta_v=0.00025, y_p=2$ in (b,d,f).  The waning rate in (a,c,e) is $\omega(y)=0$ and is given by \eqref{expwane} with $\xi=0.002$, $y_c=0.02$.  Additionally, we let the viral load threshold for infected individual to experience DHF in \eqref{dhfC} to be $V_c=30$, which generates the vertical dashed lines in (e), (f) giving the antibody level window of DHF risk.} 
  \label{fig3}
  \end{figure}
In the first example, we consider the case where there is no waning, i.e. $\omega_k(y)\equiv 0$.  Although cross-reactive antibody levels are thought to decline after primary, the absence of waning may be a reasonable approximation for Dengue endemic regions.  In particular, some studies have found antibody levels to be stable because of continual exposure to Dengue providing boosting of immunity \cite{katzelnick2016neutralizing}.  The explicit inclusion of boosting through exposure to virus as in \cite{barbarossa2015immuno} would significantly complicate the model and is beyond the scope of the current paper.  Here, we assume constant vector-host transmission rate ($\beta_v(y)\equiv \beta_v$) and susceptible point distribution, i.e. $\Lambda(y)=\Lambda \delta(y_s), s(0,y)=S_0\delta(y_s)$.  Observe Figures \ref{fig3a}, \ref{fig3c}, \ref{fig3e} displaying simulations of time-dependent solutions of infected vectors and hosts, and the final time distribution of recovered and secondary infected individuals with respect to antibody level $y$ and $y_0$, respectively.  Calculations of equilibria are also displayed, and note the long time for convergence, along with some numerical error.  In this case without waning, the vast majority of recovered individuals have a large cross-reactive antibody level, whereby secondary exposure leads to mild infection.


\subsection{Waning and temporary cross-immunity}
For examples in this and the next section, we consider particular forms of waning, $\omega_k(y)$ and antibody dependent vector-host transmission rate $\beta_v^k(y)$, motivated by epidemiological observations. 
We utilize the form of waning given $\omega_k(y)=\xi (y-y_c)$, which is consistent with a lower bound $y_c$ and exponential decline of antibodies as formulated in equation \eqref{expwane}.
 In this case, inserting \eqref{expwane} into \eqref{ssEq}, we find the (strain-$k$) recovered equilibrium component
\begin{align}
 \bar r_k(y)&=\frac{(y-y_c)^{\frac{\mu}{\xi}-1}}{\xi}\int\limits_y^{\infty}(a-y_c)^{-\frac{\mu}{\xi}}\int\limits_{y_c}^{\infty} \gamma_k(a,y_0) \bar i_k(a,y_0) dy_0  da, \notag \\
 \text{with}  \ \ \int\limits_{y_c}^{\infty} \bar r_k(y)&= \frac{1}{\mu} \int\limits_{y_c}^{\infty} \int\limits_{0}^{\infty}\gamma_k(y,y_0)\bar i_k(y,y_0)dy_0 dy.  \label{recwaneT}
\end{align}
Note the the total amount of recovered individuals obtained by the integration above \eqref{recwaneT} is precisely the number at the equilibrium with no waning \eqref{invHnoW}.   In the instance of recovery occurring at a constant level of antibody $y_*$ \eqref{recov3}, we obtain
\begin{align}
 \bar r_k(y)=\frac{(y-y_c)^{\frac{\mu}{\xi}-1}}{\xi}(y_*-y_c)^{-\frac{\mu}{\xi}}\int\limits_{0}^{\infty} \bar i_k(y_*,y_0) dy_0  .
 \label{recwaneT2}
\end{align}
Observe that recovered individuals at the lower bound of antibody level satisfies $\bar r_k(y_c)=0$ when $\mu>\xi$ (natural death rate greater than waning rate), and $\bar r_k(y_c)=+\infty$ when $\mu<\xi$, however the total amount of recovered individuals stays finite by \eqref{recwaneT}.   

Furthermore, as mentioned previously, primary Dengue infection induces a period of temporary cross-immunity.  The simplest way to include this feature is to assume that 
\begin{align}
\beta_v^k(y)=\beta_v^k\mathds{1}_{\left\{y<y_p\right\}}=\begin{cases}  0 & y\geq y_p \\ \beta_v^k & y<y_p \end{cases}, \label{tempci}
\end{align} 
so that the vector to host transmission rate is a piecewise function where there exists a threshold antibody level $y_p$ providing complete protection above it.  In this case, let $c_k:=r_k \mathds{1}_{\left\{y\geq y_p\right\}}$ and $s_k:=r_k\mathds{1}_{\left\{y<y_p\right\}}$, denoting density of recovered individuals with cross-immunity and susceptible to secondary infection, respectively.  From \eqref{host_epi}, we derive 
\begin{align}\label{host_epiW2}
\frac{\partial c_k(t,y)}{\partial t}- \frac{\partial (\omega_k(y) c_k(t,y))}{\partial y} &=\mathds{1}_{\left\{y\geq y_p\right\}}\int_0^\infty \gamma_k(y,y_0)i_k(t,y,y_0)dy_0-\mu r_k(t,y),   \\
\frac{\partial s_k(t,y)}{\partial t} - \frac{\partial (\omega_k(y) s_k(t,y))}{\partial y}&= \mathds{1}_{\left\{y< y_p\right\}}\int_0^\infty \gamma_k(y,y_0)i_k(t,y,y_0)dy_0-\mu s_k(t,y)- \beta_v^j(y) s_k(t,y)I_v^j(t),  \notag \\
s_k(t,y_p)&= c_k(t,y_p), \qquad \omega_k(y_c) s_k(t,y_c)= \lim_{y\rightarrow\infty}\omega_k(y)c_k(t,y)=0.
\end{align}

Previous results requiring constant vector-host transmission $\beta_v^k$, in particular Theorem \ref{invThm}, Proposition \ref{coex_prop} and explicit equilibria formulas \eqref{ssEq} and \eqref{rk_coex}, can be extended to this case of piecewise constant $\beta_v^k(y)$.  To see this first note that for integrals in the formulae for $\mathcal R_k, \mathcal R_{inv}^k$ and equilibria, the upper limit of integration $y_p$ will appear (in place of $\infty$) wherever $\beta_v^k(y)$ appears.  Also for $\mathcal R_{inv}^k$, $\bar r_k(y)$ can be replaced by $\bar s_k(y)$, which makes the arguments in proof of Theorem \ref{invThm} work for this case of piecewise constant $\beta_v^k(y)$.   Furthermore for the symmetric coexistence equilibrium \eqref{rk_coex}, the recovered equation formulas will be altered as follows:
\begin{align*}   
\bar c_k(y)&=\frac{1}{\omega_k(y)}\int\limits_y^{\infty}\exp{\left(\int\limits_a^y \frac{\mu}{\omega_k(s)}ds\right)} \int\limits_{0}^{y_p} \gamma_k(a,y_0) \bar i_k(a,y_0) dy_0 da, \quad y\geq y_p  \\
\bar s_k(y)&=\frac{1}{\omega_k(y)}\int\limits_y^{y_p}  \exp{\left(\int\limits_a^y \frac{\mu+\beta_v^j\bar I_v^j}{\omega_k(s)}ds\right)} \int\limits_0^{y_p}  \gamma_k(a,y_0) \bar i_k(a,y_0) dy_0da + \frac{\omega_k(y_p)}{\omega_k(y)}\bar c_k(y_p)\exp{\left(\int\limits_{y_p}^y \frac{\mu+\beta_v^j\bar I_v^j}{\omega_k(s)}ds\right)}, 
\end{align*}
where $y < y_p$ for the domain of $\bar s_k(y)$.  The secondary vector-host force infection at equilibrium then depends upon $\bar s_k(y)$.  The component $r_k(y)$ in boundary equilibrium remains as is in \eqref{ssEq} (with upper limit of integration $y_p$), as can be seen in the above formula when removing secondary infection ($I_v^j=0$), where $r_k=c_k+s_k$.

In numerical simulations, we first consider the case of susceptible point distribution ($\Lambda(y)=\Lambda \delta(y_s), S(0,y)=S_0\delta(y_s)$).  We compute time dependent solutions from initial conditions corresponding to outbreak initiation, which are shown in Figures \ref{fig3b}, \ref{fig3d}, \ref{fig3f}.  Observe how waning and temporary cross-immunity shape the distribution of recovered and secondary infected individuals with respect to $y$ in this example, as opposed to the previous case with constant vector-host transmission and no waning.  In particular, the waning allows for secondary infected cases in the window of antibody level causing DHF, resulting in around $2$ DHF cases per time unit instead of zero in previous case without waning.  Note that since there is only one susceptible antibody level, the variable step size partition of antibody level $y$ (equivalent to transforming $y$ to time-since-infection $\tau$ with fixed step size in $\tau$) is advantageous for reducing error, as described in Section \ref{Numerical}.  In contrast for the next simulation where $\Lambda(y)$ is a distribution, we utilize the fixed antibody level ($\Delta y$) step size partition combined with interpolation of the within-host ODE solver output on to this partition (see Section \ref{Numerical}). 

\subsection{Heterogeneity among susceptible antibody level}
\begin{figure}[t]
\subfigure[][]{\label{fig2bh_b}\includegraphics[width=7.4cm,height=4cm]{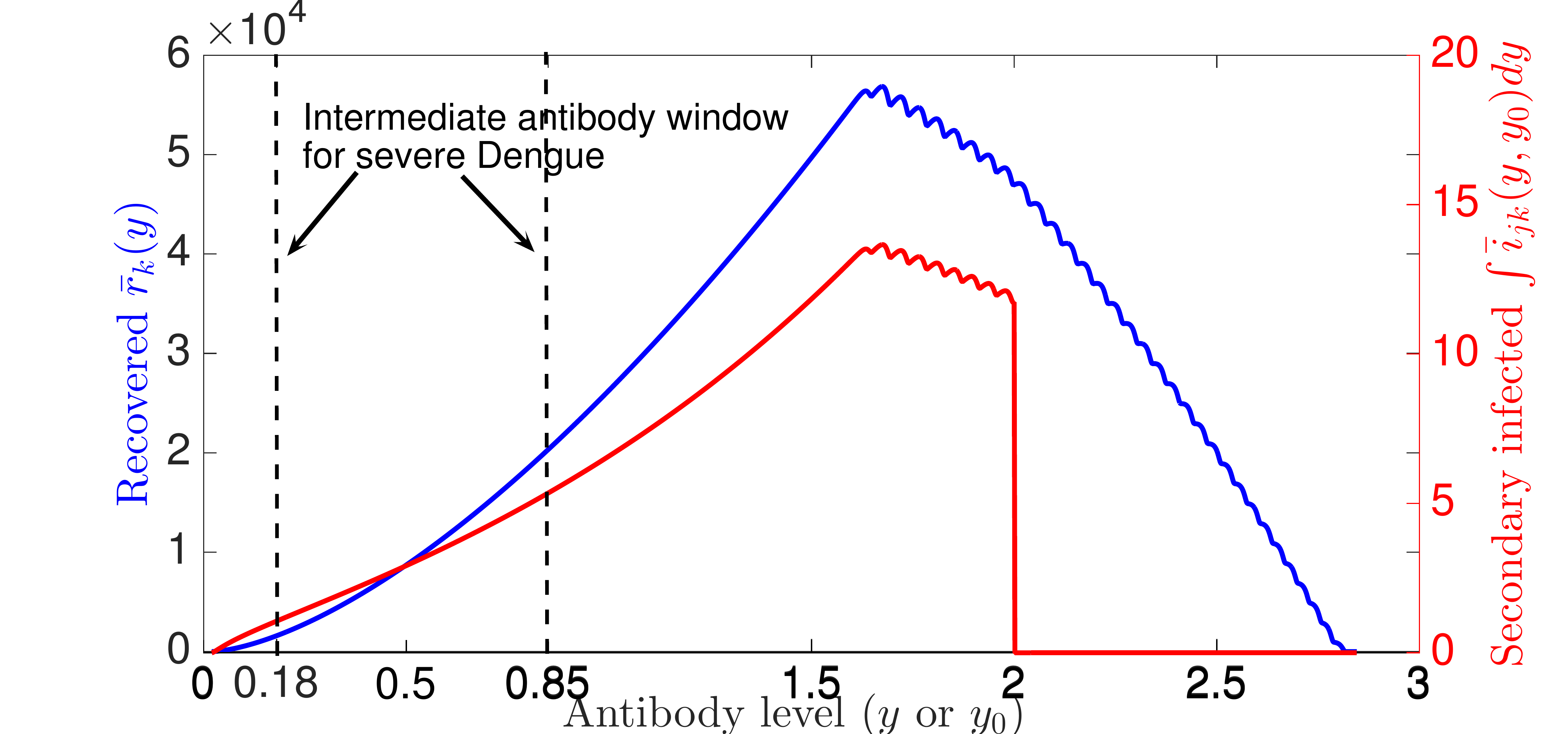}}
\subfigure[][]{\label{fig2bh_a}\includegraphics[width=7.4cm,height=4cm]{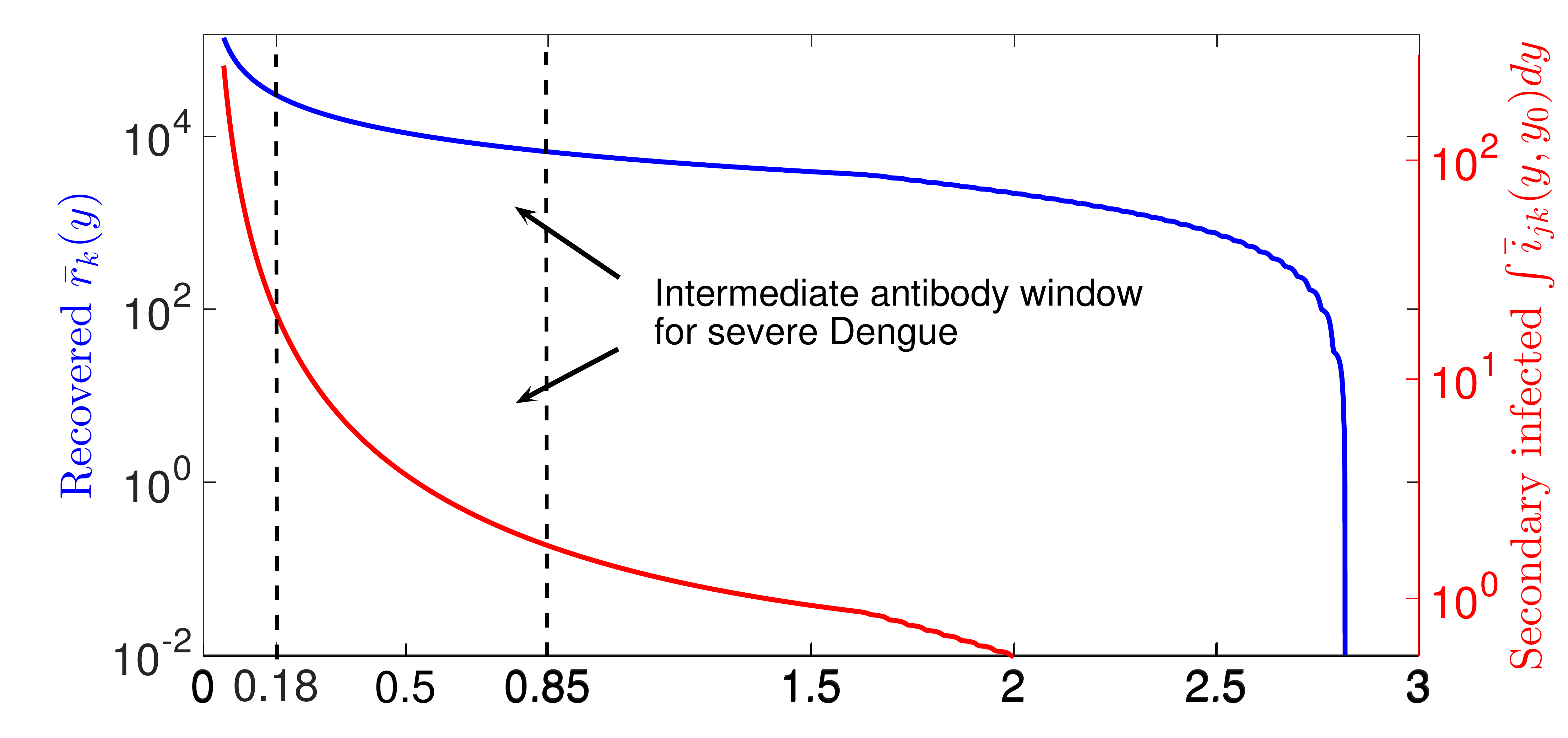}} 
\caption{\emph{Equilibria recovered host density distribution (blue curve) and total secondary infected cases (red curve) with respect to initial antibody size $y_0$ for distinct cases,} when waning rate is constant $\xi$, the transmission rate $\beta_v(y)$  is constant (describing temp. cross immunity), $\mu>\xi$, recovery rate is chosen as: $\gamma(y,y_0)=e^{-\rho x(\tau)} \mathds{1}_{\left\{f_k(\tau)<0\right\}}$, and susceptible antibody distribution $\bar s(y_0)$ is truncated normal. (a) Death rate larger than waning rate ($\mu>\xi$) leads to recovered distribution $\bar r(y)$ to go to zero as $y\searrow y_c$, (b) On the other hand if $\mu<\xi$, $\bar r(y)$ goes to $\infty$ as $y\searrow y_c$.  Other parameters, besides $\xi,\gamma(\cdot),\mu$ are same as in Fig. \ref{fig4}. }
  \label{fig2bh}
  \end{figure}
 In this section, for the (initial) distribution of susceptible antibody level given by $\Lambda(y)$, we choose the symmetric truncated normal distribution with support on the interval $[y_m-\alpha,y_m+\alpha]$ (insert in Figure \ref{fig6a}).  First, in Fig. \ref{fig2bh}, we display calculations of equilibrium recovered and secondary infected distributions for two different cases of waning rates ($\xi$), $\mu>\xi$ and $\mu<\xi$ ($\mu$ is host death rate), displaying different limiting behavior as antibody level $y$ approaches lower bound $y_c$.  Note that we analytically derived this limit dichotomy in \eqref{recwaneT2}, and larger waning ($\mu<\xi$) corresponds to larger accumulation of individuals in DHF risk window before $\bar r_1(y)\rightarrow \infty$ as $y\searrow y_c$ in this case.  
 
 Next we perform simulations utilizing the finite difference and multi-scale method outlined in Section \ref{Numerical}.  The numerical solutions (shown in Figure \ref{fig4}) are computationally more expensive due to the distribution (with width $2\alpha$) of susceptible antibody level.  Comparing to the previous case of susceptible point distribution at $y_m$, the number of DHF cases increases to eventually around 8 per unit time, however this may partially be due to the altered numerical algorithm which utilizes interpolation and fixed step size $\de y$.  Observe how the current and pre-existent antibody levels in recovered and secondary infected populations, respectively, evolve with time after initial outbreak in Fig. \ref{fig4c} and \ref{fig4d}.  In particular, individuals with pre-existent susceptible antibody levels recover with a certain boosted antibody level offering temporary cross-immunity until waning spreads recovered individuals' antibody levels to intermediate levels at risk of DHF with accumulation at smaller antibody titres for the chosen parameter regime, resulting in secondary and DHF cases.

\begin{figure}[t!]
\subfigure[][]{\label{fig4a}\includegraphics[width=8cm,height=4cm]{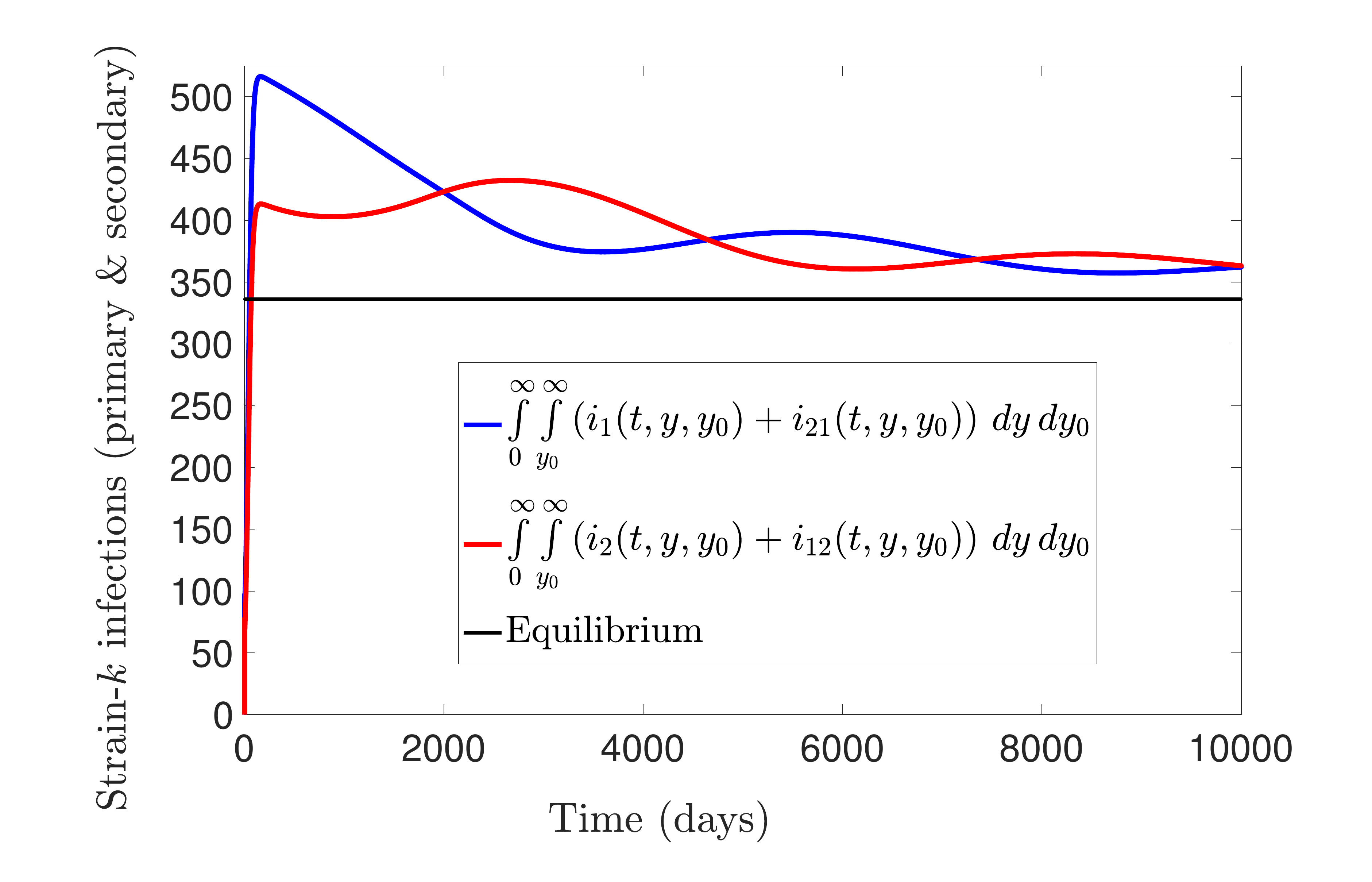}}
\subfigure[][]{\label{fig4b}\includegraphics[width=7.5cm,height=4cm]{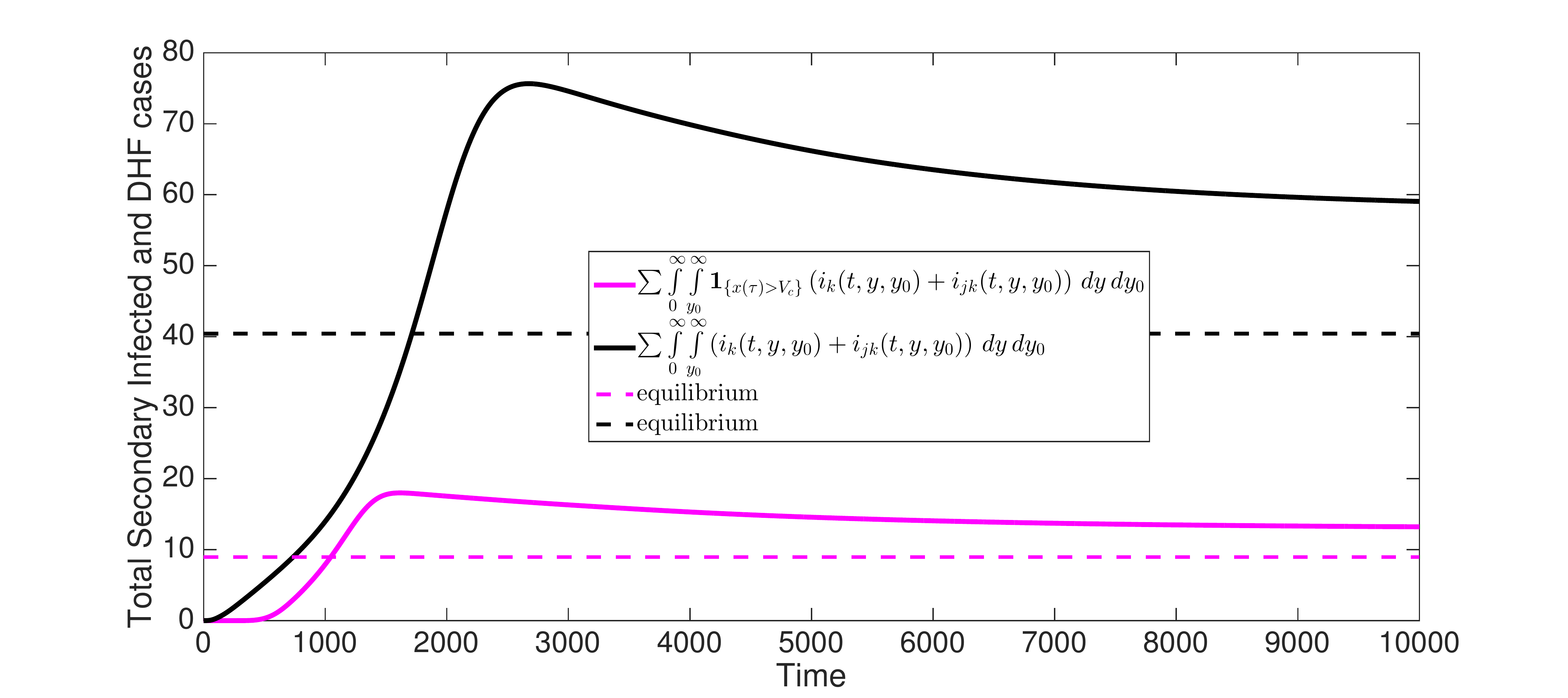}} \\
\subfigure[][]{\label{fig4c}\includegraphics[width=7.5cm,height=4cm]{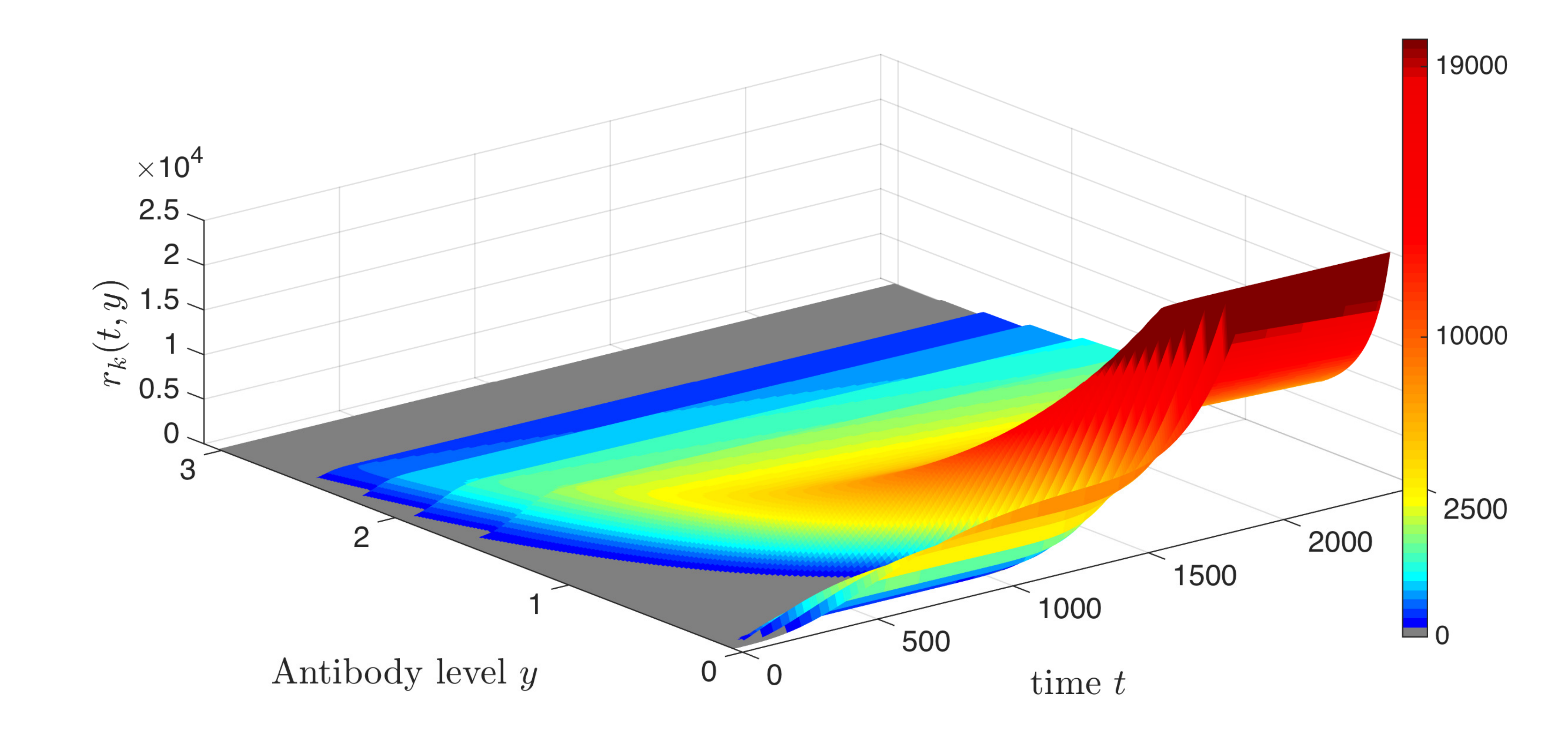}}
\subfigure[][]{\label{fig4d}\includegraphics[width=7.5cm,height=4cm]{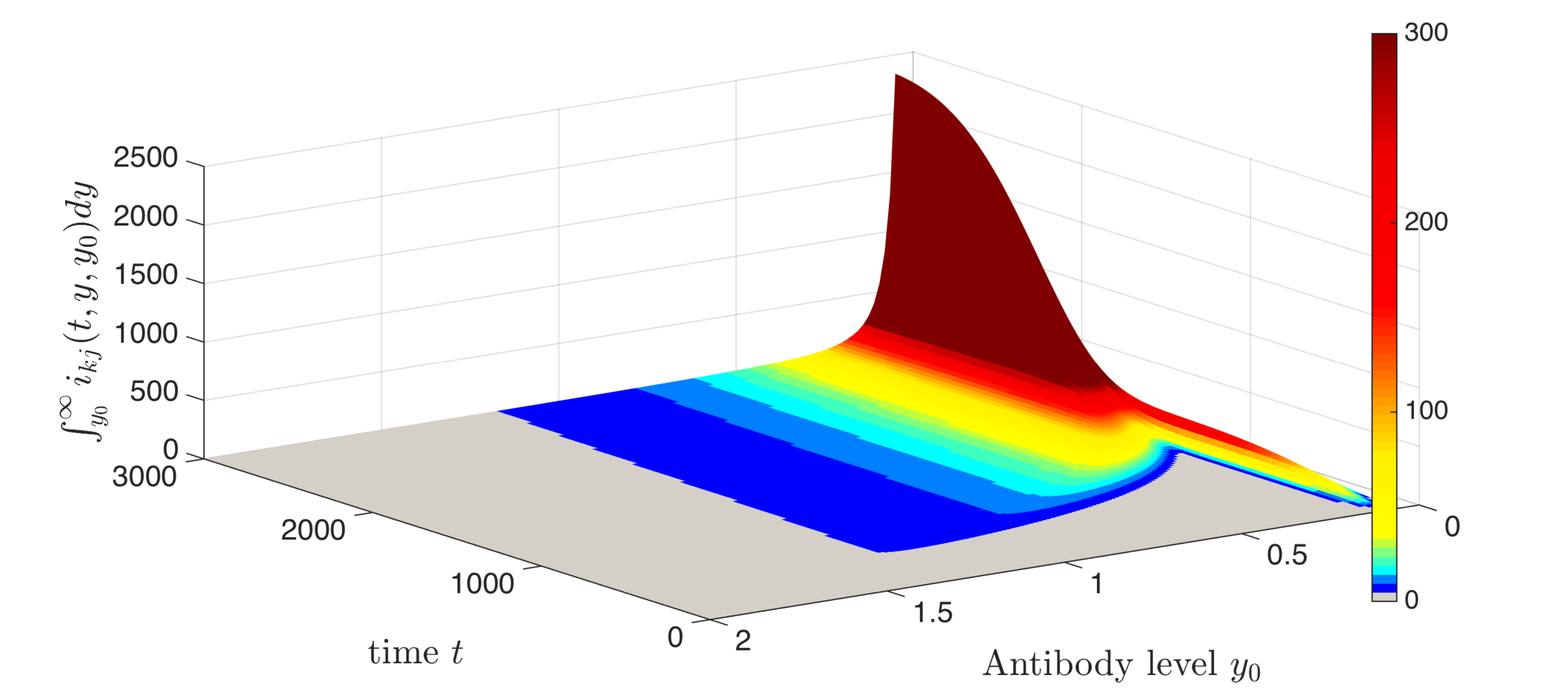}} 
\caption{Numerical simulations in the case of susceptible antibody distribution as symmetric truncated normal having support on $[y_m-\alpha,y_m+\alpha]$ (Fig. \ref{fig6a}) where $y_m=0.11$ and $\alpha=0.03$, along with initial conditions $I_v^1(0)=0.025$, $I_v^2(0)=0.02$ and other components starting at (disease-free equilibrium) $\mathcal E_0$ (outbreak scenario). (a) Total infected by strain $k$ versus time; (b) Secondary infected and DHF cases versus time; (c) Density of (strain 1) recovered individuals with respect to antibody level $y$, $r_1(t,y)$, evolving with time $t$; (d) density of secondary infected individuals with pre-existent antibody level $y_0$ (imported from population $r_1(t,y_0)$ at time of infection) as function of $t$.} 
  \label{fig4}
  \end{figure}

For the last set of numerical simulations, we consider how shifting the mean antibody level in a naive susceptible population affects disease incidence, especially DHF cases, upon introduction of both Dengue strains.  In this way, we investigate how population level antibody level can determine risk of severe dengue infection.  In particular, this can help explore the consequences of a Dengue vaccination boosting cross-reactive antibody levels.  As before, consider the distribution of susceptible antibody level (given by $\Lambda(y)$) as the symmetric truncated normal distribution with support on the interval $[y_m-\alpha,y_m+\alpha]$ (Fig. \ref{fig6a}).  Fixing $\alpha$, we vary the mean antibody level $y_m$ and calculate reproductive numbers, along with (severe) DHF cases at coexistence equilibrium $\mathcal E_c$ and initial DHF transmission level.  

Observe that $\mathcal R_0$ is unimodal with respect to $y_m$, initially increasing, peaking at a relatively small antibody level and then decaying as the population antibody level increases further (Fig. \ref{fig6a}), as opposed to invasion reproduction number $\mathcal R_{inv}$ which is decreasing (Fig. \ref{fig6c}).  The initial rise in $\mathcal R_0$ follows from the host-vector transmission, as found in Fig. \ref{fig2b}, caused by increasing within-host peak viral load from ADE.  However, decreasing infectious period also occurs with the increasing viral load leading $\mathcal R_0$ to peak at smaller antibody level than the observed absolute maximum in DHF cases (Fig. \ref{fig6a}).  Here the total DHF cases, calculated at equilibrium $\mathcal E_c$ in the immuno-epidemiological model for each $y_m$, bears some resemblance the unimodal shape of peak viral load on the within-host scale (Fig. \ref{fig1c}), but there are more ``across-scale'' forces to decipher.  In particular, first there is a local maximum caused by an overall peak in cases at maximal $\mathcal R_0$.  Next, population antibody levels rise to the intermediate window where ADE causes severe primary infection leading to a second (much larger) peak in DHF cases. However the decreasing $\mathcal R_0$ applies an opposing dampening force on total equilibrium DHF cases, and thus the absolute maximum in DHF cases occurs around $y_m=0.2$, whereas the maximum peak viral load on the within-host scale (Fig. \ref{fig1c}) occurs around $y_0=0.5$.  The sharp rise in DHF cases as a function of $y_m$ to the second maximum is caused by primary infections, as the DHF cases from secondary infections (Fig. \ref{fig6e}) is relatively small and peaks around $y_m=0.1$ for the chosen parameters.  

 Even though DHF cases at $\mathcal E_c$ decreases to zero at around $y_m=0.52$ when $\mathcal R_0$ becomes less than one, the rate of initial DHF cases per infected vector in a susceptible population is large and at maximum value for all $y_m\in [\tilde y_{\ell}+\alpha,\tilde y_u-\alpha]$ (when $\Lambda(y)$ has support in the window of antibody level causing DHF), which is centered around $y_m=0.5$, as shown in Fig. \ref{fig6f}.  Here we measure this rate of initial DHF cases by $\frac{\beta_v}{\mu_v}\int_{\tilde y_l}^{\tilde y_u} \Lambda(y)dy$.  Thus, there can be a conflict between reduction or even elimination/prevention of dengue cases by a cross-reactive antibody boosting vaccine administered to a naive population and a (possibly temporary) rise in DHF cases.  Since DHF cases can cause death and should be avoided, an effective Dengue vaccine would need to boost antibody levels of vaccinated individuals to a sufficiently large level.  This can be achieved by increased vaccine efficacy or only vaccinating individuals with large enough pre-existent antibody levels (from prior exposure to Dengue) that can be boosted past the intermediate window of DHF risk.

        \begin{figure}[t!]
    \centering  \subfigure[]{\label{fig6a}\includegraphics[width=10cm,height=5.2cm]{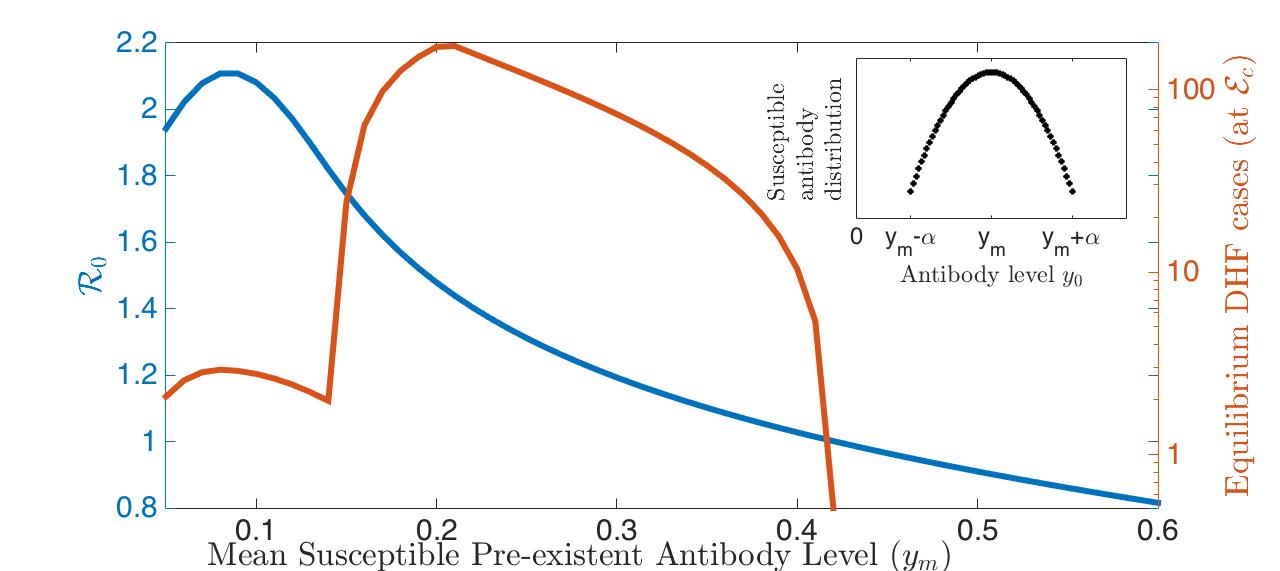}} \\
\subfigure[][]{\label{fig6c}\includegraphics[width=4.7cm,height=4cm]{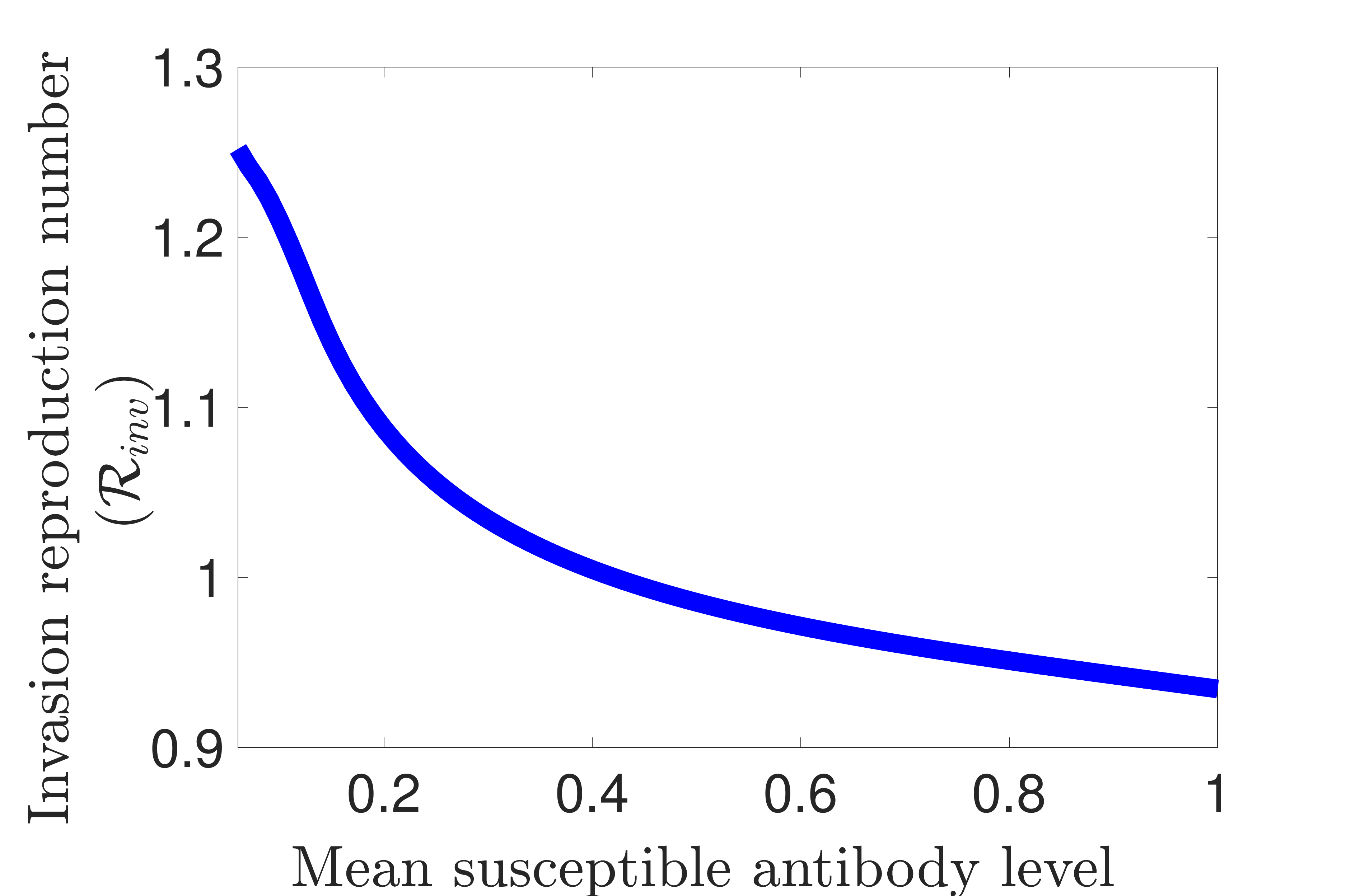}}
\subfigure[][]{\label{fig6e}\includegraphics[width=4.7cm,height=4cm]{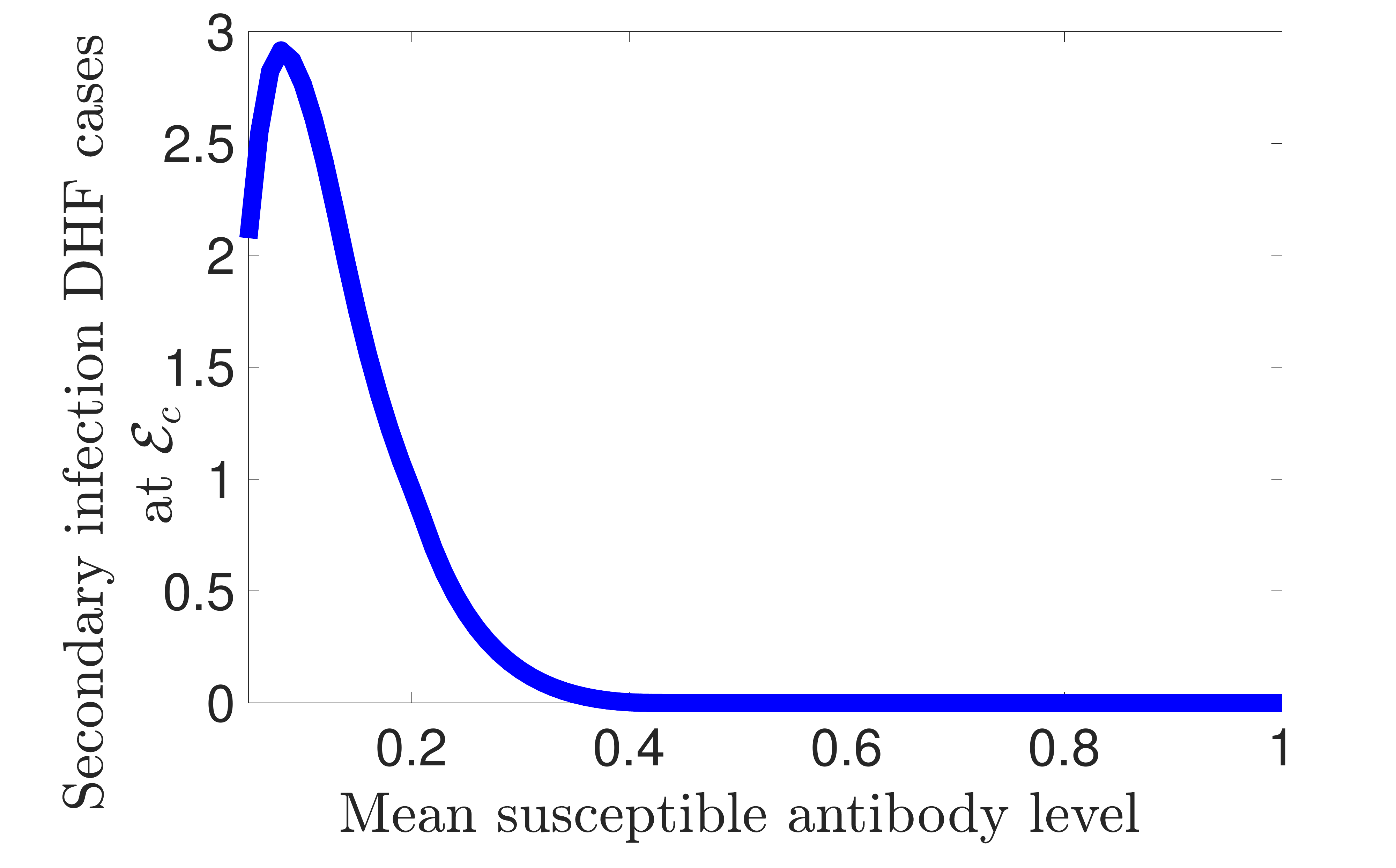}}
\subfigure[][]{\label{fig6f}\includegraphics[width=4.7cm,height=4cm]{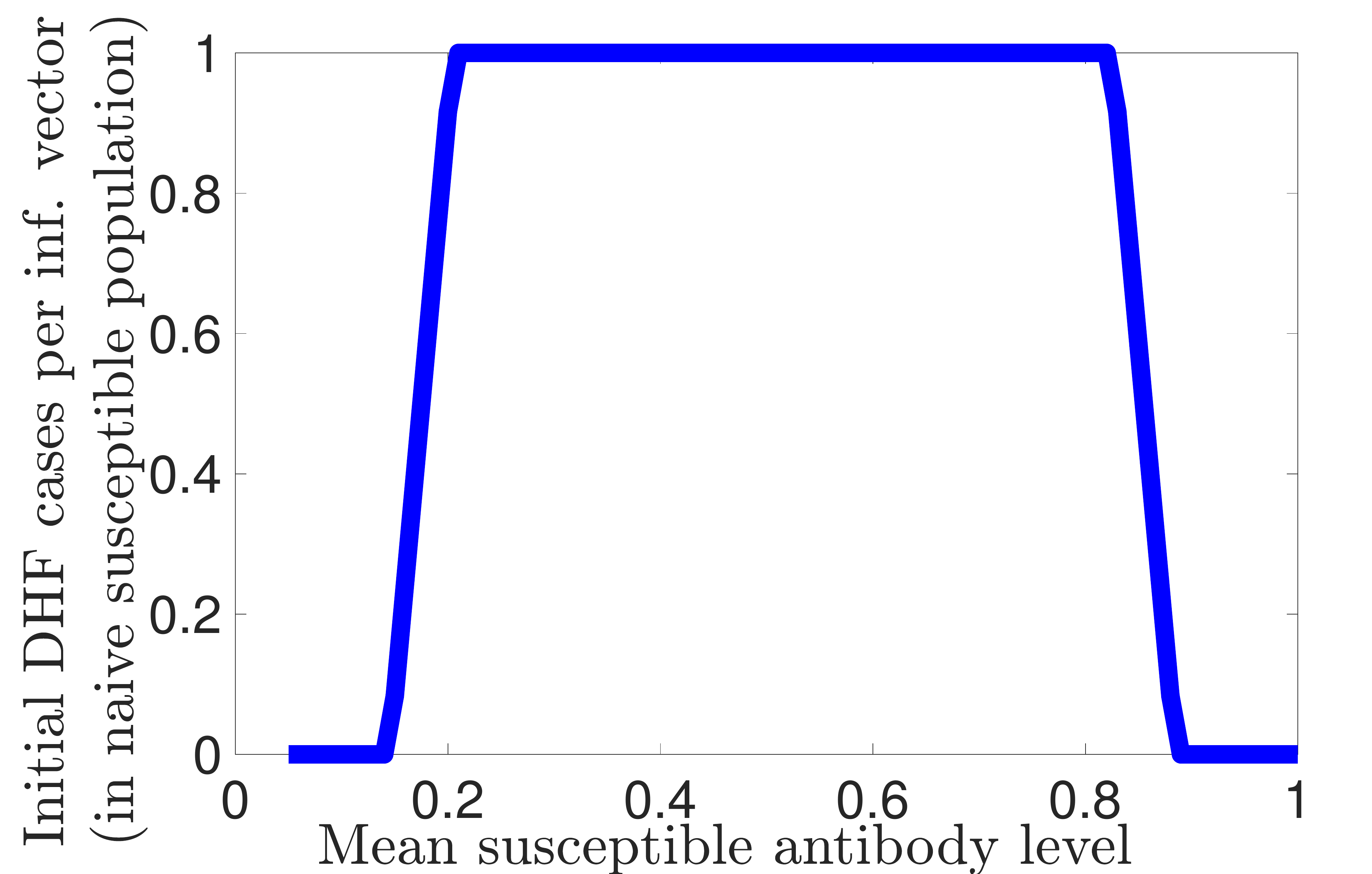}}
\caption{ \emph{ Epidemiological quantities ($\mathcal R_0, \mathcal R_{inv}$) and disease incidence (DHF) versus  immune antibody ($y$) distribution among susceptible population with mean antibody level, $y_m$.} Susceptible (naive individuals) recruitment rate, $\Lambda(y)$, is chosen as symmetric truncated normal distribution on $[y_m-\alpha,y_m+\alpha]$ with mean $y_m$.  In the simulations here, we fix $\alpha=0.03$ and vary $y_m$, in order to compute  (a)  basic reproduction number $\mathcal R_0$ (blue solid line), and DHF cases at coexistence equilibrium $\mathcal E_c$ (orange solid line), (b) invasion reproduction number $\mathcal R_{inv}$, (c) secondary infected DHF cases at $\mathcal E_c$, (d) DHF cases per infected vector in a (naive) susceptible population.}
\label{fig6}
  \end{figure}
\FloatBarrier

\section{Discussion}

In this paper, we develop an immuno-epidemiological model of Dengue tracking dynamic host cross-reactive antibody levels through infection by multiple strains and waning, which affect the overall infection trajectory and severity across the within-host and between-host scales.  The model recapitulates: 
\emph{(i) how intermediate levels of preexistent cross-reactive antibodies enhance infection within a host, and 
(ii) how to scale up to distributions of antibody levels among epidemiological classes in the host population to determine risk of DHF prevalence. }
The system is novel in its multi-scale connection of host immunity and infectious disease dynamics, and provides a unified model of Dengue with ADE phenomenon.

Our mathematical model consists of two linked systems, namely immunological (within-host) and epidemiological (between-host).  First, we formulate ordinary differential equations describing dynamics of virus and immune response for primary or secondary Dengue infection.  The model is motivated by experimental and epidemiological evidence that certain levels of pre-existent cross-reactive antibodies cause the more severe DHF infection.  Indeed, we show that our inclusion of distinct binding and neutralizing kinetics with ADE phenomenon in cross-reactive IgG (memory antibodies), along with neutralizing specific IgG, induce severe infection solely by varying pre-existent cross-reactive IgG in an intermediate window of concentration.  

Next, we connect the within-host dynamics to population scale through a hybrid ODE-PDE vector-host system structured by host antibody level.  Several features of the linking and model structure are distinguished from prior work:
\begin{itemize}
\item The epidemic model is size-structured with respect to a within-host immune variable as opposed to the infection age utilized in most immuno-epidemiological studies.  The one-to-one relationship between cross-reactive memory antibodies $(y(\tau))$ and infection age $\tau$ allows us to formulate epidemiological parameters as functions of within-host variables, similar to \cite{gandolfi2015epidemic}.
\item Immune status, in particular antibody level, is traced through multiple infections by distinct strains and waning during recovery stage, with across-scale feedback on within-host and between-host dynamics. 
\item Heterogeneity in immune response among susceptible individuals given by preexistent antibody level distribution.  Different from \cite{pugliese2011role}, our structuring variable for susceptible population is dynamic through epidemiological stages described above.  
\end{itemize}
The interplay across scales between host immunity and infection severity exhibited in Dengue induced by ADE motivate the above characteristics of our modeling framework.

On the epidemiological scale, in addition to boundedness of solutions, certain threshold dynamics are determined by strain (basic) reproduction number $\mathcal R_k$ and invasion reproduction number $\mathcal R_{inv}^j$, and these quantities depend upon within-host variables.  By linearizing the system around disease-free and single-strain equilibria, we show that local stability is sharply determined by $\mathcal R_k$ or $\mathcal R_{inv}^j$. Furthermore, we derive the formula for a unique coexistence equilibrium $\mathcal E_c$ in the case of no waning and a transcendental equation for existence of $\mathcal E_c$ when waning occurs, under the assumption of symmetric strains.  Lastly, for certain choices of vector-host transmission $\beta_v(y)$ and waning rates $\rho(y)$ representing temporary cross-immunity and exponential decay of antibodies, we obtain explicit formula for equilibria.  Further results are difficult given the model complexity; for instance there may be backward bifurcation preventing global stability of disease-free equilibrium.  Future work will explore stronger analytical results such as uniform persistence of strain $j$ when $\mathcal R_{inv}^j>1$.

Moreover, we construct a multi-scale numerical scheme in order to verify the equilibria calculations and to determine DHF risk landscape in different scenarios.  To accomplish this, we develop a finite difference algorithm combined with numerical solutions of ODE (Runge-Kutta method) and numerical integration of antibody dependent functions in order to simulate the model on both scales.  The complexity of the intertwined systems challenges efficient numerical computation.  Indeed, the combination of heterogeneity in susceptible antibody distribution, and tracking antibody levels through recovery, waning and re-infection, leads us to utilize interpolation of within-host ODE numerical solutions nested in the discretized size-structured epidemic model.  The resulting error magnifies with increasing antibody step size $\Delta y$, causing an unavoidable tradeoff between speed and accuracy.  It may be possible to efficiently transform the problem to infection-age structured system (with sufficiently small number of meshpoints), as we did for the homogenous susceptible antibody distribution (point distribution).  Future work will explore this idea, along with employing higher order finite difference methods to increase accuracy of numerical solutions.  

The numerical simulations suggest that waning after primary infection can lead to a buildup of individuals with antibody levels in risk window for DHF upon secondary infection with a distinct strain.  Different from prior Dengue epidemic modeling studies with ADE, DHF is not simply caused by secondary infection, rather it depends upon preexistent cross-reactive antibody titre, as shown in biological literature.   The model allows us to calculate the exact distribution of host population under risk for severe infection, along with DHF prevalence.  Indeed, by varying the mean of susceptible antibody level distribution ($y_m$), we observe that the reproduction numbers are decreasing in precisely the window of $y_m$ where DHF prevalence is rising.  This suggests there can be conflict between reduction or even elimination of dengue cases by cross-reactive antibody boosting vaccine administered to naive population and a rise in DHF cases. An effective Dengue vaccine would need to boost antibody levels to sufficiently large level to avoid increased DHF incidence which can lead to fatalities.

Future work can build upon the epidemiological implications by explicitly incorporating vaccination into the model to inform Dengue vaccine policy.  The controversy surrounding apparent rise in DHF among certain vaccinated individuals after a vaccination campaign \cite{aguiar2016risks,yang2017reply} and questions into how to safely vaccinate Dengue-at-risk populations, motivates the need for predictive modeling frameworks including effects of targeted vaccination on population antibody levels.   Such a model would also require robust parameter estimation and identifiability analysis extending prior work in multi-scale data fitting \cite{tuncer2016structural}.  Ideally, the model may suggest whom to target for vaccination dependent on cross-reactive antibody level measured by blood samples, in order to prevent DHF incidence and ultimately eradicate the disease.

\section*{Acknowledgement}
The authors thank two anonymous reviewers for their helpful comments and feedback on the manuscript. This work was supported by a grant from the Simons Foundation/SFARI($638193$, HG). CJB is partially supported by a U.S. National Science Foundation grant (DMS-1815095). 

\FloatBarrier

\section*{Appendix: Numerical Convergence Rates} \label{appendix}
In this section, we provide tables showing computed rates and order of convergence for numerical experiments of the finite difference and multi-scale simulation procedure described in Section \ref{Numerical}.  For the numerical tests, we calculate the error in norm between computed solutions of the $t,y$ stepping method at certain step sizes $\de t,h=\de y$ and reference solutions at some final time $t=T$.  We utilize three different types of reference solutions: (i) the numerically approximated equilibrium given by our derived formula \eqref{eqForm}, (ii)  solution of the numerical scheme with smallest step sizes $\de t,\tilde h=\de y$, and (iii) solution of the numerical scheme with step sizes multiplied by factor of $1/2$, $\frac{\de t}{2},\frac{h}{2}$.  For each error calculation at step size $h$, $e_h$, we form a sequence by successively decreasing step size by $1/2$, whereby we compute order of convergence by $\log_2(e_h/e_{h/2})$.  Furthermore, we consider two different scenarios: (a) we start the initial condition where infected vectors, $I_v^1,I_v^2$ are slightly perturbed from $\mathcal E_0$ (outbreak scenario) with final time $T=50 \ days$, (b) we start the initial condition at numerically calculated equilibrium with final time $T=500 \ days$.  For the former scenario (a), we do not use numerically calculated equilibrium as a reference solution since this may be far off from simulation at $t=50$.  

We compute the different orders of convergence because here are several sources of error and to test different initial condition scenarios.  Our method relies on distinct algorithms in addition to the finite difference scheme, such as Runge-Kutta method for within-host ODE (ode45 in MatLab), interpolation, integration and, in the case of numerical equilibrium formula, nonlinear root-finding.   Each routine can produce error, which can also propagate in the form of discontinuities in recovered distribution corresponding to an influx of recovery from primary infected individuals with pre-existent antibody levels at a certain mesh points from the initial susceptible antibody distribution.  In order to efficiently reduce error we utilize the trapezoidal integration when integrating with respect to initial susceptible antibody level $y_0$, but left endpoint integration for other antibody variables since there is small number of mesh points ($M_0$) for $s(\cdot, y_0)$ when compared to the range of antibody levels after infection and waning.  We do also provide one numerical test with only left-endpoint integration shown in last two tables, which gives more error than trapezoidal, but has more regular order of convergence pattern.  In addition, we include comparisons with a larger step size ($\de y=0.06$) for scenarios (a) and the last two tables, which forces a point distribution for susceptible antibody levels, creating different error structure.  Overall, from the different numerical tests, we observe convergence to certain error rates within a particular compartment and/or test scenario, ranging from orders that are sub-linear ($<1$) to larger than quadratic ($>2$).  When comparing with reference solutions computed by numerical simulation at smaller step size, the order of convergence is mostly faster than linear.

\begin{table}[h!]
\caption{Error analysis of $i^h_1(50,y,y_0)$ with initial $I_v^1(0)=I_v^2(0)=.02$, and other components starting at $\mathcal E_0$, for step size $\de y=h$ compared to reference solution $i^{\tilde h}_1$ ($\tilde h=\frac{.00375}{2}$) and $i^{h/2}_1$, respectively, in $L^1$ norm.}
\begin{tabular}{llllll}
$\de t$  & $h=\de y$     & $\left\| i^h_1-i^{\tilde h}_1 \right\|$   & order &       $\left\| i^{h}_1-i^{h/2}_1 \right\|$   & order    \\
\toprule
\\
0.02   & 0.03    & 524.6908885 &             & 91.15745068 &             \\
0.01   & 0.015   & 270.1464444 & 0.957726043 & 88.82019677 & 0.037472811 \\
0.005  & 0.0075  & 102.1797402 & 1.402632521 & 52.36309638 & 0.762337357 \\
0.0025 & 0.00375 & 27.53347556 & 1.891850535 & 27.53347556 & 0.927363678 \\
\bottomrule
\end{tabular}
\end{table}

\begin{table}[h!]
\caption{ Error analysis of $i^h_{12}(50,y,y_0)$ with initial $I_v^1(0)=I_v^2(0)=.02$, and other components starting at $\mathcal E_0$, for step size $\de y=h$ compared to reference solution $i_{12}^{\tilde h}$ ($\tilde h=\frac{.00375}{2}$) and $i^{h/2}_{12}$, respectively, in $L^1$ norm.}
\begin{tabular}{llllll}
$\de t$  & $h=\de y$     & $\left\| i^h_{12}-i^{\tilde h}_{12} \right\|$   & order &       $\left\| i^{h}_{12}-i^{h/2}_{12} \right\|$   & order    \\
\toprule
\\
 0.04      &    0.06     & 0.675368267 &             & 0.006352265 &             \\
0.02   & 0.03    & 0.258213027 & 1.38711261  & 0.157617642 & -4.633            \\
0.01   & 0.015   & 0.064940472 & 1.991372007 & 0.045005349 & 1.808260627 \\
0.005  & 0.0075  & 0.013576682 & 2.257986926 & 0.010354484 & 2.119840769 \\
0.0025 & 0.00375 & 0.002644786 & 2.359908155 & 0.002644786 & 1.969032926  \\
\bottomrule
\end{tabular}
\end{table}

\begin{table}[h!]
\caption{Error analysis of $r^h_1(50,y)$ with initial $I_v^1(0)=I_v^2(0)=.02$, and other components starting at $\mathcal E_0$, for step size $\de y=h$ compared to reference solution $r^{\tilde h}_1$ ($\tilde h=\frac{.00375}{2}$) and $r^{h/2}_1$, respectively, in $L^1$ norm.}
\begin{tabular}{llllll}
$\de t$  & $h=\de y$     & $\left\| r^h_1-r^{\tilde h}_1 \right\|$   & order &       $\left\| r^{h}_1-r^{h/2}_1 \right\|$   & order    \\
\toprule
\\
0.04   & 0.06    & 199.170148  &             & 19.71504953 &              \\
0.02   & 0.03    & 153.0385939 & 0.380105906 & 132.3873372 & -2.747395896 \\
0.01   & 0.015   & 87.31621564 & 0.809574014 & 81.85961804 & 0.693541296  \\
0.005  & 0.0075  & 37.51280945 & 1.218866289 & 35.05161783 & 1.223670901  \\
0.0025 & 0.00375 & 15.17871831 & 1.305333341 & 15.17871831 & 1.20743106  
\end{tabular}
\end{table}

\begin{table}[h!]
\caption{Error analysis of $I_v^{1,h}(50)$ with initial $I_v^1(0)=I_v^2(0)=.02$, and other components starting at $\mathcal E_0$, for step size $\de y=h$ compared to reference solution $I_v^{\tilde h}$ ($\tilde h=\frac{.00375}{2}$) and $I_v^{h/2}$, respectively, in $L^1$ norm.}
\begin{tabular}{llllll}
$\de t$  & $h=\de y$     & $\left\| I_v^h-I_v^{\tilde h} \right\|$   & order &       $\left\| I_v^{h}-I_v^{h/2} \right\|$   & order    \\
\toprule
\\
0.04   & 0.06    & 0.011545595 &              & 0.000253176 &              \\
0.02   & 0.03    & 0.011798771 & -0.031294078 & 0.002043507 & -3.012833819 \\
0.01   & 0.015   & 0.009755265 & 0.274383701  & 0.003205022 & -0.649287231 \\
0.005  & 0.0075  & 0.006550243 & 0.574632587  & 0.003355452 & -0.066172867 \\
0.0025 & 0.00375 & 0.003194791 & 1.035826728  & 0.003194791 & 0.07078528  
\end{tabular}
\end{table}

\begin{table}[h!]
\caption{Error analysis of $s^h(50,y)$ with initial $I_v^1(0)=I_v^2(0)=.02$, and other components starting at $\mathcal E_0$, for step size $\de y=h$ compared to reference solution $s^{\tilde h}$ ($\tilde h=\frac{.00375}{2}$) and $s^{h/2}$, respectively, in $L^1$ norm.}
\begin{tabular}{llllll}
$\de t$  & $h=\de y$     & $\left\| s^h-s^{\tilde h} \right\|$   & order &       $\left\| s^{h}-s^{h/2} \right\|$   & order    \\
\toprule
\\
0.04   & 0.06    & 332.6520601 &             & 41.58150751 &              \\
0.02   & 0.03    & 62.96919952 & 2.401295736 & 132.3873372 & -1.670751168 \\
0.01   & 0.015   & 49.82665921 & 0.33772848  & 81.85961804 & 0.693541296  \\
0.005  & 0.0075  & 32.64463389 & 0.610071986 & 35.05161783 & 1.223670901  \\
0.0025 & 0.00375 & 15.64718881 & 1.060942377 & 15.17871831 & 1.20743106  
\end{tabular}
\end{table}

\begin{table}[h!]
\caption{ Error analysis of $i^h_1(500,y,y_0)$ with initial condition set at numerically calculated equilibrium $\bar i_1(y,y_0)$ for step sizes $\de y=h$ compared to $\bar i_1$, reference solution $i^{\tilde h}_1$ ($\tilde h=.00375$) and $i^{h/2}_1$, respectively, in $L^1$ norm.}
\begin{tabular}{llllllll}
$\de t$  & $h=\de y$   & $\left\| i^h_1-\bar i_1 \right\|$ &  order &       $\left\| i^h_1-i^{\tilde h}_1 \right\|$   & order &       $\left\| i^{h}_1-i^{h/2}_1 \right\|$   & order    \\
\toprule
\\
0.02   & 0.03    & 324.5169207 &             & 47.41855063 &             & 36.62544876 &             \\
0.01   & 0.015   & 196.0711709 & 0.726916278 & 6.727908881 & 2.817221495 & 5.311540273 & 2.785644248 \\
0.005  & 0.0075  & 105.7195462 & 0.891135289 & 0.846383875 & 2.99077412  & 0.846383875 & 2.649746236 \\
0.0025 & $0.00375$ & 56.5106755  & 0.903646798 &             &             &             &    \\
\bottomrule        
\end{tabular}
\end{table}

\begin{table}[h!]
\caption{ Error analysis of $i^h_{12}(500,y,y_0)$ with initial condition set at numerically calculated equilibrium $\bar i_{12}(y,y_0)$ for step sizes $\de y=h$ compared to $\bar i_{12}$, reference solution $i^{\tilde h}_{12}$ ($\tilde h=.00375$) and $i^{h/2}_{12}$, respectively, in $L^1$ norm.}
\begin{tabular}{llllllll}
$\de t$  & $h=\de y$   & $\left\| i^h_{12}-\bar i_{12} \right\|$ &  order &       $\left\| i^h_{12}-i^{\tilde h}_{12} \right\|$   & order &       $\left\| i^{h}_{12}-i^{h/2}_{12} \right\|$   & order    \\
\toprule
\\
0.02   & 0.03    & 8.783517629 &             & 6.042508363 &             & 3.481283278 &             \\
0.01   & 0.015   & 4.503971602 & 0.963601096 & 1.631284321 & 1.88913931  & 1.082280827 & 1.685544319 \\
0.005  & 0.0075  & 2.240824134 & 1.007168304 & 0.32158709  & 2.342726861 & 0.32158709  & 1.7507935   \\
0.0025 & 0.00375 & 1.175949762 & 0.930202999 &             &             &             &            
\end{tabular}
\end{table}

\begin{table}[]
\caption{ Error analysis of $r^h_1(500,y)$ with initial condition set at numerically calculated equilibrium $\bar r_1(y)$ for step sizes $\de y=h$ compared to $\bar r_1$, reference solution $r^{\tilde h}_1$ ($\tilde h=.00375$) and $r^{h/2}_1$, respectively, in $L^1$ norm.}
\begin{tabular}{llllllll}
$\de t$  & $h=\de y$   & $\left\| r^h_1-\bar r_1 \right\|$ &  order &       $\left\| r^h_1-r^{\tilde h}_1 \right\|$   & order &       $\left\| r^{h}_1-r^{h/2}_1 \right\|$   & order    \\
\toprule
\\
0.02   & 0.03    & 618.0541417 &             & 775.4388502 &             & 599.9811402 &             \\
0.01   & 0.015   & 461.0633089 & 0.422768362 & 356.8297922 & 1.119776945 & 285.8162025 & 1.069829449 \\
0.005  & 0.0075  & 392.9077021 & 0.230774413 & 143.9789694 & 1.309377976 & 143.9789694 & 0.989229606 \\
0.0025 & 0.00375 & 357.6509033 & 0.135638364 &             &             &             &            
\end{tabular}
\end{table}

\begin{table}[]
\caption{ Error analysis of $I_v^{1,h}(500)$ with initial condition set at numerically calculated equilibrium $\bar I_v$ for step sizes $\de y=h$ compared to $\bar I_v$, reference solution $I_v^{\tilde h}$ ($\tilde h=.00375$) and $I^{h/2}_v$, respectively.}
\begin{tabular}{llllllll}
$\de t$  & $h=\de y$   & $\left | I^h_v-\bar I_v \right |$ &  order &       $\left | I^h_v-I^{\tilde h}_v \right |$   & order &       $\left | I^{h}_v-I^{h/2}_v \right |$   & order    \\
\toprule
\\
0.02   & 0.03    & 0.006679126 &              & 0.00106215 &             & 0.0008444   &             \\
0.01   & 0.015   & 0.007472821 & -0.161993635 & 0.00021775 & 2.286242828 & 0.000181798 & 2.215591493 \\
0.005  & 0.0075  & 0.007599984 & -0.024343416 & 3.60E-05   & 2.598522051 & 3.60E-05    & 2.338184328 \\
0.0025 & 0.00375 & 0.007596424 & 0.000675948  &            &             &             &            
\end{tabular}
\end{table}

\begin{table}[]
\caption{ Error analysis of $s^h(500,y)$ with initial condition set at numerically calculated equilibrium $\bar s(y)$ for step sizes $\de y=h$ compared to $\bar s$, reference solution $s^{\tilde h}$ ($\tilde h=.00375$) and $s^{h/2}$, respectively, in $L^1$ norm.}
\begin{tabular}{llllllll}
$\de t$  & $h=\de y$   & $\left\| s^h-\bar s \right\|$ &  order &       $\left\| s^h-s^{\tilde h} \right\|$   & order &       $\left\| s^{h}-s^{h/2} \right\|$   & order    \\
\toprule
\\
0.02   & 0.03    & 462.4651032 &              & 45.73632432 &              & 33.06512108 &              \\
0.01   & 0.015   & 518.3459247 & -0.164570714 & 78.8058773  & -0.784962803 & 31.29116995 & 0.079554584  \\
0.005  & 0.0075  & 526.9051322 & -0.023628012 & 47.51537534 & 0.729908803  & 47.51537534 & -0.602638825 \\
0.0025 & 0.00375 & 526.0055158 & 0.002465304  &             &              &             &             
\end{tabular}
\end{table}

\begin{table}[] \label{not_trap1}
\caption{Error analysis with initial condition set at numerically calculated equilibrium $\mathcal E_c$ for step sizes $\de y=h$ compared to $\mathcal E_c$, using left-end point approximation integration.}
\begin{tabular}{llllllll}

$\de t$ & $\de y$ & $s(500,y)$   &   Order          & $i_1(500,y,y_0)$   &    Order         & $i_{12}(500,y,y_0)$   &    Order               \\
\toprule
\\
0.4   & 0.06    & 660.9536    &              & 14.97353256 &              & 32.78271708 &              \\
0.2   & 0.03    & 2696.278157 & -2.028348437 & 1830.012949 & -6.933295423 & 34.9009843  & -0.090332293 \\
0.1   & 0.015   & 1728.915761 & 0.64110176   & 641.6430016 & 1.51201112   & 12.19733364 & 1.516701918  \\
0.05  & 0.0075  & 1078.164779 & 0.68128989   & 214.2830149 & 1.582253332  & 4.159619189 & 1.552042445  \\
0.025 & 0.00375 & 695.6780402 & 0.632086     & 74.63216773 & 1.521648004  & 1.541971374 & 1.431675474 
\end{tabular}
\end{table}

\begin{table}[]
 \label{not_trap2}
\caption{Error analysis with initial condition set at numerically calculated equilibrium $\mathcal E_c$ for step sizes $\de y=h$ compared to $\mathcal E_c$, using left-end point approximation integration.}
\begin{tabular}{llllllll}

$\de t$ & $\de y$ & $r_1(500,y)$   &   Order          & $I_v^1(500)$   &    Order                    \\
\toprule
\\
0.4   & 0.06    & 2522.062793 &             & 0.001742764 &              \\
0.2   & 0.03    & 2166.277168 & 0.582119138 & 0.038248294 & -4.455946151 \\
0.1   & 0.015   & 1559.573153 & 0.694509637 & 0.024707943 & 0.63042051   \\
0.05  & 0.0075  & 1167.784572 & 0.667748654 & 0.015523692 & 0.670503164  \\
0.025 & 0.00375 & 942.6623956 & 0.619407636 & 0.010085476 & 0.622192549 
\end{tabular}
\end{table}

\FloatBarrier

\bibliography{dengue_ref}
\bibliographystyle{spmpsci}

\end{document}